\documentclass[12pt]{article}
\usepackage{setspace}
\usepackage{natbib}
\usepackage{graphicx}
\usepackage{lscape}
\usepackage{pdflscape}
\usepackage{afterpage}
\usepackage{pdfpages}
\usepackage{float}
\usepackage{booktabs}
\usepackage[utf8]{inputenc}
\usepackage{indentfirst}
\usepackage{arydshln}
\usepackage{arydshln}
\usepackage{tikz}
\usepackage{lipsum}
\usepackage{pgffor}
\usepackage{dsfont}
\usepackage{amsfonts}
\usepackage{amsmath}
\usepackage{xfrac}
\usepackage{amssymb}
\usepackage{graphicx}
\usepackage{url}				
\usepackage{subfig} 

\usepackage{stackengine,graphicx}
\stackMath

\newcommand\norm[1]{\left\lVert#1\right\rVert}

\newenvironment{proof}[1][Proof]{\noindent \textbf{#1.} }{\  \rule{0.5em}{0.5em}}

\newtheorem{innercustomgeneric}{\customgenericname}
\providecommand{\customgenericname}{}
\newcommand{\newcustomtheorem}[2]{%
  \newenvironment{#1}[1]
  {%
   \renewcommand\customgenericname{#2}%
   \renewcommand\theinnercustomgeneric{##1}%
   \innercustomgeneric
  }
  {\endinnercustomgeneric}
}

\newcustomtheorem{proposition_b}{Proposition}
\newcustomtheorem{lemma_b}{Lemma}
\newcustomtheorem{assumption_b}{Assumption}
\newcustomtheorem{corollary_b}{Corollary}

\newtheorem{ass}{Assumption}[section]

\newtheorem{cor}{Corollary}[section]

\newtheorem{prop}{Proposition}[section]

\newtheorem{rem}{Remark}[section]

\newtheorem{lem}{Lemma}[section]

\definecolor{ao}{rgb}{0.0, 0.5, 0.0}

\usepackage{hyperref} 
\hypersetup{
	colorlinks=true, breaklinks=true, bookmarks=true,bookmarksnumbered,
	urlcolor=red, linkcolor=red, citecolor=blue, 
	pdftitle={}, 
	pdfauthor={\textcopyright}, 
	pdfsubject={}, 
	pdfkeywords={}, 
	pdfcreator={pdfLaTeX}, 
	pdfproducer={LaTeX with hyperref and ClassicThesis} 
}

\usepackage[top=1in, bottom=1in, left=1in, right=1in]{geometry}

\begin{document}

	\def\spacingset#1{\renewcommand{\baselinestretch}%
		{#1}\small\normalsize} \spacingset{1}


\title{ 
\huge On the Properties of the Synthetic Control Estimator with Many Periods and Many Controls\footnote{I would like to thank Alberto Abadie, Matias Cattaneo, Aureo de Paula, Marcelo Fernandes,  Antonio Galvao,  Victor Filipe Martins-da-Rocha,  Ricardo Masini,  Eduardo Mendes,  Whitney Newey, Vitor Possebom, and Pedro Sant'Anna for comments and suggestions. I also thank seminar and conference participants at the Chamberlain seminar, Warwick University, University of Cambridge, Queen Mary University, and of the Simonsen Lecture of the 2019 LACEA/LAMES conference.   Luis Alvarez and Lucas Barros provided truly exceptional research assistance, being fundamental in both the derivation of the theoretical results and in the simulations. I also gratefully acknowledge financial support from FAPESP. }}

\author{
Bruno Ferman\footnote{email: bruno.ferman@fgv.br; address: Sao Paulo School of Economics, FGV, Rua Itapeva no. 474, Sao Paulo - Brazil, 01332-000; telephone number: +55 11 3799-3350}  \\
\\
Sao Paulo School of Economics - FGV \\
\\
\footnotesize
First Draft: June 14th, 2019 \\ \footnotesize This Draft: May 25th, 2020
\footnotesize
}

\date{}
\maketitle

	\newsavebox{\tablebox} \newlength{\tableboxwidth}
	

	\begin{center}

\

\textbf{Abstract}

\end{center}

We consider the asymptotic properties of the Synthetic Control (SC) estimator when both the number of pre-treatment periods and control units are large. If potential outcomes follow a linear factor model, we provide conditions under which the factor loadings of the SC unit converge in probability to the factor loadings of the treated unit. This happens when there are weights diluted among an increasing number of control units such that a weighted average of the factor loadings of the control units asymptotically reconstructs the factor loadings of the treated unit. In this case, the SC estimator is asymptotically unbiased even when treatment assignment is correlated with time-varying unobservables. This result can be valid even when the number of control units is larger than the number of pre-treatment periods.

\

	\noindent%
	{\it Keywords: counterfactual analysis, comparative studies, synthetic control, policy evaluation, panel data, factor models. } 
	
	\
	
	\noindent%
	{\it JEL Codes: C13; C21; C23
} 
	
		\vfill

	\newpage
	\spacingset{1.45} 
	

\newpage
\spacingset{1.5} 
\section{Introduction}
\label{sec:intro}

The Synthetic Control (SC) estimator, proposed in a series of influential papers by \cite{Abadie2003}, \cite{Abadie2010}, and  \cite{Abadie2015}, quickly became one of the most popular methods for policy evaluation (e.g., \cite{Athey_Imbens}). An important advantage of the SC method is that it can potentially allow for correlation between treatment assignment and time-varying unobserved covariates.  Assuming a perfect pre-treatment fit condition, \cite{Abadie2010} show that the bias of the SC estimator is bounded by a function that asymptotes to zero when the number of pre-treatment periods increases and the number of control units is fixed.\footnote{We refer to perfect pre-treatment fit as the existence of weights such that a weighted average of the control units equal to outcome of the treated unit for all pre-treatment periods. \cite{FB} and \cite{Powell} also consider the properties of the SC and related estimators under a perfect pre-treatment fit condition. } However, when the perfect pre-treatment fit condition is relaxed and the number of control units is fixed, \cite{FP_SC} show that the SC estimator is generally biased when there are time-varying unobserved confounders.  In settings where the number of control units and pre-treatment periods are both large, there  are  alternative methods, many of them based on the original SC estimator,  that allow for selection on time-varying unobservables.\footnote{See, for example,  \cite{SDID}, \cite{Matrix}, \cite{Gobillon},  \cite{Bai}, and \cite{XU}.}  However, the properties of the original SC estimator --- which remains commonly used in empirical applications ---  when both the number of pre-treatment periods and control units go to infinity received less attention.  
 
In this paper, we consider the asymptotic properties of the  SC estimator when both the number of pre-treatment periods and the number of control units increase, and the pre-treatment fit is imperfect. We consider a linear factor model structure for potential outcomes, and  derive conditions under which, in this setting, the factor loadings of the SC unit --- which is a weighted average of the factor loadings  of the control units --- converge in probability to the factor loadings of the treated unit. We show that this will be the case when, as the number of control units goes to infinity, there are weights diluted among an increasing number of control units that (asymptotically) recover the factor loadings of the treated unit.   This holds even in settings in which  the number of control units is at the same magnitude or even larger than the  number of pre-treatment periods, which is common  in SC applications (e.g., \cite{Doudchenko}).

The intuition is the following. \cite{FP_SC} show that, in a setting with a fixed number of control units and imperfect pre-treatment fit, the SC weights converge to weights that, in general, do \emph{not} converge to weights that recover the factor loadings of the treated unit when the number of pre-treatment periods increases. The reason is that the SC weights converge to weights that attempt to, at the same time, recover the factor loadings of the treated unit \emph{and} minimize the variance of a weighted average of the idiosyncratic shocks of the control units. However, when the number of control units increases, the importance of the variance of this  weighted average of the idiosyncratic shocks vanishes \emph{if}  it is possible to recover the factor loadings of the treated unit with weights that are diluted among an increasing number of control units. In this case, the SC weights  converge to weights that recover the factor loadings of the treated unit. As a consequence, the SC estimator is asymptotically unbiased even when treatment assignment is correlated with time-varying unobservables.\footnote{\cite{SDID} show that SC weights with an $L_2$ penalization, to ensure that in large samples there will be many units with positive weights, consistently estimates a low-rank matrix structure when the penalization constraint becomes tighter. We do not require an $L_2$ penalization in the estimation of the weights, so our results are valid for the original SC weights, which does not use such penalization. }   

While increasing the number of control units increases the number of parameters to be estimated, as shown by \cite{Chernozhukov}, the non-negativity and adding-up constraints work as a regularization method. This is why  it is possible to consistently estimate the factor loadings of the treated unit even when the number of control units grows at a faster rate than the number of pre-treatment periods. We provide conditions for the consistency of the factor loadings of the SC unit even when the linear factor model structure induces a non-zero correlation between the outcome of the control units and the error in a linear model that relates the outcomes of the treated and the control units using balancing weights. We refer to ``balancing weights'' as weights such that a linear combination of the factor loadings of the control units recover the factor loadings of the treated unit.

We also show that such regularization implies that, asymptotically, there is no over-fitting. Asymptotically, the SC unit absorbs only the common factor structure, so that the pre-treatment fit will \emph{not} be perfect due to the idiosyncratic shocks, even when the number of control units increases. This highlights that the asymptotic unbiasedness of the SC estimator we derive does \emph{not} come from improvements in the pre-treatment fit due to an increased number of control units. Rather, it comes from the fact that, under the conditions we consider for the factor loadings, it is possible to construct balancing weights such that the variance of a  linear combination of the idiosyncratic shocks of the control units using those weights converges to zero. 

Overall, these results  extend the set of possible applications in which the SC estimator can be reliably used. While the original SC papers recommend that the method should only be used in applications that present a good pre-treatment fit for a long series of pre-treatment periods, we show that, under some conditions, it can still be reliable even when the pre-treatment fit is imperfect. The conditions we derive for asymptotic unbiasedness provide a guideline on how applied researchers should justify the use of the method in empirical applications with imperfect pre-treatment fit.

If we relax the  non-negativity constraint on the weights, then the  estimator for the factor loadings of the treated unit will still be asymptotically unbiased when both the number of pre-treatment periods and the number of controls increase.\footnote{In this case, we need that the number of pre-treatment periods is greater or equal to the number of control units, so that the estimator is well defined. We  rely on stronger assumptions for the case in which the ratio between the number of control units and the number of pre-treatment periods converges to one. } However, due to the lack of regularization, this estimator may not be consistent if the ratio between the number of control units and the number of pre-treatment periods converges to one. We provide a simple example showing that, while the bias of the estimator for the treatment effects when we relax these constraints converges to zero when the number of control units goes to infinity, the variance of its asymptotic distribution is increasing with the ratio between the number of control units and pre-treatment periods. When this ratio becomes close to one, the variance of this asymptotic distribution diverges. This highlights the importance of using regularization methods when the number of pre-treatment periods is not much larger than the number of control units. 

We present a baseline SC setting in Section \ref{setting}. In Section \ref{Section_original_SC}, we analyze the asymptotic properties of the original SC estimator when both the number of pre-treatment periods and the number of control units go to infinity. In Section \ref{Section_OLS}, we analyze the asymptotic properties of the SC estimator when we relax the non-negativity and adding-up constraints in this setting. We present a simple Monte Carlo exercise in Section \ref{MC} to illustrate the theoretical results presented in Sections \ref{Section_original_SC} and \ref{Section_OLS}. Section \ref{conclusion} concludes.

\section{Setting}
\label{setting}
There are $i=0,1,...,J$ units, where unit 0 is treated and the other units  are controls. Potential outcomes when unit $i$ at time $t$ is treated ($y_{it}^I $) and non-treated ($y_{it}^N$) are determined by a linear factor model,
\begin{eqnarray} \label{model}
\begin{cases} y_{it}^N =  \boldsymbol{\lambda}_t \boldsymbol{\mu}_i + \epsilon_{it}  \\ 
y_{it}^I = \alpha_{it} + y_{it}^N, \end{cases}
\end{eqnarray}
where $\boldsymbol{\lambda}_t = [\lambda_{1t} ~ ... ~ \lambda_{Ft}]$ is an $1 \times F$ vector of unknown common factors, $\boldsymbol{\mu}_i$ is an $F \times 1$ vector of unknown factor loadings, and the error terms $\epsilon_{it}$ are unobserved idiosyncratic shocks.

We only observe $y_{it} = d_{it} y_{it}^I  +  (1-d_{it}) y_{it}^N$, where $d_{it}=1$ if unit $i$ is treated at time $t$. We analyze the properties of the SC estimator considering a repeated sampling framework over the distribution of $\epsilon_{it}$, conditional on fixed sequence of $\boldsymbol{\lambda}_t$ and $\boldsymbol{\mu}_i$. We define $\mathbf{M}_J$ as the $J \times F$ matrix that collects the information on the factor loadings of the control units (that is, the $j$-th row of $\mathbf{M}_J$ is equal to $\boldsymbol{\mu}_j'$).  We  observe $(y_{0t},...,y_{Jt})$ for periods $t \in \{ -T_0+1,...,-1,0,1,...,T_1 \}$, where treatment is assigned to unit 0 after time 0. Therefore, we have $T_0$ pre-treatment periods and $T_1$ post-treatment periods. Let $\mathcal{T}_0$ ($\mathcal{T}_1$) be the set of time indices in the pre-treatment (post-treatment) periods.  The main goal of the SC method is to estimate the effect of the treatment for unit 0 for each  $t \in \mathcal{T}_1$,  $\{ \alpha_{01},...,\alpha_{0T_1} \}$.

In a sequence of papers,  \cite{Abadie2003}, \cite{Abadie2010}, and \cite{Abadie2015}  proposed the SC method to estimate weights for the control units to construct a counter-factual for $\{ y_{01}^N,...,y_{0T_1}^N \}$.   In a version of the method where all pre-treatment outcome lags are included as predictor variables, those weights are estimated by minimizing the pre-treatment sum of squared residuals subject to the constraints that weights must be non-negative and sum one. \cite{Abadie2010} show that, if there are weights that provide a perfect pre-treatment fit, then the bias of the SC estimator is bounded by a function that asymptotes to zero when $T_0$ increases, even when $J$ is fixed. By perfect pre-treatment fit we mean that there is a $(w_1,...,w_J) \in \Delta^{J-1}$  such that $y_{0t} = \sum_{j =1}^J {w_j} y_{jt} $ for all $t \in \mathcal{T}_0$,  where $\Delta^{J-1} \equiv \{ (w_1,...,w_J) \in \mathbb{R}^{J} | w_j \geq 0 \mbox{ and } \sum_{j=1}^J w_j = 1\}$. However, \cite{FP_SC} show that, if the pre-treatment fit is imperfect, then the SC weights will not generally recover the factor loadings of the treated unit, so the SC estimator will be biased if there is selection on unobservables. They show that this result is valid even when $T_0 \rightarrow \infty$, as long as $J$ is fixed.  The main reason is that, for any $\mathbf{w}^\ast \in \mathbb{R}^J$ such that $\boldsymbol{\mu}_0 = {\mathbf{M}_J}' \mathbf{w}^\ast$, it is possible to write
\begin{eqnarray} \label{pop_reg}
y_{0t}^N = \mathbf{y}_t ' \mathbf{w}^\ast+ \epsilon_{0t}  - \boldsymbol{\epsilon}_t ' \mathbf{w}^\ast,
\end{eqnarray}
where $\mathbf{y}_t = (y_{1t},...,y_{Jt})'$, and  $\boldsymbol{\epsilon}_t = (\epsilon_{1t},...,\epsilon_{Jt})'$. Therefore, the outcomes of the control units serve as a proxy for the factor loadings of the treated unit. However, the linear factor model structure inherently generates a correlation between $\mathbf{y}_t$ and the error in this model due to the idiosyncratic shocks $\boldsymbol{\epsilon}_t$. As a consequence, with $J$ fixed, the SC weights will generally not converge in probability to a $\mathbf{w}^\ast$ such that $\boldsymbol{\mu}_0 = {\mathbf{M}_J}' \mathbf{w}^\ast$, even when $T_0 \rightarrow \infty$.

\section{Asymptotic Behavior of the Original SC Estimator with Large $T_0$ and Large $J$} \label{Section_original_SC}

We analyze the properties of the SC estimator when both the number of control units ($J$)  and the number of pre-treatment periods ($T_0$) increase. This provides a better asymptotic approximation to settings in which the number of pre-treatment periods and the number of control observations are roughly of the same size, as is common in SC applications (e.g., \cite{Doudchenko}). 

 Considering a SC specification that includes all pre-treatment outcome lags as predictors, the SC weights are given by
\begin{eqnarray} \label{SC_eq}
\widehat{\mathbf{w}}_{\mbox{\tiny SC}} =  \underset{{\mathbf{w}}  \in \Delta^{J-1}}{\mbox{argmin}} \left \{ \frac{1}{T_0} \sum_{t \in \mathcal{T}_0} \left( y_{0t} - { \mathbf{w}} ' {\mathbf{y}}_t  \right)^2    \right\}. 
\end{eqnarray}

The main challenge in analysing the behavior of the SC estimator in a setting with large $J$ and large $T_0$ is that, when $T_0 \rightarrow \infty$, the dimension of $\widehat{{\mathbf{w}}}_{\mbox{\tiny SC}}$ increases. However, we are not inherently interested in $\widehat{{\mathbf{w}}}_{\mbox{\tiny SC}}$, but in the \emph{implied} estimator of the factor loadings of the treated unit that is generated from $\widehat{{\mathbf{w}}}_{\mbox{\tiny SC}}$, the $F \times 1$ vector $\widehat{{\boldsymbol{\mu}}}_{\mbox{\tiny SC}} =  {\mathbf{M}_J}' \widehat{{\mathbf{w}}}_{\mbox{\tiny SC}}$. 
We consider, therefore, the asymptotic behavior of $\widehat{{\boldsymbol{\mu}}}_{\mbox{\tiny SC}}$. Note that model (\ref{model}) is equivalent to a model $y^N_{it} = \widetilde{\boldsymbol{\lambda}}_t \widetilde{\boldsymbol{\mu}}_i + \epsilon_{it} $, where $\widetilde{\boldsymbol{\lambda}}_t = \boldsymbol{\lambda}_t \mathbf{A}^{-1}$ and $\widetilde{\boldsymbol{\mu}}_i = \mathbf{A} \boldsymbol{\mu}_i$ for any invertible $(F \times F)$ matrix $\mathbf{A}$. This, however, does not invalidate our analysis.   If we consider $(\widetilde{\boldsymbol{\lambda}}_t, \widetilde{\boldsymbol{\mu}}_i)$ instead of $(\boldsymbol{\lambda}_t,\boldsymbol{\mu}_i)$ for any  invertible $(F \times F)$ matrix $\mathbf{A}$, then the synthetic control weights would remain the same, and the implied estimator for the factor loadings of the treated unit, given a sequence of common factors $\widetilde{\boldsymbol{\lambda}}_t$, would be $\widehat{\widetilde{\boldsymbol{\mu}}}_{\mbox{\tiny SC}} = \mathbf{A} {\mathbf{M}_J}' \widehat{{\mathbf{w}}}_{\mbox{\tiny SC}} = \mathbf{A} \widehat{\boldsymbol{\mu}}_{\mbox{\tiny SC}}$. Therefore, we have that $\widehat{\boldsymbol{\mu}}_{\mbox{\tiny SC}}  \buildrel p \over \rightarrow \boldsymbol{\mu}_0$ if, and only if, $\widehat{\widetilde{\boldsymbol{\mu}}}_{\mbox{\tiny SC}}  \buildrel p \over \rightarrow \widetilde{\boldsymbol{\mu}}_0 = \mathbf{A}\boldsymbol{\mu}_0$. Importantly, note that the estimator $\widehat{{\boldsymbol{\mu}}}_{\mbox{\tiny SC}}$ is not observed, because $\mathbf{M}_J$ is not observed. Rather, this is a construct to analyze whether the SC weights lead to a SC unit that is affected by the common factors in the same way as the treated unit. We do not aim to directly estimate $\boldsymbol{\mu}_0$, so this lack of identification for factor models does not pose any problem for our analysis.   

For a given ${\mathbf{w}}$, let $\boldsymbol{\mu} \equiv   {\mathbf{M}_J}' {{\mathbf{w}}}$. From the objective function in equation $(\ref{SC_eq})$, 
\begin{eqnarray} \label{eq_objective}
  \frac{1}{T_0} \sum_{t \in \mathcal{T}_0} \left( y_{0t} -  \mathbf{w} ' {\mathbf{y}}_t  \right)^2  =      \frac{1}{T_0} \sum_{t \in \mathcal{T}_0} \left( \boldsymbol{\lambda}_t (\boldsymbol{\mu}_0 - \boldsymbol{\mu}) + \epsilon_{0t} -  { \mathbf{w}} ' {\boldsymbol{\epsilon}}_t  \right)^2.
 \end{eqnarray}

Now define
\begin{eqnarray} \label{eq_H}  
\mathcal{H}_{J} (\boldsymbol{\mu}) =  {\underset{{\mathbf{w}}  \in \Delta^{J-1}: ~ { {\mathbf{M}_J}'{\mathbf{w}}  = \boldsymbol{\mu}}}{\mbox{min}} \left \{   \frac{1}{T_0} \sum_{t \in \mathcal{T}_0} ( \bar {\lambda}_t(\boldsymbol{\mu}) -   { \mathbf{w}} ' {\boldsymbol{\epsilon}}_t )  ^2 \right \}},
 \end{eqnarray}
where  $\bar {\lambda}_t(\boldsymbol{\mu}) \equiv \boldsymbol{\lambda}_t (\boldsymbol{\mu}_0 - \boldsymbol{\mu}) + \epsilon_{0t}$. Then $\widehat{{\boldsymbol{\mu}}}_{\mbox{\tiny SC}} =  {\mathbf{M}_J}' {\widehat{\mathbf{w}}}_{\mbox{\tiny SC}} = \underset{ \boldsymbol{\mu} \in \mathcal{M}_J  }{\mbox{argmin}} \mathcal{H}_{J}(\boldsymbol{\mu}) $, where  $ \mathcal{M}_J \equiv \{ \boldsymbol{\mu} \in \mathbb{R}^F | \boldsymbol{\mu} =   {\mathbf{M}_J}' {\mathbf{w}} \mbox{ for some } {\mathbf{w}} \in \Delta^{J-1}  \}$ is the set of factor loadings that can be attained with weights $\mathbf{w} \in \Delta^{J-1}$ when there are $J$ control units.

Using this characterization of $\widehat{{\boldsymbol{\mu}}}_{\mbox{\tiny SC}}$, we provide conditions under which $\widehat{{\boldsymbol{\mu}}}_{\mbox{\tiny SC}} \buildrel p \over \rightarrow \boldsymbol{\mu}_0$ when $T_0$ and $J \rightarrow \infty$. We consider the following assumptions on the idiosyncratic shocks. 

\begin{ass}{(idiosyncratic shocks)} \label{Assumption_e}
\normalfont
(a) $\mathbb{E}[\epsilon_{it}]=0$  for all $i$ and $t$; (b) $\{ \epsilon_{it} \}_{t \in \mathcal{T}_0 \cup \mathcal{T}_1}$ are independent across $i$; (c) $\{\epsilon_{0t},...,\epsilon_{Jt} \}_{t \in \mathcal{T}_0}$ is $\alpha$-mixing; (d) $\epsilon_{it}$ have uniformly bounded fourth moments across $i$ and $t$, and $\frac{1}{T_0}\sum_{t \in \mathcal{T}_0} \mathbb{E}[\epsilon_{0t}^2] \rightarrow \sigma^2_0$; (e) $\exists \underline{\gamma}>0$ such that  $ \mathbb{E}[\epsilon_{it}^2]  \geq \underline{\gamma}$ across $i$ and $t$.
\end{ass}

Since we are considering treatment assignment, factor loadings, and common factors as  fixed,  Assumption \ref{Assumption_e}(a) implies that the idiosyncratic shocks are uncorrelated with the treatment assignment and with the factor structure.\footnote{We can think of Assumption \ref{Assumption_e}(a) as the expected value of the idiosyncratic shocks being equal to zero conditional on treatment assignment, factor loadings, and common factors. Therefore, if we consider an underlying distribution for the treatment assignment, factor loadings, and common factors, then Assumption   \ref{Assumption_e}(a)  implies that the idiosyncratic shocks are mean independent conditional on these variables.  } Note, however, that this would allow for, for example, dependence between $var(\epsilon_{it})$ and treatment assignment or the factor structure.  Importantly, by conditioning on the treatment assignment, factor loadings, and common factors, we do not impose any restriction on the relationship between treatment assignment and the factor structure.  Assumption \ref{Assumption_e}(b)  implies that the idiosyncratic shocks are uncorrelated across units, so that all spatial correlation is captured by the factor structure. While we allow for serial correlation in $\epsilon_{it}$, Assumption \ref{Assumption_e}(c) restricts such dependence by assuming a mixing condition. Finally, while we do not require stationarity, Assumptions \ref{Assumption_e}(d) and \ref{Assumption_e}(e) impose some restrictions on the moments of $\epsilon_{it}$.

We also consider the following assumptions on the sequence of factor loadings and common factors. Let $\norm{.}_2$ be the Frobenius norm. 

\begin{ass}{(factor loadings)} \label{Assumption_mu}
\normalfont
(a) As $T_0 \rightarrow \infty$, there is a sequence $\mathbf{w}^\ast_J \in \Delta^{J-1}$ such that $\norm{{\mathbf{M}_J}' \mathbf{w}^\ast_J - \boldsymbol{\mu}_0}_2 \rightarrow 0$, and $\norm{\mathbf{w}^\ast_J}_2 \rightarrow 0$, and (b) the sequence $\boldsymbol{\mu}_i$ is uniformly bounded. 

\end{ass}

Assumption \ref{Assumption_mu}(a) implies that there is a sequence of weights ($\mathbf{w}^\ast_J$) diluted among an increasing number of control units, and that are such that the implied factor loadings associates with those weights ($\boldsymbol{\mu}^\ast_J \equiv {\mathbf{M}_J}' \mathbf{w}^\ast_J$)  reconstruct the factor loadings of the treated unit ($\boldsymbol{\mu}_0$) in the limit. Importantly, Assumption \ref{Assumption_mu} implies that $J \rightarrow \infty$ when $T_0 \rightarrow \infty$. Otherwise, it would not be possible to reconstruct $\boldsymbol{\mu}_0$ with weights such that $\norm{\mathbf{w}^\ast_J}_2 \rightarrow 0$.

Recall that this analysis is conditional on a fixed sequence of factor loadings. If we assume, for example, that the underlying  distribution of factor loadings has finite support $\{\mathbf{m}_1,...,\mathbf{m}_{\bar q} \}$, with $Pr(\boldsymbol{\mu}_i = \mathbf{m}_q) = p_q > 0$ independent across $i$, then the conditions imposed in Assumption \ref{Assumption_mu}  for the factor loadings would be satisfied with probability one  (details in Appendix \ref{Appendix_mu}).  This assumption would also be satisfied with probability one even if we consider a case in which the distributions of $\boldsymbol{\mu}_0$ and $\boldsymbol{\mu}_i$ for $i>0$  are different, as long as every point in the support of the distribution of $\boldsymbol{\mu}_0$ is in the convex hull of $\{\mathbf{m}_1,...,\mathbf{m}_{\bar q} \}$. Assumption \ref{Assumption_mu}(b) guarantees that the parameter space $\mathcal{M} =\mbox{cl} \left(\cup_{J \in \mathbb{N}} \mathcal{M}_J \right)$, which is the closure of $\cup_{J \in \mathbb{N}} \mathcal{M}_J$, is compact. 

\begin{ass}{(common factors)} \label{Assumption_lambda}
\normalfont
$\frac{1}{T_0} \sum_{t \in \mathcal{T}_0} \boldsymbol{\lambda}_t ' \boldsymbol{\lambda}_t   \rightarrow \boldsymbol{\Omega}$ positive definite.

\end{ass}

Assumption \ref{Assumption_lambda} implies that common factors generate enough independent variation so that we can identify the effects of each factor on the pre-treatment outcomes.  \cite{Abadie2010} consider a similar assumption. If we consider an underlying distribution for $\boldsymbol{\lambda}_t$ such that, for example, $\boldsymbol{\lambda}_t$ is $\alpha$-mixing with uniformly bounded fourth moments, and that ${T_0}^{-1} \sum_{t \in \mathcal{T}_0} \mathbb{E}[ \boldsymbol{\lambda}_t ' \boldsymbol{\lambda}_t ] \rightarrow \boldsymbol{\Omega}$, then ${T_0}^{-1} \sum_{t \in \mathcal{T}_0} \boldsymbol{\lambda}_t ' \boldsymbol{\lambda}_t  \buildrel a.s. \over \rightarrow \boldsymbol{\Omega}$. In this case,  Assumption \ref{Assumption_lambda} would be satisfied with probability one.

We also assume some technical conditions that are important to take into account that the number of control units goes to infinity with the number of pre-treatment periods.  
 
\begin{ass}{(other assumptions)} \label{Assumption_technical}
\normalfont
(a) $\underset{1 \leq j \leq J}{\mbox{max}} \left\{  \left| \frac{1}{T_0} \sum_{t \in \mathcal{T}_0} \epsilon_{0t} \epsilon_{jt} \right| \right\}=o_p(1)$ and, for all $f=1,...,F$, $\underset{0 \leq j \leq J}{\mbox{max}} \left\{ \left| \frac{1}{T_0} \sum_{t \in \mathcal{T}_0} \lambda_{ft} \epsilon_{jt} \right| \right\}=o_p(1)$; (b) $\exists$ $c>0$ such that  $\underset{{1 \leq j \leq J}}{\mbox{min}} \left\{ \sum_{t \in \mathcal{T}_0} \left| \epsilon_{jt}^2 \right| \right\}\geq c T_0$ with probability $1-o(1)$, and $\underset{1 \leq i, j \leq J, i \neq j}{\mbox{max}} \left\{ \left| \frac{1}{T_0} \sum_{t \in \mathcal{T}_0}  \epsilon_{it} \epsilon_{jt} \right| \right\}=o_p(1)$.

\end{ass}

These high-level conditions essentially determine the rate in which $J$ can diverge when $T_0 \rightarrow \infty$. Whether these conditions are satisfied depend   crucially on the rates in which $J$ and $T_0$ diverge,  on the dependence of $\epsilon_{it}$, and on the number of uniformly bounded moments of  $\epsilon_{it}$. If we allow $J$  to diverge at a faster rate than $T_0$, or we allow time-series dependence on $\epsilon_{it}$, then we need a larger number of uniformly bounded moments of  $\epsilon_{it}$. See Appendix \ref{Appendix_technical} for some simple examples in which Assumption \ref{Assumption_technical} is satisfied even when $J$ diverges at a faster rate than $T_0$. 

Given these conditions, we derive the following results.
 
\begin{prop} \label{Convergence_SC}

Suppose we observe $(y_{0t},...,y_{Jt})$ for periods $t \in \{ -T_0+1,...,-1,0,1,...,T_1 \}$, where $J$ is a function of $T_0$. Potential outcomes are defined in equation (\ref{model}). Let $\widehat{{\boldsymbol{\mu}}}_{\mbox{\tiny SC}}$ be defined as ${\mathbf{M}_J}' \widehat{{\mathbf{w}}}_{\mbox{\tiny SC}}$, where $\widehat{{\mathbf{w}}}_{\mbox{\tiny SC}}$ is defined in equation $(\ref{SC_eq})$. Suppose Assumptions \ref{Assumption_e} to \ref{Assumption_lambda}, and Assumption \ref{Assumption_technical}(a) hold. Then, as $T_0 \rightarrow \infty$, (i) $\widehat{{\boldsymbol{\mu}}}_{\mbox{\tiny SC}} \buildrel p \over \rightarrow \boldsymbol{\mu}_0$, and (ii) $ \frac{1}{T_0} \sum_{t \in \mathcal{T}_0} \left( y_{0t} -\widehat{\mathbf{w}}_{\mbox{\tiny SC}} ' {\mathbf{y}}_t  \right)^2  \buildrel p \over \rightarrow \sigma_0^2$. Moreover, if we add Assumption \ref{Assumption_technical}(b), then (iii) $\norm{\widehat{{\mathbf{w}}}_{\mbox{\tiny SC}}}_2 \buildrel p \over \rightarrow 0$.

\end{prop}

The first result in Proposition \ref{Convergence_SC} shows that, asymptotically, the SC unit will be affected by the common shocks $\boldsymbol{\lambda}_t$ in the same way as the treated unit. The main idea of the proof is the following. We consider an extension of the function $\mathcal{H}_{J} (\boldsymbol{\mu})$ to the domain $\mathcal{M} =\mbox{cl} \left( \cup_{J \in  \mathbb{N}}\mathcal{M}_J \right)$ such that $ \underset{ \boldsymbol{\mu} \in \mathcal{M}  }{\mbox{argmin}} \widetilde{\mathcal{H}}_{J}(\boldsymbol{\mu}) = \underset{ \boldsymbol{\mu} \in \mathcal{M}_J  }{\mbox{argmin}} \mathcal{H}_{J}(\boldsymbol{\mu}) $. Then we show that  $\widetilde{\mathcal{H}}_{J} (\boldsymbol{\mu}_0) \buildrel p \over \rightarrow  \sigma^2_0$, and that  $\widetilde{\mathcal{H}}_{J} (\boldsymbol{\mu})$ is bounded from below by a function $\widetilde{\mathcal{H}}^{LB}_{T_0} (\boldsymbol{\mu})$ that converges uniformly in $\boldsymbol{\mu} \in \mathcal{M}$ to $(\boldsymbol{\mu}_0 -  \boldsymbol{\mu})' \boldsymbol{\Omega} (\boldsymbol{\mu}_0 -  \boldsymbol{\mu})+ \sigma^2_0$. Under Assumption \ref{Assumption_lambda}, $(\boldsymbol{\mu}_0 -  \boldsymbol{\mu})' \boldsymbol{\Omega} (\boldsymbol{\mu}_0 -  \boldsymbol{\mu})+ \sigma^2_0$ is uniquely minimized at $\boldsymbol{\mu}_0$, which implies that $\widehat{{\boldsymbol{\mu}}}_{\mbox{\tiny SC}}  \buildrel p \over \rightarrow \boldsymbol{\mu}_0$. The details of the proof are presented in  Appendix \ref{Proof_Convergence_SC}.

 \cite{FP_SC} show that, when the number of control units is fixed,  the SC weights converge to weights that do not, in general, recover $\boldsymbol{\mu}_0$. This happens because, in a setting with a fixed number of control units, the SC weights converge to weights that simultaneously attempt to minimize both the second moments of the remaining common shocks, and the variance of a weighted average of the idiosyncratic shocks of the control units. Intuitively, this first result from  Proposition \ref{Convergence_SC} comes from the fact that, when  both the number of pre-treatment periods and the number of controls increase, the importance of this variance of a weighted average of the idiosyncratic shocks of the control units vanishes. As a consequence, the asymptotic bias of  $\widehat{{\boldsymbol{\mu}}}_{\mbox{\tiny SC}}$ disappears when both the number of pre-treatment periods and the number of controls increase.  

A crucial condition for this result is that, as the number of control units increases, it is possible to recover $\boldsymbol{\mu}_0$ with weights that are diluted among an increasing number of control units (Assumption \ref{Assumption_mu}). If we consider, for example, a setting such that there is only a fixed number of control units that can be used to recover $\boldsymbol{\mu}_0$, and the additional control units are uncorrelated with $y_{0t}$, then the result from \cite{FP_SC} would still apply, and  $\widehat{{\boldsymbol{\mu}}}_{\mbox{\tiny SC}}$ would not converge to $\boldsymbol{\mu}_0$. Such setting would  be inconsistent with Assumption \ref{Assumption_mu}.

Proposition \ref{Convergence_SC} also shows that the SC weights will get diluted among an increasing number of control units, so that $\norm{\widehat{{\mathbf{w}}}_{\mbox{\tiny SC}}}_2 \buildrel p \over \rightarrow 0$. An immediate consequence is that, if we assume that idiosyncratic shocks in the post-treatment periods are independent from the idiosyncratic shocks in the pre-treatment periods, then, for any $t \in \mathcal{T}_1$, $\hat  \alpha_{0t}^{\mbox{\tiny SC}} \equiv y_{0t} - \mathbf{y}_t '\widehat{{\mathbf{w}}}_{\mbox{\tiny SC}}  \buildrel p \over \rightarrow \alpha_{0t} + \epsilon_{0t}$ when $T_0 \rightarrow \infty$.

\begin{cor} \label{corollary}

Suppose all assumptions for Proposition \ref{Convergence_SC} are satisfied, and that, for all $t \in \mathcal{T}_1$, $\epsilon_{it}$ is independent from $\{\epsilon_{i\tau}\}_{\tau \in \mathcal{T}_0}$. Then, for any $t \in \mathcal{T}_1$, $\hat  \alpha_{0t}^{\mbox{\tiny SC}}  \buildrel p \over \rightarrow \alpha_{0t} + \epsilon_{0t}$ when $T_0 \rightarrow \infty$.
 
\end{cor}

This happens because not only $\widehat{{\boldsymbol{\mu}}}_{\mbox{\tiny SC}}   \buildrel p \over \rightarrow \boldsymbol{\mu}_0$, but also $\widehat{\mathbf{w}}_{\mbox{\tiny SC}}$ is diluted among an increasing number of control units, implying that $\boldsymbol{\epsilon}_t'\widehat{{\mathbf{w}}}_{\mbox{\tiny SC}}  \buildrel p \over \rightarrow 0$ (see details in Appendix \ref{Proof_corollary}). Therefore, if treatment assignment is uncorrelated with $ \epsilon_{0t}$, as considered in Assumption \ref{Assumption_e}(a), then the SC estimator is asymptotically unbiased for $\alpha_{0t}$ even when treatment assignment is correlated with the factor structure. Moreover, the asymptotic distribution of the SC estimator depends only on the idiosyncratic shocks of the treated unit in period $t$.   In Appendix \ref{Appendix_corollary} we present Corollary \ref{appendix_corollary}, in which  we derive $\hat  \alpha_{0t}^{\mbox{\tiny SC}}  \buildrel p \over \rightarrow \alpha_{0t} + \epsilon_{0t}$ allowing for   time dependence in the idiosyncratic shocks.

Our results are closely linked to Theorem 5 by \cite{SDID}, who  consider a penalized version of the SC weights. These penalized SC weights  solve the minimization problem presented in equation $(\ref{SC_eq})$ subject to the additional constraint that $|| {\mathbf{w}}||_2 \leq a_w$. Since $|| {\mathbf{w}} ||_1 = 1 \Rightarrow ||{\mathbf{w}}||_2 \leq 1$, note that the original SC weights are equivalent to the penalized SC weights with $a_w=1$. They show that the approximation error for their low-rank matrix structure goes to zero if, among other conditions,  $a_w \rightarrow 0$. In contrast, we show that, in our setting, the SC weights achieve such balancing even when the penalty term $a_w$ does not go to zero, so it is not necessary  to force weights to be positive for many control units in large samples with an $L_2$ penalization term.  In this case, the original SC method --- which does not include an $L_2$ penalization term --- provides a consistent estimator for $\boldsymbol{\mu}_0$ with weights such that $\widehat{\mathbf{w}}_{\mbox{\tiny SC}} ' \boldsymbol{\epsilon}_t  \buildrel p \over \rightarrow 0$.

Finally, under the assumptions considered in Proposition \ref{Convergence_SC}, $ {T_0}^{-1} \sum_{t \in \mathcal{T}_0} \left( y_{0t} -\widehat{\mathbf{w}}_{\mbox{\tiny SC}} ' {\mathbf{y}}_t  \right)^2$ converges in probability to $ \sigma_{\bar \lambda}^2( \boldsymbol{\mu}_0) = \sigma^2_0$, which is the asymptotic variance of $\epsilon_{0t}$. Therefore, the SC unit will asymptotically absorb all variability of $y_{0t}$ that is related to the factor structure, but will \emph{not} over-fit the idiosyncratic shocks of the treated unit. This happens because the non-negativity and adding-up constraints on the weights work as a regularization method, as presented by \cite{Chernozhukov}.\footnote{\cite{Chernozhukov} derive conditions under which the original SC estimator converges in probability. While they consider the case in which the outcomes of the control units are uncorrelated with the error in a model similar to the one presented in equation $(\ref{pop_reg})$, such condition would not be satisfied if we consider a linear factor model as the one presented in model (\ref{model}) for the potential outcomes. Proposition \ref{Convergence_SC} provides conditions under  which  the original SC estimator converges to weights that recover $\boldsymbol{\mu}_0$ even when the linear factor model structure induces such correlation. Increasing the number of control units is not sufficient to generate this result. It is crucial that  the number of control units that can be used to recover $\boldsymbol{\mu}_0$ increases with the total number of control units, so that Assumption \ref{Assumption_mu} is satisfied. } 
 This implies that  we should \emph{not} expect a perfect pre-treatment fit in this setting, even when ${J}$ grows at a faster rate than  $T_0$. Therefore, we provide conditions in which the SC estimator can be reliably used even in a setting in which the original SC papers recommend that the method should not be used (e.g., \cite{Abadie2010} and \cite{Abadie2015}). Moreover, this highlights that the asymptotic unbiasedness result from Proposition \ref{Convergence_SC} does not come from a better pre-treatment fit when we increase the number of control units. Rather, it comes from the fact that increasing the number of control units implies existence of balancing weights that are diluted among an increasing number of control units, implying that the problems highlighted by \cite{FP_SC}  become asymptotically irrelevant.

\begin{rem}
\normalfont
{Proposition \ref{Convergence_SC} remains valid if we consider a demeaned SC estimator, as proposed by \cite{FP_SC}, which is numerically the same as including a constant in the minimization problem (\ref{SC_eq}), as proposed by \cite{Doudchenko}. We show in Appendix \ref{Appendix_demeaned} that this would require only minor adjustments in the proof of Proposition \ref{Convergence_SC}.  }

\end{rem}

\begin{rem}

\normalfont
While we focus on the SC specification that includes all pre-treatment outcome lags as predictors, we consider a setting with covariates in Appendix \ref{Appendix_covariates}. We show that the conclusions from Proposition \ref{Convergence_SC}  remain valid for SC specifications that include time-invariant covariates as predictors, as long as the number of pre-treatment outcomes lags used as predictors goes to infinity when $T_0 \rightarrow \infty$. This result is an extension of the conclusions from \cite{FPP} for the case in which both $J$ and $T_0$ diverge. We also present in Appendix \ref{Appendix_covariates}  Monte Carlo simulations considering a setting with covariates and different SC specifications. In this setting, both the SC weights estimated from equation (\ref{SC_eq}), and the SC weights using half of the pre-treatment outcomes and covariates as predictors, approximately recover both the factor loadings and the time-invariant covariates of the treated unit when  $(T_0,J)$ are large.\footnote{We also include in our simulations in Appendix  \ref{Appendix_covariates}  a SC specification that does not satisfy the condition on the number of pre-treatment outcomes used as predictors going to infinity. In particular, we evaluate a SC specification that includes the average of the pre-treatment outcomes and additional covariates as predictors. In this case, the SC weights failed to recover the factor loadings of the treated unit even when $(T_0,J)$ are large.     } 

\end{rem}

\section{Relaxing the non-negativity constraints } \label{Section_OLS}

We consider now the importance of the regularization provided by the non-negativity and adding-up constraints for the results presented in Section \ref{Section_original_SC}. When the non-negativity constraint is relaxed, the estimator of $\boldsymbol{\mu}_0$ would remain asymptotically unbiased, but it may not be consistent if $J/T_0 \rightarrow 1$. We consider  the case without both the adding-up and the non-negativity constraints. The case with only the adding-up constraint is similar. In this case, the weights are estimated using the OLS regression
\begin{eqnarray} \label{OLS_eq}
\widehat{{\mathbf{b}}}_{\mbox{\tiny OLS}} =  \underset{{\mathbf{b}}  \in \mathbb{R}^{J}}{\mbox{argmin}} \left \{ \frac{1}{T_0} \sum_{t \in \mathcal{T}_0} \left( y_{0t} - { \mathbf{b}} ' {\mathbf{y}}_t  \right)^2    \right\}. 
\end{eqnarray}

Following the same arguments presented in Section \ref{Section_original_SC}, we can define
\begin{eqnarray} \label{eq_H_OLS}
\mathcal{H}^{\mbox{\tiny OLS}}_{J} (\boldsymbol{\mu}) ={\underset{{\mathbf{b}}  \in \mathbb{R}^{J} : ~  {\mathbf{M}_J}'{\mathbf{b}}  = \boldsymbol{\mu}}{\mbox{min}} \left \{   \frac{1}{T_0} \sum_{t \in \mathcal{T}_0} ( \bar {\lambda}_t(\boldsymbol{\mu}) -   { \mathbf{b}} ' {\boldsymbol{\epsilon}}_t )  ^2 \right \}}, 
 \end{eqnarray}
so that $\widehat{\boldsymbol{\mu}}_{\mbox{\tiny OLS}} \equiv  {\mathbf{M}_J}'\widehat{{\mathbf{b}}}_{\mbox{\tiny OLS}}$ is the solution to $\underset{ \boldsymbol{\mu} \in \mathcal{M}_J^{\mbox{\tiny OLS}}  }{\mbox{argmin}} \mathcal{H}^{\mbox{\tiny OLS}}_{T_0}(\boldsymbol{\mu}) $, where $ \mathcal{M}_J^{\mbox{\tiny OLS}} \equiv \{ \boldsymbol{\mu} \in \mathbb{R}^F | \boldsymbol{\mu} =   {\mathbf{M}_J}'{\mathbf{b}} \mbox{ for some } {\mathbf{b}} \in \mathbb{R}^{J}  \}$.

A crucial difference in this case is that, by not imposing any restriction on ${\mathbf{b}}$, this minimization problem is subject to over-fitting when $J$ increases with $T_0$. As a consequence, the lower bound we derive in the proof of Proposition \ref{Convergence_SC}, which in this case would be given by  ${\underset{{\mathbf{b}}  \in \mathbb{R}^{J} }{\mbox{min}} \left \{   \frac{1}{T_0} \sum_{t \in \mathcal{T}_0} ( \bar {\lambda}_t(\boldsymbol{\mu}) -   { \mathbf{b}} ' {\boldsymbol{\epsilon}}_t )  ^2 \right \}}$, would not generally converge to $\sigma^2_{\bar \lambda}(\boldsymbol{\mu})$.  In the extreme example in which $T_0 = J$, this lower bound would be equal to zero for all $\boldsymbol{\mu}$  with probability one. We can still show, however, that, when $J/T_0 \rightarrow c <1$,  $\widehat{\boldsymbol{\mu}}_{\mbox{\tiny OLS}}  \buildrel p \over \rightarrow \boldsymbol{\mu}_0$. Moreover, we can show that, under some conditions, $\mathbb{E}[\widehat{\boldsymbol{\mu}}_{\mbox{\tiny OLS}} - \boldsymbol{\mu}_0] \rightarrow 0$ even when $J/T_0 \rightarrow 1$.

Consider first the case in which $J/T_0 \rightarrow c <1$. We continue to consider the properties of the  estimator over the distribution of the idiosyncratic shocks, and conditional on fixed sequences of common factors and factor loadings. Regarding the idiosyncratic shocks, we continue to consider  Assumption \ref{Assumption_e}. We impose the following assumptions on the sequence of factor loadings.

\begin{ass}{(factor loadings)}
\normalfont \label{Assumption_mu_ols}
For some $\underline{a},\bar{a}>0$, let $R$ be the number of disjoint groups of $F$ control units we can arrange such that the $F \times F$ matrix with the factor loadings for each of those groups has its smallest eigenvalue greater than $\underline{a}$, and its largest eigenvalue smaller than $\bar{a}$. We assume that $R \rightarrow \infty$ when $J \rightarrow \infty$.

\end{ass}

Differently from the setting considered in Section \ref{Section_original_SC}, note that there will always exist a $\boldsymbol{\beta} \in \mathbb{R}^J$ such that ${\mathbf{M}_J}'\boldsymbol{\beta} = \boldsymbol{\mu}_0$, as long as there is at least one group of $F$ control units such that their factor loadings form a basis for $\mathbb{R}^F$. Assumption \ref{Assumption_mu_ols} implies not only that we can find a  $\boldsymbol{\beta} \in \mathbb{R}^J$ such that ${\mathbf{M}_J}' \boldsymbol{\beta} = \boldsymbol{\mu}_0$ when $J$ is large enough, but also that we can find a sequence of weights $\boldsymbol{\beta}^\ast_J \in \mathbb{R}^J$ such that ${\mathbf{M}_J}'\boldsymbol{\beta}^\ast_J = \boldsymbol{\mu}_0$ and $\norm{\boldsymbol{\beta}^\ast_J}_2 \rightarrow 0$.

Let $\boldsymbol{\Lambda}$  be the  $T_0 \times F$ matrix with rows equal to $\boldsymbol{\lambda}_t$ for $t \in \mathcal{T}_0$, and $\boldsymbol{\epsilon}_0$ be the $T_0 \times 1$ vector with $\epsilon_{0t}$ for $t \in \mathcal{T}_0$. We assume that the sequence $\boldsymbol{\Lambda}$ satisfies the following conditions.

\begin{ass}{(common factors)}
\normalfont \label{Assumption_lambda_ols}
Let $\mathbf{Q}_{T_0}$ be a sequence of random symmetric and idempotent $T_0 \times T_0$ matrices with rank $K$, where $K \rightarrow \infty$ when $T_0 \rightarrow \infty$. Then $\frac{1}{K}\boldsymbol{\Lambda}'\mathbf{Q}_{T_0} \boldsymbol{\Lambda}=O_p(1)$ and $\left(\frac{1}{K}\boldsymbol{\Lambda}'\mathbf{Q}_{T_0} \boldsymbol{\Lambda} \right)^{-1}=O_p(1)$. Moreover, if $\mathbf{Q}_{T_0} $ are independent of $\boldsymbol{\epsilon}_0 $, then $\frac{1}{K}\boldsymbol{\Lambda}'\mathbf{Q}_{T_0} \boldsymbol{\epsilon}_0 =o_p(1)$.
 
\end{ass}

Assumption \ref{Assumption_lambda_ols} is satisfied with probability one if the underlying distribution for $\boldsymbol{\lambda}_t$ is iid normal with mean zero, and  $\boldsymbol{\lambda}_t$ is independent from $\mathbf{Q}_{T_0}$.\footnote{Note that assuming $\boldsymbol{\lambda}_t$ iid normal with mean zero, we can assume without loss of generality that $\lambda_{ft}$ and $\lambda_{f't}$ are independent. In this case, we would just have to normalize the covariance matrix of $\boldsymbol{\lambda}_t$.} In this particular case, we would have $K^{-1}\boldsymbol{\Lambda}'\mathbf{Q}_{T_0} \boldsymbol{\Lambda} = K^{-1} \sum_{q=1}^K \widetilde{\boldsymbol{\lambda}}_q' \widetilde{\boldsymbol{\lambda}}_q$, where $\widetilde{\boldsymbol{\lambda}}_q$ is iid and has the same distribution as $\boldsymbol{\lambda}_q$, which implies that  $K^{-1}\boldsymbol{\Lambda}'\mathbf{Q}_{T_0} \boldsymbol{\Lambda}   \buildrel a.s. \over \rightarrow  \mathbb{E}[\boldsymbol{\lambda}_t ' \boldsymbol{\lambda}_t]$. Likewise, if we also have $\epsilon_{0t}$ iid normal and independent from $\boldsymbol{\lambda}_t$, then $K^{-1}\boldsymbol{\Lambda}'\mathbf{Q}_{T_0} \boldsymbol{\epsilon}_0 = K^{-1} \sum_{q=1}^K \widetilde{\boldsymbol{\lambda}}_q \tilde \epsilon_{0t}  \buildrel a.s. \over \rightarrow \mathbb{E}[ \widetilde{\boldsymbol{\lambda}}_q \tilde \epsilon_{0t}]=0$.

Given these assumptions, we show that, when $J/T_0 \rightarrow c<1$,   $\widehat{\boldsymbol{\mu}}_{\mbox{\tiny OLS}} \buildrel p \over \rightarrow \boldsymbol{\mu}_0$. 

\begin{prop} \label{Prop_OLS_K_large}

Suppose we observe $(y_{0t},...,y_{Jt})$ for periods $t \in \{ -T_0+1,...,-1,0,1,...,T_1 \}$, where $J$ is a function of $T_0$. Potential outcomes are defined in equation (\ref{model}).  Let  $\widehat{\boldsymbol{\mu}}_{\mbox{\tiny OLS}}$  be defined as ${\mathbf{M}_J}' \widehat{{\mathbf{b}}}_{\mbox{\tiny OLS}}$, where $\widehat{{\mathbf{b}}}_{\mbox{\tiny OLS}}$ is defined in equation $(\ref{OLS_eq})$. Assume that $J/T_0 \rightarrow c \in [0,1)$, and that Assumptions \ref{Assumption_e}, \ref{Assumption_mu_ols}, and  \ref{Assumption_lambda_ols} hold. Then, when  $T_0 \rightarrow \infty$,  $\widehat{\boldsymbol{\mu}}_{\mbox{\tiny OLS}} \buildrel p \over \rightarrow \boldsymbol{\mu}_0$.

\end{prop}

The intuition is the same as the intuition in Proposition \ref{Convergence_SC}.  When the number of control units increases, we are able to have a diluted weighted average of the control units that recover $\boldsymbol{\mu}_0$. This reduces the importance of the variance of the linear combination of the idiosyncratic shocks of the control units in the minimization problem $(\ref{OLS_eq})$ for the estimation of $\widehat{{\mathbf{b}}}_{\mbox{\tiny OLS}}$, making the problem raised by \cite{FP_SC} less relevant. Since the number of degrees of freedom, $T_0 - J$, goes to infinity, the  estimator $\widehat{\boldsymbol{\mu}}_{\mbox{\tiny OLS}}$ converges in probability even when $J \rightarrow \infty$. This is consistent with Theorem 1 from \cite{Cattaneo}, once we consider a change in variables so that we can divide the $J$ control variables into a group of $F$ variables such that their associated estimators give us $\widehat{\boldsymbol{\mu}}_{\mbox{\tiny OLS}}$, and a remaining group of $J-F$ variables that we are not inherently interested in. As in \cite{Cattaneo},  $J$ can be a nonvanishing  fraction of $T_0$, but we cannot have that $J/T_0 \rightarrow 1$. It is easy to show that the assumptions for Theorem 1 from \cite{Cattaneo} hold if we assume that the data is iid normal. We consider here  an alternative  proof for Proposition \ref{Prop_OLS_K_large} where we take advantage of the specific details of our application, so that we can consider a weaker set of assumptions. See details of the proof in Appendix \ref{Proof_OLS_large}.

When  $J/T_0 \rightarrow 1$, we will not generally have that $\widehat{\boldsymbol{\mu}}_{\mbox{\tiny OLS}} \buildrel p \over \rightarrow \boldsymbol{\mu}_0$. If we impose a stronger set assumptions, however, we can still show that  $\mathbb{E}[\widehat{\boldsymbol{\mu}}_{\mbox{\tiny OLS}} - \boldsymbol{\mu}_0] \rightarrow 0$, regardless of the ratio $J/T_0$. The only restriction is that $T_0 \geq J$, so that the OLS estimator is well specified.  In this case, we continue to condition on a sequence of $\boldsymbol{\mu}_i$, but we consider $\boldsymbol{\lambda}_t$ stochastic. 

\begin{ass}{(normality)}
\normalfont \label{Assumption_normal_ols}
$\epsilon_{it}  \sim N(0,\sigma_i^2)$ iid across $t$ for all $i \in \mathbb{N}\cup \{0\}$, and $\boldsymbol{\lambda}_t   \buildrel iid \over \sim N(0,\boldsymbol{\Omega})$, where $\boldsymbol{\Omega}$ is positive definite. All these variables are independent of each other.
 
\end{ass}

\begin{prop} \label{Prop_OLS_K_finite}

Suppose we observe $(y_{0t},...,y_{Jt})$ for periods $t \in \{ -T_0+1,...,-1,0,1,...,T_1 \}$, where $J$ is a function of $T_0$. Potential outcomes are defined in equation (\ref{model}).  Let  $\widehat{\boldsymbol{\mu}}_{\mbox{\tiny OLS}}$  be defined as ${\mathbf{M}_J}' \widehat{{\mathbf{b}}}_{\mbox{\tiny OLS}}$, where $\widehat{{\mathbf{b}}}_{\mbox{\tiny OLS}}$ is defined in equation $(\ref{OLS_eq})$. Assume that $T_0 \geq J$, and that Assumptions \ref{Assumption_mu_ols} and \ref{Assumption_normal_ols} hold. Then,  when $T_0 \rightarrow \infty$, $\mathbb{E}[\widehat{\boldsymbol{\mu}}_{\mbox{\tiny OLS}}  - \boldsymbol{\mu}_0] \rightarrow 0$. 

\end{prop}

Proposition \ref{Prop_OLS_K_finite} reinforces that the bias in the estimator of $\boldsymbol{\mu}_0$ goes to zero when $J \rightarrow \infty$ because it makes the endogeneity problem highlighted  by \cite{FP_SC} less relevant (if Assumption \ref{Assumption_mu_ols} holds).  The only assumption we make on the number of control units and pre-treatment periods is that $T_0 \geq J$, so that the OLS estimator is well defined. Therefore, this conclusion is valid  even when $T_0 - J$ does not go to infinity.  However, relaxing the non-negativity and adding-up constraints when $T_0 - J$ does not diverge comes at the cost of having an estimator for $\boldsymbol{\mu}_0$ that may  not be consistent (even though the bias goes to zero), which translates into larger variance. See details of the proof of Proposition \ref{Prop_OLS_K_finite} in Appendix \ref{Proof_OLS_finite}. The assumption that $\mathbb{E}[\boldsymbol{\lambda}_t]=0$ can be relaxed if we consider the demeaned SC estimator.

To illustrate the trade-offs between using the constraints or not, we consider a very simple example in which we can derive the asymptotic distribution of $\widehat{\alpha}_{0t}^{\mbox{\tiny OLS}}$ when $T_0 \rightarrow \infty$, depending on the value of $c \in [0,1)$ such that $J/T_0 \rightarrow c$. Consider a setting with $F=1$, where $y_{it}^N = {\lambda}_t + \epsilon_{it}$ for all $i \in \mathbb{N} \cup \{0\}$, and Assumption  \ref{Assumption_normal_ols} holds with $\sigma_i^2 = \sigma^2$ for all $i$. From Proposition \ref{Convergence_SC} and Corollary \ref{corollary}, we know that the SC estimator converges in distribution to a $N(\alpha_{0t}, \sigma^2)$ in this case. Moreover, from Proposition \ref{Prop_OLS_K_large}, we know that $\widehat{{\mu}}_{\mbox{\tiny OLS}} \buildrel p \over \rightarrow \mu_0 = 1$. In Appendix \ref{Appendix_simple}, we show that $\widehat{\alpha}_{0t}^{\mbox{\tiny OLS}}$ converges in distribution to a $N\left( \alpha_{0t}, {\sigma^2}(1-c)^{-1} \right)$. Therefore, the asymptotic variance of $\widehat{\alpha}_{0t}^{\mbox{\tiny OLS}}$  equals the asymptotic variance of the SC estimator when $J/T_0 \rightarrow 0$. However, if $J/T_0  \rightarrow c>0$,  the asymptotic variance of $\widehat{\alpha}_{0t}^{\mbox{\tiny OLS}}$ is larger than the asymptotic variance of the SC estimator. Moreover, the asymptotic variance of $\widehat{\alpha}_{0t}^{\mbox{\tiny OLS}}$ diverges to infinity when $c \rightarrow 1$. 

Overall, combining the results from Sections \ref{Section_original_SC} and \ref{Section_OLS}, we have that  using an OLS regression to estimate the weights without any regularization method can be a reasonable idea when the number of control units is large, but the number of pre-treatment periods is much larger than the number of controls units. An advantage relative to the original SC estimator is that Assumption \ref{Assumption_mu_ols} requires a sequence of factor loadings that reconstructs $\boldsymbol{\mu}_0$ without the constraints on the weights.  However, an important disadvantage of using the  OLS estimator without any regularization method is that the variance of the estimator may be larger. As we show in our simple example, this cost can be substantial when the number of pre-treatment periods is not much larger than the number of control units.     Including only the adding-up constraint (without the non-negativity constraint) only increases the number of degrees of freedom by one, so all results in this section   remain valid in this case.  When the number of pre-treatment periods is not much larger than the number of control units, other regularization methods could be used, as considered by, for example,  \cite{Doudchenko}, \cite{SDID},  \cite{arco},  \cite{Chernozhukov}, \cite{Hsiao}, and \cite{Li}.

\section{Monte Carlo simulations} \label{MC}

We present a simple Monte Carlo (MC) exercise to illustrate the main results presented in Sections \ref{Section_original_SC} and \ref{Section_OLS}.  We consider a setting in which there are two common factors, $\lambda_{1t}$ and $\lambda_{2t}$. Potential outcomes for the treated unit and for half of the control units depend on the first common factor, so $y_{jt} = \lambda_{1t} + \epsilon_{jt}$ for $j=0,1,...,J/2$, while $y_{jt} = \lambda_{2t} + \epsilon_{jt}$ for $j=J/2+1,...,J$. In this case, $\boldsymbol{\mu}_0 = (\mu_{1,0}, {\mu}_{2,0}) = (1,0)$. Therefore, the goal of the SC method is to set positive weights only to units  $j=1,...,J/2$, which would imply that the asymptotic distribution of $\hat \alpha_{0t}$ for $t \in \mathcal{T}_1$ does not depend on the common factors.    The common factors are normally distributed with a serial correlation equal to  0.5 and variance equal to 1; $\lambda_{1t}$ and $\lambda_{2t}$ are independent. The idiosyncratic shocks $\epsilon_{jt}$ are iid normally distributed with variance equal to 1.

Columns 1 to 4 of  Table \ref{Table_MC} present results for the SC method. Panel A considers a setting with $T_0 = J+5 $, so the number of pre-treatment periods and the number of control units are roughly of the same size. When the number of control units is small ($J=4$ or $J=10$), there is distortion in the proportion of weights allocated to the control units that follow the same common factor as the treated unit. For example, when there are 10 control units, around 82\% of the weights are correctly allocated, while around 18\% of the weights are misallocated. When $J$ and $T_0$ increase, the proportion of misallocated weights goes to zero, which is consistent with Proposition \ref{Convergence_SC}. Interestingly, the standard error of $\widehat{{\boldsymbol{\mu}}}_{\mbox{\tiny SC}}$ goes to zero when $J$ increases, even when $J$ and $T_0$ remain roughly at the same size.   Moreover, the standard error of the treatment effect one period ahead, $\hat \alpha_{01}^{\mbox{\tiny SC}}$, converges to the standard deviation of the idiosyncratic shocks ($\sqrt{var(\epsilon_{jt})}=1$), which is consistent with  Corollary \ref{corollary}. We find similar results when $T_0 = 2 \times J$ (columns 1 to 4, Panel B). 

\begin{center}

[Table \ref{Table_MC} here]

\end{center}

Columns 5 to 8 of  Table \ref{Table_MC} present results using OLS to estimate the weights. In this case, $\mathbb{E}[\hat \mu_{1,0}^{\mbox{\tiny OLS} }] < 1$ when $J$ is small, due to the endogeneity generated by the idiosyncratic shocks of the control units. When $J$ increases, however, $\mathbb{E}[\hat \mu_{1,0}^{\mbox{\tiny OLS} }] \rightarrow 1$, which is consistent with Propositions \ref{Prop_OLS_K_large} and \ref{Prop_OLS_K_finite}. However, differently from the SC weights,  the standard error of $\hat \mu_{1,0}^{\mbox{\tiny OLS} } $  does not go to zero, and remains roughly constant when $J$  increases but $J$ and $T_0$ remains roughly at the same size (Panel A). In contrast, when  $T_0 - J$ increases (Panel B), then the standard error of $\hat \mu_{1,0}^{\mbox{\tiny OLS} } $ goes to zero. The standard error of $\hat \alpha_{01}^{\mbox{\tiny OLS}}$ diverge with $J$ when $T_0 = J+5$. In contrast, it is decreasing with $J$ when   $T_0 = 2 \times J$, although it never reaches the standard error of $\hat \alpha_{01}^{\mbox{\tiny SC}}$. These results are  consistent with the simple example presented in Section \ref{Section_OLS}.  

When weights are estimated with OLS using only the adding-up constraint, results are similar to the unrestricted OLS. The only difference is that  $\mathbb{E}[\hat \mu_{2,0}^{\mbox{\tiny OLS} }] = 0 $ regardless of $J$ when we consider the unrestricted OLS. This happens because $\mu_{2,0} = 0$, so there is no endogeneity problem for this parameter when we consider the unrestricted OLS. In contrast, there is distortion in $\hat \mu_{2,0}$ when we include the restriction that weights should sum one (see columns 9 to 12 of  Table \ref{Table_MC}).

In Appendix \ref{Appendix_covariates}, we present Monte Carlo simulations considering a setting with covariates and different SC specifications.

\section{Conclusion}
\label{conclusion}

We provide conditions under which the SC estimator is asymptotically unbiased when both the number  of pre-treatment periods and the number of control units increase. This will be the case when, as the number of control units goes to infinity, there are weights diluted among an increasing number of control units that  asymptotically recover the factor loadings of the treated unit.  Under this condition, the SC estimator can be asymptotically unbiased even when treatment assignment is correlated with time-varying unobserved confounders. 

We show that the non-negative and adding-up constraints are crucial for this result, as they provide regularization for cases in which the number of parameters to be estimated is larger than the number of pre-treatment periods. Without these constraints, the estimator for the treatment effect remains asymptotically unbiased, but it will generally have a larger  variance, unless the number of pre-treatment periods is much larger than then number of control units. 

Overall, our results extend the set of possible applications in which the SC estimator can be reliably used. While the original SC papers recommend that the method should only be used in applications that present a good pre-treatment fit for a long series of pre-treatment periods, we show that, under some conditions, it can still allow for time-varying unobserved confounders even when the pre-treatment fit is imperfect.  In this case, however, researchers would have to evaluate the plausibility of the conditions we present in this paper. Observing that the SC weights are diluted among a large number of control units in a given application provides supportive evidence that these conditions hold, although it would not be a sufficient condition for the validity of the method. 
In this case, we need  that  possible time-varying unobserved confounders that may be correlated with treatment assignment also affect a large number of control units, so that a weighted average of the control units with diluted weights could absorb such effects. In this case, the SC estimator would be asymptotically unbiased when both the number of pre-treatment periods and the number of control units increase, even in settings where we should not expect to have a good pre-treatment fit.

  \pagebreak

\bigskip
\begin{center}
{\large\bf SUPPLEMENTARY MATERIAL}
\end{center}

\appendix

 \section{Appendix}
 
 \subsection{Proof of main results}

For a generic $m \times n$ matrix $\mathbf{A}$, define $\norm{\mathbf{A}}_2 =\sqrt{ \sum_{p=1}^m \sum_{q=1}^n \left| a_{pq} \right|^2}$, $\norm{\mathbf{A}}_\infty =\underset{p \in \{ 1,...,m\}, q \in \{ 1,...,n\}}{\mbox{max}} \left\{\left| a_{pq} \right| \right\}$, and $\norm{\mathbf{A}}_1 ={ \sum_{p=1}^m \sum_{q=1}^n \left| a_{pq} \right|}$.

 \subsubsection{Proof of Proposition \ref{Convergence_SC}} \label{Proof_Convergence_SC}

 \begin{proposition_b}{\ref{Convergence_SC}} 
 
Suppose we observe $(y_{0t},...,y_{Jt})$ for periods $t \in \{ -T_0+1,...,-1,0,1,...,T_1 \}$, where $J$ is a function of $T_0$. Potential outcomes are defined in equation (\ref{model}). Let $\widehat{{\boldsymbol{\mu}}}_{\mbox{\tiny SC}}$ be defined as ${\mathbf{M}_J}' \widehat{{\mathbf{w}}}_{\mbox{\tiny SC}}$, where $\widehat{{\mathbf{w}}}_{\mbox{\tiny SC}}$ is defined in equation $(\ref{SC_eq})$. Suppose Assumptions \ref{Assumption_e} to \ref{Assumption_lambda}, and Assumption \ref{Assumption_technical}(a) hold. Then, as $T_0 \rightarrow \infty$, (i) $\widehat{{\boldsymbol{\mu}}}_{\mbox{\tiny SC}} \buildrel p \over \rightarrow \boldsymbol{\mu}_0$, and (ii) $ \frac{1}{T_0} \sum_{t \in \mathcal{T}_0} \left( y_{0t} -\widehat{\mathbf{w}}_{\mbox{\tiny SC}} ' {\mathbf{y}}_t  \right)^2  \buildrel p \over \rightarrow \sigma_0^2$. Moreover, if we add Assumption \ref{Assumption_technical}(b), then (iii) $\norm{\widehat{{\mathbf{w}}}_{\mbox{\tiny SC}}}_2 \buildrel p \over \rightarrow 0$.

\end{proposition_b}

 \begin{proof}

 We start by extending the function $\mathcal{H}_{J}(\boldsymbol{\mu})$ to the domain $  \mathcal{M} = \mbox{cl} \left( \cup_{J \in \mathbb{N}} \mathcal{M}_J \right)$, where  $\mbox{cl} \left(\mathcal{A}\right)$ is the closure of set $\mathcal{A}$. For $\boldsymbol{\mu} \in \mathcal{M}$, we define
\begin{eqnarray} \label{Definition_H_tilde}
\widetilde{\mathcal{H}}_{J} (\boldsymbol{\mu})  =  \underset{ \widetilde{\boldsymbol{\mu}} \in \mathcal{M}_J }{\mbox{min}} \left\{  \mathcal{H}_J (\widetilde{\boldsymbol{\mu}}) + K \norm{\boldsymbol{\mu} - \widetilde{\boldsymbol{\mu}}}_2  \right\},
\end{eqnarray} 
where we define later in the proof what $K$ is. For now, it suffices to consider that $K>0$ almost surely  and that  $K=O_p(1)$.  Since for any $\boldsymbol{\mu} \in \mathcal{M} \setminus\mathcal{M}_J$ there is an $\boldsymbol{\mu}' \in \mathcal{M}_J$ such that $\widetilde{\mathcal{H}}_{J} ( \boldsymbol{\mu}') < \widetilde{\mathcal{H}}_{J} (\boldsymbol{\mu}) $, it follows that $\widehat{{\boldsymbol{\mu}}}_{\mbox{\tiny SC}} =  {\mathbf{M}_J}' {\widehat{\mathbf{w}}}_{\mbox{\tiny SC}} = \underset{ \boldsymbol{\mu} \in \mathcal{M}_J  }{\mbox{argmin}} \mathcal{H}_{J}(\boldsymbol{\mu}) = \underset{ \boldsymbol{\mu} \in \mathcal{M}  }{\mbox{argmin}} \widetilde{\mathcal{H}}_{J}(\boldsymbol{\mu})$.\footnote{Even if the minimization problem presented in equation \ref{eq_objective} has multiple solutions, we would still have that the set of solutions of the two minimization problems would be the same.  } Therefore, we can analyze the behavior of the implied estimator for the factor loadings of the treated unit considering the objective function $\widetilde{\mathcal{H}}_{J} (\boldsymbol{\mu})  $.  

Define the function $\sigma_{\bar \lambda}^2(\boldsymbol{\mu}) \equiv  (\boldsymbol{\mu}_0 - \boldsymbol{\mu})' \boldsymbol{\Omega} (\boldsymbol{\mu}_0 - \boldsymbol{\mu}) + \sigma^2_0$, where $\boldsymbol{\Omega}$ is positive definite by Assumption \ref{Assumption_lambda}. Therefore, $\sigma_{\bar \lambda}^2(\boldsymbol{\mu}) $ is  uniquely minimized at $\boldsymbol{\mu}_0$.  We first show that $\widetilde{\mathcal{H}}_{J} (\boldsymbol{\mu}_0) $ is bounded from above by a term that converges in probability to  $\sigma_{ \bar \lambda}^2( \boldsymbol{\mu}_0)$. Consider the $\mathbf{w}^\ast_J$ defined in Assumption \ref{Assumption_mu}, and let $\boldsymbol{\mu}_J^\ast \equiv  {\mathbf{M}_J}'\mathbf{w}^\ast_J$. Then, since $\boldsymbol{\mu}_J^\ast  \in \mathcal{M}_J$, and since  $\mathbf{w}^\ast_J$ is a candidate solution for the minimization problem defined in $ \mathcal{H}_J (\boldsymbol{\mu}^\ast_J)$, it follows that
\begin{eqnarray}
\widetilde{\mathcal{H}}_{J} (\boldsymbol{\mu}_0) & \leq &\mathcal{H}_J(\boldsymbol{\mu}^\ast_J) + K \norm{\boldsymbol{\mu}_0-\boldsymbol{\mu}_J^\ast}_2 \leq \frac{1}{T_0} \sum_{t \in \mathcal{T}_0} \left( \bar {\lambda}_t(\boldsymbol{\mu}^\ast_J) - ({\boldsymbol{\epsilon}}_t'\mathbf{w}^\ast_J)  \right)^2 + K \norm{\boldsymbol{\mu}_0-\boldsymbol{\mu}_J^\ast}_2 \\ \nonumber
&=& \frac{1}{T_0} \sum_{t \in \mathcal{T}_0} \bar {\lambda}_t(\boldsymbol{\mu}^\ast_J)^2 + \frac{1}{T_0} \sum_{t \in \mathcal{T}_0} ({\boldsymbol{\epsilon}}_t'\mathbf{w}^\ast_J)^2 - 2\frac{1}{T_0} \sum_{t \in \mathcal{T}_0} \bar {\lambda}_t(\boldsymbol{\mu}^\ast_J) ({\boldsymbol{\epsilon}}_t'\mathbf{w}^\ast_J)  +  K \norm{\boldsymbol{\mu}_0-\boldsymbol{\mu}_J^\ast}_2 \\ \nonumber
&\equiv& \widetilde{\mathcal{H}}_{J}^{UB} (\boldsymbol{\mu}_0).
\end{eqnarray} 

The first term of $\widetilde{\mathcal{H}}_{J}^{UB} (\boldsymbol{\mu}_0)$ equals to 
\begin{eqnarray}
 \frac{1}{T_0} \sum_{t \in \mathcal{T}_0} \bar {\lambda}_t(\boldsymbol{\mu}^\ast_J)^2 =   (\boldsymbol{\mu}_0 - \boldsymbol{\mu}_J^\ast)' \left(\frac{1}{T_0} \sum_{t \in \mathcal{T}_0} \boldsymbol{\lambda}_t' \boldsymbol{\lambda}_t \right)(\boldsymbol{\mu}_0 - \boldsymbol{\mu}_J^\ast) + 2 (\boldsymbol{\mu}_0 - \boldsymbol{\mu}_J^\ast) ' \frac{1}{T_0} \sum_{t \in \mathcal{T}_0} \boldsymbol{\lambda}_t' \epsilon_{0t} + \frac{1}{T_0} \sum_{t \in \mathcal{T}_0}  \epsilon^2_{0t},
\end{eqnarray} 
where $(\boldsymbol{\mu}_0 - \boldsymbol{\mu}_J^\ast)' \left(\frac{1}{T_0} \sum_{t \in \mathcal{T}_0} \boldsymbol{\lambda}_t' \boldsymbol{\lambda}_t \right)(\boldsymbol{\mu}_0 - \boldsymbol{\mu}_J^\ast) = o_p(1)$ because $(\boldsymbol{\mu}_0 - \boldsymbol{\mu}_J^\ast) = o(1)$ and $\left(\frac{1}{T_0} \sum_{t \in \mathcal{T}_0} \boldsymbol{\lambda}_t' \boldsymbol{\lambda}_t \right)= O(1)$, $\left| (\boldsymbol{\mu}_0 - \boldsymbol{\mu}_J^\ast) ' \frac{1}{T_0} \sum_{t \in \mathcal{T}_0} \boldsymbol{\lambda}_t' \epsilon_{0t}\right| \leq \norm{\boldsymbol{\mu}_0 - \boldsymbol{\mu}_J^\ast}_1 \norm{\frac{1}{T_0} \sum_{t \in \mathcal{T}_0} \boldsymbol{\lambda}_t' \epsilon_{0t}}_\infty$, where $\norm{\boldsymbol{\mu}_0 - \boldsymbol{\mu}_J^\ast}_1=o(1)$ from Assumption \ref{Assumption_mu} and  $\norm{\frac{1}{T_0} \sum_{t \in \mathcal{T}_0} \boldsymbol{\lambda}_t' \epsilon_{0t}}_\infty = o_p(1)$ from Assumption \ref{Assumption_technical}(a), and $\frac{1}{T_0} \sum_{t \in \mathcal{T}_0}  \epsilon^2_{0t}  \buildrel p \over \rightarrow \sigma^2_0$. Therefore, $ \frac{1}{T_0} \sum_{t \in \mathcal{T}_0} \bar {\lambda}_t(\boldsymbol{\mu}^\ast_J)^2  \buildrel p \over \rightarrow \sigma_{\bar \lambda}^2(\boldsymbol{\mu}_0)$.  

For the term $\frac{1}{T_0} \sum_{t \in \mathcal{T}_0} ({\boldsymbol{\epsilon}}_t'\mathbf{w}^\ast_J)^2$, note that  $var({\boldsymbol{\epsilon}}_t'\mathbf{w}^\ast_J) \leq \bar \gamma \norm{\mathbf{w}^\ast_J}_2^2$, where $\bar \gamma = \mbox{sup}_{j,t}  \left \{var(\epsilon_{jt}) \right\} < \infty$ since $\epsilon_{jt}$ has uniformly bounded fourth moments. Therefore,
\begin{eqnarray} \label{eq_w_star}
\frac{1}{T_0} \sum_{t \in \mathcal{T}_0} ({\boldsymbol{\epsilon}}_t'\mathbf{w}^\ast_J)^2 \leq \bar \gamma \norm{\mathbf{w}^\ast_J}_2^2 \frac{1}{T_0} \sum_{t \in \mathcal{T}_0} \left( \frac{{\boldsymbol{\epsilon}}_t'\mathbf{w}^\ast_J}{\sqrt{var({\boldsymbol{\epsilon}}_t'\mathbf{w}^\ast_J)}} \right)^2.
\end{eqnarray} 

If we define $z_t \equiv \left( \frac{{\boldsymbol{\epsilon}}_t'\mathbf{w}^\ast_J}{\sqrt{var({\boldsymbol{\epsilon}}_t'\mathbf{w}^\ast_J)}} \right)^2 - 1$, then $\mathbb{E}[z_t] = 0 $. Moreover,
\begin{eqnarray*}
var(z_t) &=& \mathbb{E}\left[\left( \frac{{\boldsymbol{\epsilon}}_t'\mathbf{w}^\ast_J}{\sqrt{var({\boldsymbol{\epsilon}}_t'\mathbf{w}^\ast_J)}} \right)^4 \right] - 2 \mathbb{E}\left[\left( \frac{{\boldsymbol{\epsilon}}_t'\mathbf{w}^\ast_J}{\sqrt{var({\boldsymbol{\epsilon}}_t'\mathbf{w}^\ast_J)}} \right)^2 \right] + 1 =\frac{\mathbb{E}[({\boldsymbol{\epsilon}}_t'\mathbf{w}^\ast_J)^4]}{(var({\boldsymbol{\epsilon}}_t'\mathbf{w}^\ast_J))^2} - 1.
\end{eqnarray*}

Now note that $\mathbb{E}[({\boldsymbol{\epsilon}}_t'\mathbf{w}^\ast_J)^4] = \sum_{p,q,r,s} \mathbb{E}[\epsilon_{pt}\epsilon_{qt}\epsilon_{rt}\epsilon_{st}]w^\ast_{p}w^\ast_{q}w^\ast_{r}w^\ast_{s} = \sum_{p,q} \mathbb{E}[\epsilon^2_{pt}\epsilon^2_{qt}](w^\ast_{p})^2(w^\ast_{q})^2$, where the last equality follows from $\epsilon_{it}$ independent across $i$, and $\mathbb{E}[\epsilon_{it}]=0$. Now given that $\epsilon_{it}$ has uniformly bounded fourth moments across $i$ and $t$, we can define $\bar \xi = \mbox{sup}_{i,t} \{\mathbb{E}[\epsilon^4_{it}] \} < \infty$. It follows that $\mathbb{E}[({\boldsymbol{\epsilon}}_t'\mathbf{w}^\ast_J)^4] \leq \mbox{max} \left\{  \bar \xi , \bar \gamma^2 \right\}  \sum_{p,q}(w^\ast_{p})^2(w^\ast_{q})^2 = \mbox{max} \left\{ \bar \xi , \bar \gamma^2 \right\}  \norm{\mathbf{w}^\ast_J}_2^4$. For the denominator, if we define $\underline \gamma = \mbox{inf}_{j,t}  \left \{var(\epsilon_{jt}) \right\} >0$, then $(var({\boldsymbol{\epsilon}}_t'\mathbf{w}^\ast_J))^2 \geq \underline \gamma^2 \norm{\mathbf{w}^\ast_J}_2^4$. Combining these two results, we have that $var(z_t)$ is uniformly bounded. It follows  from \cite{Andrews1988} that $\frac{1}{T_0} \sum_{t \in \mathcal{T}_0} z_t  \buildrel p \over \rightarrow 0$, which implies that $\frac{1}{T_0} \sum_{t \in \mathcal{T}_0} \left( \frac{{\boldsymbol{\epsilon}}_t'\mathbf{w}^\ast_J}{\sqrt{var({\boldsymbol{\epsilon}}_t'\mathbf{w}^\ast_J)}} \right)^2  \buildrel p \over \rightarrow 1$.  Since $\norm{\mathbf{w}^\ast_J}_2^2 \rightarrow 0 $ by Assumption \ref{Assumption_mu}, it follows that this term is $o_p(1)$.  

The term $\frac{1}{T_0} \sum_{t \in \mathcal{T}_0} \bar {\lambda}_t(\boldsymbol{\mu}^\ast_J) ({\boldsymbol{\epsilon}}_t'\mathbf{w}^\ast_J)$  is given by $(\boldsymbol{\mu}_0 - \boldsymbol{\mu}^\ast_J)' \frac{1}{T_0} \sum_{t \in \mathcal{T}_0} \boldsymbol{\lambda}_t'  ({\boldsymbol{\epsilon}}_t'\mathbf{w}^\ast_J) + \frac{1}{T_0} \sum_{t \in \mathcal{T}_0} \epsilon_{0t}  ({\boldsymbol{\epsilon}}_t'\mathbf{w}^\ast_J)$. Note that
\begin{eqnarray}
\norm{\frac{1}{T_0} \sum_{t \in \mathcal{T}_0} \boldsymbol{\lambda}_t'  ({\boldsymbol{\epsilon}}_t'\mathbf{w}^\ast_J)}_1 = \sum_{f=1}^F \left| \frac{1}{T_0} \sum_{t \in \mathcal{T}_0} {\lambda}_{ft}  ({\boldsymbol{\epsilon}}_t'\mathbf{w}^\ast_J) \right|  \leq  \sum_{f=1}^F \norm{\frac{1}{T_0} \sum_{t \in \mathcal{T}_0} {\lambda}_{ft}  {\boldsymbol{\epsilon}}_t}_\infty \norm{\mathbf{w}^\ast_J}_1 \leq \sum_{f=1}^F\norm{\frac{1}{T_0} \sum_{t \in \mathcal{T}_0} \lambda_{ft}  {\boldsymbol{\epsilon}}_t}_\infty,
\end{eqnarray}
where the first inequality follows from H\"{o}lder's inequality, and the second one follows from $\mathbf{w}^\ast_J \in \Delta^{J-1}$. From Assumption \ref{Assumption_technical}(a), we have  $ \norm{\frac{1}{T_0} \sum_{t \in \mathcal{T}_0} \lambda_{ft}  {\boldsymbol{\epsilon}}_t}_\infty = o_p(1)$. Since ${\boldsymbol{\mu}_0 - \boldsymbol{\mu}^\ast_J}$ is bounded, it follows that $(\boldsymbol{\mu}_0 - \boldsymbol{\mu}^\ast_J)' \frac{1}{T_0} \sum_{t \in \mathcal{T}_0} \boldsymbol{\lambda}_t'  ({\boldsymbol{\epsilon}}_t'\mathbf{w}^\ast_J)=o_p(1)$. Also, note that $\epsilon_{0t}  ({\boldsymbol{\epsilon}}_t'\mathbf{w}^\ast_J)$ has zero mean and uniformly bounded variance, given that $\epsilon_{0t} $ and $ ({\boldsymbol{\epsilon}}_t'\mathbf{w}^\ast_J)$ have mean zero and uniformly bounded variance, and they are independent. Since we also have that $\epsilon_{0t}  ({\boldsymbol{\epsilon}}_t'\mathbf{w}^\ast_J)$ is $\alpha$-mixing, it follows from \cite{Andrews1988} that  $\frac{1}{T_0} \sum_{t \in \mathcal{T}_0} \epsilon_{0t}  ({\boldsymbol{\epsilon}}_t'\mathbf{w}^\ast_J)=o_p(1)$.  Finally, the fourth term is $o_p(1)$ because $K=O_p(1)$ and $\norm{\boldsymbol{\mu}_0-\boldsymbol{\mu}_J^\ast}_2 = o(1)$. Therefore, $\widetilde{\mathcal{H}}_{J} (\boldsymbol{\mu}_0)  \leq  \widetilde{\mathcal{H}}_{J}^{UB} (\boldsymbol{\mu}_0)  \buildrel p \over \rightarrow \sigma_{\bar \lambda}^2(\boldsymbol{\mu}_0)$.

We show next that $\widetilde{\mathcal{H}}_{J} (\boldsymbol{\mu}) $ is bounded from below by a function that converges uniformly in probability to  $\sigma_{\bar \lambda}^2(\boldsymbol{\mu})$. We have that 
\begin{eqnarray}
\widetilde{\mathcal{H}}_{J} (\boldsymbol{\mu})  \geq  \underset{ \widetilde{\boldsymbol{\mu}} \in \mathcal{M} }{\mbox{min}} \left\{ {\underset{{\mathbf{b}}  \in \mathcal{W}}{\mbox{min}} \left \{   \frac{1}{T_0} \sum_{t \in \mathcal{T}_0} ( \bar {\lambda}_t(\widetilde{\boldsymbol{\mu}}) -   {\mathbf{b}} ' {\boldsymbol{\epsilon}}_t )  ^2 \right \}} + K \norm{\boldsymbol{\mu} - \widetilde{\boldsymbol{\mu}}}_2  \right\} \equiv \widetilde{\mathcal{H}}_{J}^{LB} (\boldsymbol{\mu}),
\end{eqnarray} 
where $\mathcal{W} = \{ \mathbf{w} \in \mathbb{R}^J | \norm{\mathbf{w}}_1 \leq 1  \}$. This inequality holds because we are relaxing three  constraints in the minimization problems from  $\widetilde{\mathcal{H}}_{J} ({\boldsymbol{\mu}})$. We consider $\widetilde{\boldsymbol{\mu}} \in \mathcal{M} \supset \mathcal{M}_J$, we relax the condition ${\mathbf{M}_J}'{\mathbf{w}} = \widetilde{\boldsymbol{\mu}}$, and we consider a set $\mathcal{W} \supset \Delta^{J-1}$.  

We first show that the function $Q_J(\boldsymbol{\mu}) \equiv {\underset{{\mathbf{b}}  \in \mathcal{W}}{\mbox{min}} \left \{   \frac{1}{T_0} \sum_{t \in \mathcal{T}_0} ( \bar {\lambda}_t(\boldsymbol{\mu}) -   {\mathbf{b}} ' {\boldsymbol{\epsilon}}_t )  ^2 \right \}}$ is Lipschitz with a constant that is $O_p(1)$. Let ${\widehat{\mathbf{b}}}(\boldsymbol{\mu})$ be the solution to this minimization problem for a given $\boldsymbol{\mu}$. Since $\mathcal{M}$ is convex, from the mean value and the envelope theorems, it follows that, for any $\boldsymbol{\mu}$ and $\boldsymbol{\mu}' \in \mathcal{M}$, there is a $\widetilde{\boldsymbol{\mu}} \in \mathcal{M}$ such that
\begin{eqnarray} \label{eq_lipschitz}
|Q_J (\boldsymbol{\mu}) - Q_J (\boldsymbol{\mu}') | &=& \left|  \frac{2}{T_0} \sum_{t \in \mathcal{T}_0} \boldsymbol{\lambda}_t' [\boldsymbol{\lambda}_t(\boldsymbol{\mu}_0 - \widetilde{\boldsymbol{\mu}}) + \epsilon_{0t} -   {\widehat{\mathbf{b}}(\widetilde{\boldsymbol{\mu}})} ' {\boldsymbol{\epsilon}}_t ]  \cdot (\boldsymbol{\mu} - \boldsymbol{\mu}') \right|\\ \nonumber
&\leq& \norm{\frac{2}{T_0} \sum_{t \in \mathcal{T}_0} \boldsymbol{\lambda}_t' [\boldsymbol{\lambda}_t(\boldsymbol{\mu}_0 - \widetilde{\boldsymbol{\mu}}) + \epsilon_{0t} -   {\widehat{\mathbf{b}}(\widetilde{\boldsymbol{\mu}})} ' {\boldsymbol{\epsilon}}_t ] }_2 \times \norm{\boldsymbol{\mu} - \boldsymbol{\mu}'}_2 \\ \nonumber
&\leq&  \left(\norm{ \sum_{t \in \mathcal{T}_0} \frac{\boldsymbol{\lambda}_t' \boldsymbol{\lambda}_t}{T_0}}_2 \norm{\boldsymbol{\mu}_0 - \widetilde{\boldsymbol{\mu}}}_2 +\norm{\sum_{t \in \mathcal{T}_0} \frac{\boldsymbol{\lambda}_t' \epsilon_{0t}}{T_0}}_2 +  \norm{\sum_{t \in \mathcal{T}_0} \frac{ \boldsymbol{\lambda}_t' {\boldsymbol{\epsilon}}_t' }{T_0} \widehat{\mathbf{b}}(\widetilde{\boldsymbol{\mu}})}_2  \right) \times \norm{\boldsymbol{\mu} - \boldsymbol{\mu}'}_2. 
 \end{eqnarray}

Since $\norm{\boldsymbol{\mu}_0 - \widetilde{\boldsymbol{\mu}}}_2 \leq C$ for some constant $C$  (Assumption \ref{Assumption_mu}), and $ \sum_{t \in \mathcal{T}_0} \frac{\boldsymbol{\lambda}_t' \boldsymbol{\lambda}_t}{T_0} \rightarrow \boldsymbol{\Omega}$ (Assumption \ref{Assumption_lambda}), we have that  $\norm{ \sum_{t \in \mathcal{T}_0} \frac{\boldsymbol{\lambda}_t' \boldsymbol{\lambda}_t}{T_0}}_2 \norm{\boldsymbol{\mu}_0 - \widetilde{\boldsymbol{\mu}}}_2 \leq C \norm{ \sum_{t \in \mathcal{T}_0} \frac{\boldsymbol{\lambda}_t' \boldsymbol{\lambda}_t}{T_0}}_2  = O_p(1)$. From Assumption \ref{Assumption_technical}(a),   $\norm{\sum_{t \in \mathcal{T}_0} \frac{\boldsymbol{\lambda}_t' \epsilon_{0t}}{T_0}}_2=o_p(1)$. Finally,  $\norm{\sum_{t \in \mathcal{T}_0} \frac{ \boldsymbol{\lambda}_t' {\boldsymbol{\epsilon}}_t' }{T_0} \widehat{\mathbf{b}}(\widetilde{\boldsymbol{\mu}})}_2 \leq \sum_{f=1}^F \left| \sum_{t \in \mathcal{T}_0} \frac{ \lambda_{ft} {\boldsymbol{\epsilon}}_t' }{T_0} \widehat{\mathbf{b}}(\widetilde{\boldsymbol{\mu}}) \right| \leq \sum_{f=1}^F \norm{\sum_{t \in \mathcal{T}_0} \frac{ \lambda_{ft} {\boldsymbol{\epsilon}}_t }{T_0}}_\infty \norm{\widehat{\mathbf{b}}(\widetilde{\boldsymbol{\mu}})}_1 \leq \sum_{f=1}^F \norm{\sum_{t \in \mathcal{T}_0} \frac{ \lambda_{ft} {\boldsymbol{\epsilon}}_t }{T_0}}_\infty = o_p(1)$ from Assumption \ref{Assumption_technical}(a). Therefore, $|Q_J (\boldsymbol{\mu}) - Q_J (\boldsymbol{\mu}') | \leq \tilde K \norm{\boldsymbol{\mu} - \boldsymbol{\mu}'}_2$, where $\tilde K=O_p(1)$ and does not depend on $\boldsymbol{\mu}$ and $\boldsymbol{\mu}'$. We define $K$ used in  function $\widetilde{\mathcal{H}}_{J} (\boldsymbol{\mu}) $ as $K = 1+\tilde K$, so $K>0$ and $K=O_p(1)$. Given that $K$ is greater than the Lipschitz constant of $Q_J (\boldsymbol{\mu})$, we  have that  $\widetilde{\mathcal{H}}_{J}^{LB} (\boldsymbol{\mu})  = Q_J (\boldsymbol{\mu})$ for all $\boldsymbol{\mu} \in \mathcal{M}$. Therefore, $\widetilde{\mathcal{H}}_{J}^{LB} (\boldsymbol{\mu}) $ is Lipschitz with a constant $O_p(1)$.

Now we show that $\widetilde{\mathcal{H}}_{J}^{LB} (\boldsymbol{\mu})  \buildrel p \over \rightarrow \sigma_{\bar \lambda}^2(\boldsymbol{\mu})$ pointwise.  Note that $ \frac{1}{T_0} \sum_{t \in \mathcal{T}_0} ( \bar {\lambda}_t(\boldsymbol{\mu}) -   {\mathbf{b}} ' {\boldsymbol{\epsilon}}_t )^2 =  \frac{1}{T_0} \sum_{t \in \mathcal{T}_0} ( \boldsymbol{\lambda}_t(\boldsymbol{\mu}_0 - \boldsymbol{\mu}) )^2 +  \frac{2}{T_0} \sum_{t \in \mathcal{T}_0}( \boldsymbol{\lambda}_t(\boldsymbol{\mu}_0 - \boldsymbol{\mu}) ) (\epsilon_{0t} -  {\mathbf{b}} ' {\boldsymbol{\epsilon}}_t ) +  \frac{1}{T_0} \sum_{t \in \mathcal{T}_0} (\epsilon_{0t} -  {\mathbf{b}} ' {\boldsymbol{\epsilon}}_t )^2 $. We define  $\widehat{\mathbf{c}} \in \underset{\mathbf{b} \in \mathcal{W}}{\mbox{argmin}} \left\{\frac{1}{T_0} \sum_{t \in \mathcal{T}_0} (\epsilon_{0t} -  {\mathbf{b}} ' {\boldsymbol{\epsilon}}_t )^2\right\}$. Since $\mathcal{W}$ is compact, we can also define $\widehat{\mathbf{d}} \in \underset{\mathbf{b} \in \mathcal{W}}{\mbox{argmin}} \left\{\frac{2}{T_0} \sum_{t \in \mathcal{T}_0}( \boldsymbol{\lambda}_t(\boldsymbol{\mu}_0 - \boldsymbol{\mu}) ) (\epsilon_{0t} -  {\mathbf{b}} ' {\boldsymbol{\epsilon}}_t )\right\}$. Therefore,
\begin{eqnarray} \label{Appendix_inequality}
\frac{1}{T_0} \sum_{t \in \mathcal{T}_0} ( \boldsymbol{\lambda}_t(\boldsymbol{\mu}_0 - \boldsymbol{\mu}) )^2 +  \frac{2}{T_0} \sum_{t \in \mathcal{T}_0}( \boldsymbol{\lambda}_t(\boldsymbol{\mu}_0 - \boldsymbol{\mu}) ) (\epsilon_{0t} -  {\widehat{\mathbf{d}}} ' {\boldsymbol{\epsilon}}_t ) +  \frac{1}{T_0} \sum_{t \in \mathcal{T}_0} (\epsilon_{0t} -  {\widehat{\mathbf{c}}} ' {\boldsymbol{\epsilon}}_t )^2 \leq Q_J(\boldsymbol{\mu}) \leq \\
\leq \frac{1}{T_0} \sum_{t \in \mathcal{T}_0} ( \boldsymbol{\lambda}_t(\boldsymbol{\mu}_0 - \boldsymbol{\mu}) )^2 +  \frac{2}{T_0} \sum_{t \in \mathcal{T}_0}( \boldsymbol{\lambda}_t(\boldsymbol{\mu}_0 - \boldsymbol{\mu}) ) (\epsilon_{0t} -  {\widehat{\mathbf{c}}} ' {\boldsymbol{\epsilon}}_t ) +  \frac{1}{T_0} \sum_{t \in \mathcal{T}_0} (\epsilon_{0t} -  {\widehat{\mathbf{c}}} ' {\boldsymbol{\epsilon}}_t )^2,
\end{eqnarray}
where the first inequality holds because we give more flexibility in the minimization problem in $Q_J(\boldsymbol{\mu})$ by allowing different parameters to minimize the second and the third terms. The second inequality holds because $\widehat{\mathbf{c}}$ is a candidate solution to the minimization problem in $Q_J(\boldsymbol{\mu})$. Note that $\frac{1}{T_0} \sum_{t \in \mathcal{T}_0} ( \boldsymbol{\lambda}_t(\boldsymbol{\mu}_0 - \boldsymbol{\mu}) )^2 \rightarrow (\boldsymbol{\mu}_0 -\boldsymbol{\mu})' \boldsymbol{\Omega} (\boldsymbol{\mu}_0 - \boldsymbol{\mu})$ and $ \frac{2}{T_0} \sum_{t \in \mathcal{T}_0}( \boldsymbol{\lambda}_t(\boldsymbol{\mu}_0 - \boldsymbol{\mu}) ) (\epsilon_{0t} -  {{\mathbf{b}}} ' {\boldsymbol{\epsilon}}_t )=o_p(1)$ regardless of ${\mathbf{b}}$. Therefore, we only have to show that $\frac{1}{T_0} \sum_{t \in \mathcal{T}_0} (\epsilon_{0t} -  {\widehat{\mathbf{c}}} ' {\boldsymbol{\epsilon}}_t )^2  \buildrel p \over \rightarrow \sigma^2_0$ to conclude that $Q_J(\boldsymbol{\mu})  \buildrel p \over \rightarrow  \sigma^2_{\bar \lambda}(\boldsymbol{\mu})$.

We essentially apply Lemma 2  from \cite{Chernozhukov}.  Since $\widehat{\mathbf{c}} \in \mathcal{W}$ is the argmin of $ \left\{\frac{1}{T_0} \sum_{t \in \mathcal{T}_0} (\epsilon_{0t} -  {\mathbf{b}} ' {\boldsymbol{\epsilon}}_t )^2\right\}$, and $0 \in \mathcal{W}$,  it follows that  $\norm{\boldsymbol{\epsilon}_0-{\mathbf{E}\widehat{\mathbf{c}}}}_2^2 \leq \norm{\boldsymbol{\epsilon}_0}_2^2$, where $\mathbf{E}$ is the $T_0 \times J$ matrix with information on $\epsilon_{jt}$ for all $j=1,...,J$ and $t \in \mathcal{T}_0$, and $\boldsymbol{\epsilon}_0$ is the $T_0$ vector with information on $\epsilon_{0t}$ for all $t \in \mathcal{T}_0$.   Therefore, equivalently to equation 40 from \cite{Chernozhukov}, we have the inequality
\begin{eqnarray} \label{inequality_chernozhukov}
\frac{1}{T_0} \norm{{\mathbf{E}\widehat{\mathbf{c}}}}_2^2 \leq  \frac{1}{T_0} 2 {\boldsymbol{\epsilon}_0'{\mathbf{E}\widehat{\mathbf{c}}}} \leq \frac{1}{T_0}  2 \norm{\mathbf{E}'\boldsymbol{\epsilon}_0 }_\infty \norm{\widehat{\mathbf{c}}}_1 \leq \frac{1}{T_0}  2 \norm{\mathbf{E}'\boldsymbol{\epsilon}_0 }_\infty.
\end{eqnarray}

From Assumption \ref{Assumption_technical}(a), $\frac{1}{T_0}  2 \norm{\mathbf{E}'\boldsymbol{\epsilon}_0 }_\infty=o_p(1)$, which implies that $\frac{1}{T_0} \norm{{\mathbf{E}\widehat{\mathbf{c}}}}_2^2=o_p(1)$ and $ \frac{1}{T_0} 2 {\boldsymbol{\epsilon}_0'{\mathbf{E}\widehat{\mathbf{c}}}}=o_p(1)$. Therefore, it follows that  $\widetilde{\mathcal{H}}^{LB}_{T_0}(\boldsymbol{\mu})  \buildrel p \over \rightarrow \sigma_{ \bar \lambda}^2(\boldsymbol{\mu})$.   Since $\sigma_{\bar \lambda}^2(\boldsymbol{\mu})$ is continuous and $\mathcal{M}$ is compact, then,  based on  Corollary 2.2 of \cite{Newey}, we have that $ \mathcal{H}^{LB}_{T_0}(\boldsymbol{\mu})$ converges uniformly in probability to $\sigma_{\bar \lambda}^2(\boldsymbol{\mu})$.

Combining the results above that (i) $\widetilde{\mathcal{H}}_{J} (\boldsymbol{\mu}_0)  \leq  \widetilde{\mathcal{H}}_{J}^{UB} (\boldsymbol{\mu}_0)  \buildrel p \over \rightarrow \sigma_{\bar \lambda}^2(\boldsymbol{\mu}_0)$, and (ii) $\widetilde{\mathcal{H}}_{J} (\boldsymbol{\mu})  \geq  \widetilde{\mathcal{H}}_{J}^{LB} (\boldsymbol{\mu})  \buildrel p \over \rightarrow \sigma_{\bar \lambda}^2(\boldsymbol{\mu})$ uniformly in $\boldsymbol{\mu} \in \mathcal{M}$, we show that $\widehat{\boldsymbol{\mu}}_{\mbox{\tiny SC}}  \buildrel p \over \rightarrow \boldsymbol{\mu}_0$. This is a simple extension of Theorem 2.1 from \cite{NM}. For a given $\eta>0$, since $\widehat{\boldsymbol{\mu}}_{\mbox{\tiny SC}} = \underset{ \boldsymbol{\mu} \in \mathcal{M}  }{\mbox{argmin}} \widetilde{\mathcal{H}}_{J}(\boldsymbol{\mu})$, $\widetilde{\mathcal{H}}_{J} (\widehat{\boldsymbol{\mu}}_{\mbox{\tiny SC}}) < \widetilde{\mathcal{H}}_{J} (\boldsymbol{\mu}_0) + \frac{\eta}{3}$ with probability approaching one (wpa1). From (ii), we have that  $\sigma_{\bar \lambda}^2(\boldsymbol{\mu}) <  \widetilde{\mathcal{H}}_{J}^{LB} (\boldsymbol{\mu}) + \frac{\eta}{3} \leq  \widetilde{\mathcal{H}}_{J} (\boldsymbol{\mu}) + \frac{\eta}{3}$ for all $\boldsymbol{\mu}$ wpa1. Combining (i) and (ii), we have that $ \widetilde{\mathcal{H}}_{J} (\boldsymbol{\mu}_0)  \buildrel p \over \rightarrow \sigma_{\bar \lambda}^2(\boldsymbol{\mu}_0)$, which implies that $ \widetilde{\mathcal{H}}_{J} (\boldsymbol{\mu}_0) < \sigma_{\bar \lambda}^2(\boldsymbol{\mu}_0) + \frac{\eta}{3}$ wpa1. Combining these three inequalities, we have that $\sigma_{\bar \lambda}^2(\widehat{\boldsymbol{\mu}}_{\mbox{\tiny SC}}) < \sigma_{\bar \lambda}^2(\boldsymbol{\mu}_0) +\eta$ wpa1. Now let $\mathcal{V}$ be any open subset of $\mathcal{M}$ containing $\boldsymbol{\mu}_0$. Since $\mathcal{M} \cap \mathcal{V}^C$ is compact, $\boldsymbol{\mu}_0 =  \underset{ \boldsymbol{\mu} \in \mathcal{M}  }{\mbox{argmin}} \sigma_{\bar \lambda}^2(\boldsymbol{\mu})$, and $\sigma_{\bar \lambda}^2(\boldsymbol{\mu})$ is continuous, then $\underset{ \boldsymbol{\mu} \in \mathcal{M}\cap \mathcal{V}^C  }{\mbox{inf}} \sigma_{\bar \lambda}^2(\boldsymbol{\mu}) = \sigma_{\bar \lambda}^2(\boldsymbol{\mu}^\ast) >\sigma_{\bar \lambda}^2(\boldsymbol{\mu}_0)$ for some $\boldsymbol{\mu}^\ast \in  \mathcal{M}\cap \mathcal{V}^C $. Let $\eta = \sigma_{\bar \lambda}^2(\boldsymbol{\mu}^\ast) - \sigma_{\bar \lambda}^2(\boldsymbol{\mu}_0)$. Then, wpa1, $\sigma_{\bar \lambda}^2(\widehat{\boldsymbol{\mu}}_{\mbox{\tiny SC}}) < \sigma_{\bar \lambda}^2(\boldsymbol{\mu}^\ast)$, which implies that $\widehat{\boldsymbol{\mu}}_{\mbox{\tiny SC}} \in \mathcal{V}$. Therefore, $\widehat{\boldsymbol{\mu}}_{\mbox{\tiny SC}}  \buildrel p \over \rightarrow \boldsymbol{\mu}_0$.

Now we prove the second result from Proposition \ref{Convergence_SC}, that  $\frac{1}{T_0} \sum_{t \in \mathcal{T}_0} \left( y_{0t} -\widehat{\mathbf{w}}_{\mbox{\tiny SC}} ' {\mathbf{y}}_t  \right)^2  \buildrel p \over \rightarrow  \sigma_{\bar \lambda}^2(\boldsymbol{\mu}_0) = \sigma^2_0$. Since $\boldsymbol{\mu}_J^\ast  \in \mathcal{M}_J$, and since  $\mathbf{w}^\ast_J$ is a candidate solution for the minimization problem defined in $ \mathcal{H}_J (\boldsymbol{\mu}^\ast_J)$, it follows that
\begin{eqnarray}
\widetilde{\mathcal{H}}_{J} (\widehat{\boldsymbol{\mu}}_{\mbox{\tiny SC}}) & \leq &\mathcal{H}_J(\boldsymbol{\mu}^\ast_J) + K \norm{\widehat{\boldsymbol{\mu}}_{\mbox{\tiny SC}}-\boldsymbol{\mu}_J^\ast}_2 \leq \frac{1}{T_0} \sum_{t \in \mathcal{T}_0} \left( \bar {\lambda}_t(\boldsymbol{\mu}^\ast_J) - ({\boldsymbol{\epsilon}}_t'\mathbf{w}^\ast_J)  \right)^2 + K \norm{\widehat{\boldsymbol{\mu}}_{\mbox{\tiny SC}}-\boldsymbol{\mu}_J^\ast}_2. \nonumber
\end{eqnarray} 

Following the same arguments as above, $\frac{1}{T_0} \sum_{t \in \mathcal{T}_0} \left( \bar {\lambda}_t(\boldsymbol{\mu}^\ast_J) - ({\boldsymbol{\epsilon}}_t'\mathbf{w}^\ast_J)  \right)^2  \buildrel p \over \rightarrow \sigma^2_{\bar \lambda}(\boldsymbol{\mu}_0)$. Moreover, since $\widehat{\boldsymbol{\mu}}_{\mbox{\tiny SC}}  \buildrel p \over \rightarrow \boldsymbol{\mu}_0$ and $\boldsymbol{\mu}^\ast_J \rightarrow \boldsymbol{\mu}_0$, it follows that $ K \norm{\widehat{\boldsymbol{\mu}}_{\mbox{\tiny SC}} -\boldsymbol{\mu}_J^\ast}_2 \buildrel p \over \rightarrow 0 $. Therefore,  $\widetilde{\mathcal{H}}_{J} (\widehat{\boldsymbol{\mu}}_{\mbox{\tiny SC}}) $ is bounded from above by a term that converges in probability to $\sigma^2_{\bar \lambda}(\boldsymbol{\mu}_0)$. Since we also have that   $\widetilde{\mathcal{H}}_{J} (\boldsymbol{\mu}) $ is bounded from below by a function that converges uniformly in $\boldsymbol{\mu} \in \mathcal{M}$ to the continuous function $\sigma^2_{\bar \lambda}(\boldsymbol{\mu})$, it follows that $\widetilde{\mathcal{H}}_{J} (\widehat{\boldsymbol{\mu}}_{\mbox{\tiny SC}}) = \frac{1}{T_0} \sum_{t \in \mathcal{T}_0} \left( y_{0t} -\widehat{\mathbf{w}}_{\mbox{\tiny SC}} ' {\mathbf{y}}_t  \right)^2  \buildrel p \over \rightarrow \sigma_{\bar \lambda}^2(\boldsymbol{\mu}_0) = \sigma^2_0$.

Finally, we consider the third result, that  $\norm{\widehat{\mathbf{w}}_{\mbox{\tiny SC}}}_2 \buildrel p \over \rightarrow 0$. We first show that  $\frac{1}{T_0} \sum_{t \in \mathcal{T}_0} \left(\boldsymbol{\epsilon}_t ' \widehat{\mathbf{w}}_{\mbox{\tiny SC}} \right)^2 \buildrel p \over \rightarrow 0$. Let $\boldsymbol{\Lambda}$ be the $(T_0 \times F)$ matrix with rows $\boldsymbol{\lambda}_t$.  Since $\widehat{\mathbf{w}}_{\mbox{\tiny SC}}$ is the argmin of  equation \ref{eq_objective}, it follows that $\norm{\boldsymbol{\epsilon}_0  + \boldsymbol{\Lambda}(\boldsymbol{\mu}_0 - \widehat{\boldsymbol{\mu}}_{\mbox{\tiny SC}}) - \mathbf{E}  \widehat{\mathbf{w}}_{\mbox{\tiny SC}}}_2^2 \leq \norm{\boldsymbol{\epsilon}_0  + \boldsymbol{\Lambda} (\boldsymbol{\mu}_0 - \boldsymbol{\mu}^\ast_J) - \mathbf{E}  \mathbf{w}^\ast_J}_2^2 $, which implies
\begin{eqnarray}
\norm{\mathbf{E}  \widehat{\mathbf{w}}_{\mbox{\tiny SC}}}_2^2 &\leq & \norm{\mathbf{E}   \widehat{\mathbf{w}}_{\mbox{\tiny SC}}}_2^2 + \norm{\boldsymbol{\Lambda}(\boldsymbol{\mu}_0 - \widehat{\boldsymbol{\mu}}_{\mbox{\tiny SC}})}_2^2  \leq 2 \left| \boldsymbol{\epsilon}_0' \boldsymbol{\Lambda} (\boldsymbol{\mu}_0-\widehat{\boldsymbol{\mu}}_{\mbox{\tiny SC}})  \right| + 2 \left| \boldsymbol{\epsilon}_0' \boldsymbol{\Lambda} (\boldsymbol{\mu}_0- \boldsymbol{\mu}_J^\ast)  \right| \\
&& +  2 \left| (\boldsymbol{\mu}_0-\widehat{\boldsymbol{\mu}}_{\mbox{\tiny SC}})' \boldsymbol{\Lambda}' \mathbf{E}  \widehat{\mathbf{w}}_{\mbox{\tiny SC}} \right| + 2 \left| (\boldsymbol{\mu}_0- \boldsymbol{\mu}_J^\ast)' \boldsymbol{\Lambda}' \mathbf{E}  \mathbf{w}_J^\ast    \right|  + 2 \left| \boldsymbol{\epsilon}_0' \mathbf{E}  \widehat{\mathbf{w}}_{\mbox{\tiny SC}} \right| + \\
&&+ 2 \left| \boldsymbol{\epsilon}_0' \mathbf{E}  \mathbf{w}_J^\ast \right| +  \norm{\mathbf{E}  \mathbf{w}_J^\ast}_2^2 +  \norm{\boldsymbol{\Lambda}(\boldsymbol{\mu}_0 -  \boldsymbol{\mu}_J^\ast)}_2^2.
\end{eqnarray}

We show that all the terms on the right hand side of the above equation are $o_p(1)$ when divided by $T_0$. For any  $\boldsymbol{\mu}$, $2 \left| \boldsymbol{\epsilon}_0' \boldsymbol{\Lambda} (\boldsymbol{\mu}_0- \boldsymbol{\mu})  \right| \leq \norm{\boldsymbol{\epsilon}_0' \boldsymbol{\Lambda}}_\infty \norm{\boldsymbol{\mu}_0- \boldsymbol{\mu}}_1 \leq c\norm{\boldsymbol{\epsilon}_0' \boldsymbol{\lambda}}_\infty$, where $\frac{1}{T_0}\norm{\boldsymbol{\epsilon}_0' \boldsymbol{\lambda}}_\infty = o_p(1) $ from Assumption \ref{Assumption_technical}(a). We also have  $ \left| (\boldsymbol{\mu}_0-\widehat{\boldsymbol{\mu}}_{\mbox{\tiny SC}})' \boldsymbol{\Lambda}' \mathbf{E}  \widehat{\mathbf{w}}_{\mbox{\tiny SC}} \right| \leq \norm{\boldsymbol{\mu}_0-\widehat{\boldsymbol{\mu}}_{\mbox{\tiny SC}}}_2 \norm{\boldsymbol{\Lambda}' \mathbf{E}  \widehat{\mathbf{w}}_{\mbox{\tiny SC}}}_2 \leq c\norm{\boldsymbol{\Lambda}' \mathbf{E}  \widehat{\mathbf{w}}_{\mbox{\tiny SC}}}_1 \leq c\norm{\boldsymbol{\Lambda}' \mathbf{E} }_\infty \norm{\widehat{\mathbf{w}}_{\mbox{\tiny SC}}}_1 \leq c\norm{\boldsymbol{\Lambda}' \mathbf{E} }_\infty$, where $\frac{1}{T_0}\norm{\boldsymbol{\Lambda}' \mathbf{E} }_\infty=o_p(1)$ from Assumption \ref{Assumption_technical}(a). Likewise, $\frac{1}{T_0}\left| \boldsymbol{\epsilon}_0' \mathbf{E}  \widehat{\mathbf{w}}_{\mbox{\tiny SC}} \right| =o_p(1)$ because $\frac{1}{T_0}\norm{\boldsymbol{\epsilon}_0' \mathbf{E} }_\infty=o_p(1)$ from Assumption \ref{Assumption_technical}(a). Moreover, from equation \ref{eq_w_star}, $\frac{1}{T_0} \sum_{t \in \mathcal{T}_0} ({\boldsymbol{\epsilon}}_t'\mathbf{w}^\ast_J)^2 = o_p(1)$.  Finally,
\begin{eqnarray}
\frac{1}{T_0} \norm{\boldsymbol{\Lambda}(\boldsymbol{\mu}_0 -  \boldsymbol{\mu}_J^\ast)}_2^2 = (\boldsymbol{\mu}_0 - \boldsymbol{\mu}_J^\ast)' \left( \frac{1}{T_0} \sum_{t \in \mathcal{T}_0} \boldsymbol{\lambda}_t'\boldsymbol{\lambda}_t \right)(\boldsymbol{\mu}_0 - \boldsymbol{\mu}_J^\ast)= o(1).
\end{eqnarray}

Combining all these results, we have $\frac{1}{T_0} \sum_{t \in \mathcal{T}_0} \left(\boldsymbol{\epsilon}_t ' \widehat{\mathbf{w}}_{\mbox{\tiny SC}} \right)^2 \buildrel p \over \rightarrow 0$. Now suppose $\norm{\widehat{\mathbf{w}}_{\mbox{\tiny SC}}}_2$ does not converge in probability to zero. Since $\mathbf{w}  \in \Delta^{J-1}$, this would imply that there is a constant $b$ such that $P\left( \mbox{max}_i \left\{ \widehat w_i^2\right\} > b  \right) $ does not converge to zero. Note that this would not be true if we considered $\mathbf{w} \in \mathbb{R}^{J}$. However, given the restriction $\mathbf{w}  \in \Delta^{J-1}$, we have that $\norm{\mathbf{w}}_2^2 \leq { \mbox{max}_i \left\{ w_i \right\} + \mbox{max}_i \left\{ w_i^2\right\}}$. In this case, for infinitely many $T_0$, we have with probability greater than some $\xi >0$,
\begin{eqnarray}
\frac{1}{T_0} \sum_{t \in \mathcal{T}_0} \left(\boldsymbol{\epsilon}_t ' \widehat{\mathbf{w}}_{\mbox{\tiny SC}} \right)^2  &=& \frac{1}{T_0}\sum_{t \in \mathcal{T}_0}  \sum_{i=1}^J {\epsilon}_{it}^2 \widehat{{w}}_{i}^2 +  \frac{1}{T_0}\sum_{t \in \mathcal{T}_0}  \sum_{i \neq j} {\epsilon}_{it}{\epsilon}_{jt} \widehat{{w}}_{i}\widehat{{w}}_{j} \geq \\ 
& \geq & b \underset{1 \leq i \leq J}{\mbox{min}} \left\{ \frac{1}{T_0}\sum_{t \in \mathcal{T}_0}  {\epsilon}_{it}^2 \right\} +  \frac{1}{T_0}\sum_{t \in \mathcal{T}_0}  \sum_{i \neq j} {\epsilon}_{it}{\epsilon}_{jt} \widehat{{w}}_{i}\widehat{{w}}_{j}.
\end{eqnarray}

From Assumption \ref{Assumption_technical}(b), $b \underset{1 \leq i \leq J}{\mbox{min}} \left\{ \frac{1}{T_0}\sum_{t \in \mathcal{T}_0}  {\epsilon}_{it}^2 \right\} > bc$ with probability $1-o(1)$. Finally, note that $\widehat{{w}}_{i}\widehat{{w}}_{j} > 0$ and $\sum_{i \neq j} \widehat{{w}}_{i}\widehat{{w}}_{j} < 1$. Therefore,
\begin{eqnarray}
\left | \frac{1}{T_0}\sum_{t \in \mathcal{T}_0}  \sum_{i \neq j} {\epsilon}_{it}{\epsilon}_{jt} \widehat{{w}}_{i}\widehat{{w}}_{j} \right| \leq \underset{1 \leq i, j \leq J, i \neq j}{\mbox{max}} \left\{ \left| \frac{1}{T_0} \sum_{t \in \mathcal{T}_0}  \epsilon_{it} \epsilon_{jt} \right| \right\}=o_p(1),
\end{eqnarray}
from Assumption \ref{Assumption_technical}(b). Combining these results, this contradicts  $\frac{1}{T_0} \sum_{t \in \mathcal{T}_0} \left(\boldsymbol{\epsilon}_t ' \widehat{\mathbf{w}}_{\mbox{\tiny SC}} \right)^2 =o_p(1)$, which implies that  $\norm{\widehat{\mathbf{w}}_{\mbox{\tiny SC}}}_2 \buildrel p \over \rightarrow 0$. \end{proof}

\subsubsection{Proof of Corollary \ref{corollary}}
\label{Proof_corollary}

\begin{corollary_b}{\ref{corollary}} 

Suppose all assumptions for Proposition \ref{Convergence_SC} are satisfied, and that, for all $t \in \mathcal{T}_1$, $\epsilon_{it}$ is independent from $\{\epsilon_{i\tau}\}_{\tau \in \mathcal{T}_0}$. Then, for any $t \in \mathcal{T}_1$, $\hat  \alpha_{0t}^{\mbox{\tiny SC}}  \buildrel p \over \rightarrow \alpha_{0t} + \epsilon_{0t}$ when $T_0 \rightarrow \infty$.
 
\end{corollary_b}

\begin{proof}
Note that, for $t \in \mathcal{T}_1$, $\hat  \alpha_{0t}^{\mbox{\tiny SC}} =\alpha_{0t} + \epsilon_{0t} +  \boldsymbol{\lambda}_t (\boldsymbol{\mu}_0 - \widehat{{\boldsymbol{\mu}}}_{\mbox{\tiny SC}}) - \boldsymbol{\epsilon}_t' \widehat{{\mathbf{w}}}_{\mbox{\tiny SC}}$. Since we consider $  \boldsymbol{\lambda}_t $ fixed, it follows from Proposition \ref{Convergence_SC}(i) that $ \boldsymbol{\lambda}_t (\boldsymbol{\mu}_0 - \widehat{{\boldsymbol{\mu}}}_{\mbox{\tiny SC}})  \buildrel p \over \rightarrow 0$. Now under the assumptions from Corollary \ref{corollary}, it follows that $\boldsymbol{\epsilon}_t$ for $t \in \mathcal{T}_1$ is independent of $\widehat{{\mathbf{w}}}_{\mbox{\tiny SC}}$. In this case, $\mathbb{E}\left[ \left(\boldsymbol{\epsilon}_t' \widehat{{\mathbf{w}}}_{\mbox{\tiny SC}}  \right)^2  \right] \leq \left[ \mbox{sup}_{1 \leq i \leq J}  \mathbb{E} \left[\epsilon_{it}^2  \right]\right]  \mathbb{E} \left[ \norm{{\widehat{\mathbf{w}}}_{\mbox{\tiny SC}} }_2^2  \right]$, where the first term is $O(1)$ given Assumption \ref{Assumption_e}, and $\mathbb{E} \left[ \norm{{\widehat{\mathbf{w}}}_{\mbox{\tiny SC}} }_2^2  \right] \rightarrow 0$ from Proposition \ref{Convergence_SC}(iii) and $\norm{\widehat{{\mathbf{w}}}}_2^2$ bounded. Combining these results, $\hat  \alpha_{0t}^{\mbox{\tiny SC}}  \buildrel p \over \rightarrow \alpha_{0t} + \epsilon_{0t}$ when $T_0 \rightarrow \infty$. 
\end{proof}

\subsubsection{Proof of Proposition \ref{Prop_OLS_K_large}}
\label{Proof_OLS_large}

\begin{proposition_b}{\ref{Prop_OLS_K_large}}
Suppose we observe $(y_{0t},...,y_{Jt})$ for periods $t \in \{ -T_0+1,...,-1,0,1,...,T_1 \}$, where $J$ is a function of $T_0$. Potential outcomes are defined in equation (\ref{model}).  Let  $\widehat{\boldsymbol{\mu}}_{\mbox{\tiny OLS}}$  be defined as ${\mathbf{M}_J}' \widehat{{\mathbf{b}}}_{\mbox{\tiny OLS}}$, where $\widehat{{\mathbf{b}}}_{\mbox{\tiny OLS}}$ is defined in equation $(\ref{OLS_eq})$. Assume that $J/T_0 \rightarrow c \in [0,1)$, and that Assumptions \ref{Assumption_e}, \ref{Assumption_mu_ols}, and  \ref{Assumption_lambda_ols} hold. Then, when  $T_0 \rightarrow \infty$,  $\widehat{\boldsymbol{\mu}}_{\mbox{\tiny OLS}} \buildrel p \over \rightarrow \boldsymbol{\mu}_0$.

\end{proposition_b}

\begin{proof}
Given Assumption \ref{Assumption_mu_ols}, we can label the first $RF$ control units so that each block of $F$ control units is such that the $F \times F$ matrix of factor loadings for each of those blocks, $\boldsymbol{\mu}(p)$ for $p=1,...,R$, are invertible with uniformly bounded $\norm{\boldsymbol{\mu}(p)}_2$ and $\norm{(\boldsymbol{\mu}(p))^{-1}}_2$. The  $(J-RF) \times F$ matrix with the factor loadings  of the remaining $J-RF$ control units is defined as $\boldsymbol{\mu}(R+1)$. Therefore, $\mathbf{M}_J = [\boldsymbol{\mu}(1)' ~ \hdots ~ \boldsymbol{\mu}(R)' ~  \boldsymbol{\mu}(R+1)']'$. Likewise, let $\mathbf{y}_t(p)$ ($\boldsymbol{\epsilon}_t(p)$) be the $F \times 1$ vector of outcomes (errors) for the $F$ control units in block $p \in \{1,...,R\}$ at time $t$, while $\mathbf{y}_t(R+1)$  ($\boldsymbol{\epsilon}_t(R+1)$) is the same information for the remaining $J-RF$ control units. We also define $\mathbf{Y}(p)$ ($\mathbf{E}(p)$) as the $T_0 \times F$ matrix with the pre-treatment periods observations for the outcomes (errors) of control units in group $p$. These terms without the index in the parenthesis will refer to information on all $J$  control units. Finally, we define $\mathbf{y}_0$ ($\boldsymbol{\epsilon}_0$) as the $T_0 \times 1$ vector of pre-treatment outcomes (errors) of the treated unit. 

Under Assumption \ref{Assumption_mu_ols}, we can construct a $\boldsymbol{\beta}^\ast_J  = (\boldsymbol{\beta}^\ast_J(1)', ~ ... ~ ,\boldsymbol{\beta}^\ast_J(R)' , \boldsymbol{\beta}^\ast_J(R+1)')' \in \mathbb{R}^J$ such that ${\mathbf{M}_J}' \boldsymbol{\beta}^\ast_J = \sum_{p=1}^{R+1} \boldsymbol{\mu}(p)' \boldsymbol{\beta}^\ast_J(p) = \boldsymbol{\mu}_0$. For each $p=1,...,R$, let $\boldsymbol{\beta}^\ast_J(p) = \frac{1}{R} (\boldsymbol{\mu}(p)')^{-1}\boldsymbol{\mu}_0 $, and  $\boldsymbol{\beta}^\ast_J(R+1)=0$. In this case, $\boldsymbol{\beta}^\ast_J$ satisfies $\boldsymbol{\mu}_0 = {\mathbf{M}_J}' \boldsymbol{\beta}^\ast_J$, and we have that  $\norm{\boldsymbol{\beta}^\ast_J}_2 = o(1) $.  It follows that  $\mathbf{y}_{0} = \mathbf{Y}\boldsymbol{\beta}^\ast_J + \boldsymbol{\epsilon}_0 - \mathbf{E} \boldsymbol{\beta}^\ast_J$. We consider a change in variables so that we can focus on $\boldsymbol{\mu}_0$. Since $\boldsymbol{\mu}(1)$ is invertible, we have that  $\mathbf{y}_{0} = \mathbf{Y} \mathbf{H} \mathbf{H}^{-1} \boldsymbol{\beta}^\ast_J + \boldsymbol{\epsilon}_0 - \mathbf{E} \boldsymbol{\beta}^\ast_J$, where 
\begin{eqnarray}
\mathbf{H} = \left[  \begin{array}{ccccc} 
(\boldsymbol{\mu}(1)')^{-1} & -(\boldsymbol{\mu}(1)')^{-1}\boldsymbol{\mu}(2)' &   -(\boldsymbol{\mu}(1)')^{-1}\boldsymbol{\mu}(3)' & \hdots & -(\boldsymbol{\mu}(1)')^{-1}\boldsymbol{\mu}(R+1)' \\
0 &\mathbb{I}_F & 0 &... & 0  \\
\vdots  & & \ddots &  \\
0 & 0 & 0 & \hdots & \mathbb{I}_{J-RF} \\
 \end{array}  \right]
\end{eqnarray}

 and 
\begin{eqnarray}
\mathbf{H}^{-1} = \left[  \begin{array}{ccccc} 
\boldsymbol{\mu}(1)' & \boldsymbol{\mu}(2)' & \boldsymbol{\mu}(3)' &... & \boldsymbol{\mu}(R+1)'  \\
0 &\mathbb{I}_F & 0 &... & 0  \\
\vdots  & & \ddots &  \\
0 & 0 & 0 & \hdots & \mathbb{I}_{J-RF} \\
 \end{array}  \right],
\end{eqnarray}
where $ \mathbb{I}_{q}$ is a $q \times q$  identity matrix. Therefore, we have 
\begin{eqnarray}
\mathbf{y}_0 = \left[ \mathbf{Y}(1) (\boldsymbol{\mu}(1)')^{-1} \right] \boldsymbol{\mu}_0 + \sum_{p=2}^{R+1}\left[ \mathbf{Y}(p) -  \mathbf{Y}(1)(\boldsymbol{\mu}(1)')^{-1}\boldsymbol{\mu}(p)'  \right] \boldsymbol{\beta}^\ast_J(p) + \boldsymbol{\epsilon}_0 - \mathbf{E} \boldsymbol{\beta}^\ast_J.
\end{eqnarray}

Note that $\mathbf{Y}(p) -  \mathbf{Y}(1)(\boldsymbol{\mu}(1)')^{-1}\boldsymbol{\mu}(p) = \boldsymbol{\epsilon}(p) -  \boldsymbol{\epsilon}(1)(\boldsymbol{\mu}(1)')^{-1}\boldsymbol{\mu}(p)'$, which implies that  
\begin{eqnarray}
\mathbf{y}_0 = \left[ \mathbf{Y}(1) (\boldsymbol{\mu}(1)')^{-1} \right] \boldsymbol{\mu}_0 + \sum_{p=2}^{R+1}\left[ \boldsymbol{\epsilon}(p) -  \boldsymbol{\epsilon}(1)(\boldsymbol{\mu}(1)')^{-1}\boldsymbol{\mu}(p)'  \right] \boldsymbol{\beta}^\ast_J(p) + \boldsymbol{\epsilon}_0 - \mathbf{E} \boldsymbol{\beta}^\ast_J.
\end{eqnarray}

Now let $\widehat{{\mathbf{b}}}_{\mbox{\tiny OLS}}$ be the OLS estimator of $\mathbf{y}_0$ on $\mathbf{Y}$. Doing the same changes in variables as above, we have that
\begin{eqnarray}
\mathbf{y}_0 = \left[ \mathbf{Y}(1) (\boldsymbol{\mu}(1)')^{-1} \right] \widehat{\boldsymbol{\mu}}_{\mbox{\tiny OLS}} + \sum_{p=2}^{R+1}\left[ \boldsymbol{\epsilon}(p) -  \boldsymbol{\epsilon}(1)(\boldsymbol{\mu}(1)')^{-1}\boldsymbol{\mu}(p)'  \right] \widehat{{\mathbf{b}}}_{\mbox{\tiny OLS}}(p) + \widehat{\mathbf{u}},
\end{eqnarray}
where $ \widehat{\boldsymbol{\mu}}_{\mbox{\tiny OLS}}  = {\mathbf{M}_J}' \widehat{{\mathbf{b}}}_{\mbox{\tiny OLS}}$, and $ \widehat{\mathbf{u}}$ is the OLS residual from $\mathbf{y}_0$ on $\mathbf{Y}$.

Using Frisch-Waugh-Lovell theorem, we have that
\begin{eqnarray} \label{equation_ols_hat}
\widehat{\boldsymbol{\mu}}_{\mbox{\tiny OLS}} &=& \left( (\boldsymbol{\mu}(1))^{-1} \mathbf{Y}(1)' \mathbf{Q} \mathbf{Y}(1) (\boldsymbol{\mu}(1)')^{-1}\right)^{-1} \left( (\boldsymbol{\mu}(1))^{-1} \mathbf{Y}(1)' \mathbf{Q} \mathbf{y}_0 \right) \\
&=&   \boldsymbol{\mu}(1)'\left(  \mathbf{Y}(1)' \mathbf{Q} \mathbf{Y}(1) \right)^{-1} \left(  \mathbf{Y}(1)' \mathbf{Q} \mathbf{y}_0 \right)  \\
&=& \boldsymbol{\mu}_0 + \boldsymbol{\mu}(1)' \left(  \mathbf{Y}(1)' \mathbf{Q} \mathbf{Y}(1) \right)^{-1} \left(  \mathbf{Y}(1)' \mathbf{Q} \boldsymbol{\epsilon_0} \right) \\
&&- \boldsymbol{\mu}(1)'\left(  \mathbf{Y}(1)' \mathbf{Q} \mathbf{Y}(1) \right)^{-1} \left(  \mathbf{Y}(1)' \mathbf{Q}  \mathbf{E} \boldsymbol{\beta}^\ast_J \right) ,
\end{eqnarray}
where $\mathbf{Q}$ is the $(T_0 \times T_0)$ residual-maker matrix for a regression on $\left\{\boldsymbol{\epsilon}(p) -  \boldsymbol{\epsilon}(1)(\boldsymbol{\mu}(1)')^{-1}\boldsymbol{\mu}(p)' \right\}_{p=2}^{R+1}$.

We want to show that $\widehat{\boldsymbol{\mu}}_{\mbox{\tiny OLS}}  \buildrel p \over \rightarrow \boldsymbol{\mu}_0$. Consider first the term $  \mathbf{Y}(1)' \mathbf{Q} \mathbf{Y}(1)  =  \boldsymbol{\mu}(1) \boldsymbol{\Lambda}' \mathbf{Q} \boldsymbol{\Lambda} \boldsymbol{\mu}(1)'  + 2 \boldsymbol{\mu}(1)  \boldsymbol{\Lambda}' \mathbf{Q} \boldsymbol{\epsilon}(1) +  \boldsymbol{\epsilon}(1)' \mathbf{Q} \boldsymbol{\epsilon}(1)$.  Let $K = T_0 - J + F$. Since $J/T_0 \rightarrow c \in [0,1)$, we have that $K \rightarrow \infty$. Also, $rank(\mathbf{Q}) = K$. From Assumption \ref{Assumption_lambda_ols}, we have that $\norm{\frac{1}{\sqrt{K}}   \boldsymbol{\Lambda}' \mathbf{Q} }_2^2  = \mbox{tr} \left( \frac{1}{K}  \boldsymbol{\Lambda}' \mathbf{Q} \boldsymbol{\Lambda} \right) = O_p(1)$, which implies that $\norm{\frac{1}{\sqrt{K}} \boldsymbol{\mu}(1)  \boldsymbol{\Lambda}' \mathbf{Q} }_2^2 = \mbox{tr} \left( \frac{1}{K} \boldsymbol{\mu}(1) \boldsymbol{\Lambda}' \mathbf{Q} \boldsymbol{\Lambda} \boldsymbol{\mu}(1)'  \right)= O_p(1)$. Now consider the term $\mathbf{Q}\boldsymbol{\epsilon}(1)$. By definition of $\mathbf{Q}$, we have that $\mathbf{Q}(\boldsymbol{\epsilon}(p) -  \boldsymbol{\epsilon}(1)(\boldsymbol{\mu}(1)')^{-1}\boldsymbol{\mu}(p)') =0$, which implies $\mathbf{Q} \boldsymbol{\epsilon}(1) = \mathbf{Q} \boldsymbol{\epsilon}(p)(\boldsymbol{\mu}(p)'))^{-1}\boldsymbol{\mu}(1)'$ for all $p=1,..,R$. Therefore, $\mathbf{Q} \boldsymbol{\epsilon}(1) = \mathbf{Q} \frac{1}{R} \sum_{p=1}^R \boldsymbol{\epsilon}(p)(\boldsymbol{\mu}(p)')^{-1}\boldsymbol{\mu}(1)'$. Now define the $T_0 \times F$ matrix $\widetilde{\boldsymbol{\epsilon}}(p) \equiv \boldsymbol{\epsilon}(p)(\boldsymbol{\mu}(p)')^{-1}\boldsymbol{\mu}(1)'$, with elements $\tilde \epsilon_{ft}(p) =\mathbf{a}_{f}(p)' \boldsymbol{\epsilon}_t(p)$, where $\mathbf{a}_{f}(p)$ is an $F \times 1 $ given by the $f$-th column of $(\boldsymbol{\mu}(p)')^{-1}\boldsymbol{\mu}(1)'$. Given that $var(\epsilon_{it})$ is uniformly bounded by $\bar \gamma$, it follows that $var(\tilde \epsilon_{ft}(p)) \leq \bar \gamma \norm{\mathbf{a}_{f}(p)}_2^2$. Given Assumption \ref{Assumption_mu_ols},  $\norm{\mathbf{a}_{f}(p)}_2^2$ is uniformly bounded by an $\bar{\mathbf{a}}$, which implies that $var(\tilde \epsilon_{ft}(p)) \leq \bar \gamma \bar{\mathbf{a}}$. Now note that $\norm{\mathbf{Q}\frac{1}{R} \sum_{p=1}^R \widetilde{\boldsymbol{\epsilon}}(p)}_2^2 = \sum_{f=1}^F \norm{ \left[ \mathbf{Q}\frac{1}{R} \sum_{p=1}^R \widetilde{\boldsymbol{\epsilon}}(p) \right]_f}_2^2$, where, for a generic matrix $\mathbf{A}$, we define $[\mathbf{A}]_f$ as the $f$-th column of  $\mathbf{A}$. Note that $\norm{ \left[ \mathbf{Q}\frac{1}{R} \sum_{p=1}^R \widetilde{\boldsymbol{\epsilon}}(p) \right]_f}_2^2$ is the sum of  squared residual of the OLS regression of $\left[ \frac{1}{R} \sum_{p=1}^R \widetilde{\boldsymbol{\epsilon}}(p) \right]_f$ on $\left\{\boldsymbol{\epsilon}(p) -  \boldsymbol{\epsilon}(1)(\boldsymbol{\mu}(1)')^{-1}\boldsymbol{\mu}(p)' \right\}_{p=2}^{R+1}$.   Since $\mathbf{b} = 0 \in \mathbb{R}^{J-F}$ is a candidate solution for the OLS, we have
\begin{eqnarray} \nonumber
\norm{\frac{1}{\sqrt{K}} \left( \left[\mathbf{Q} \frac{1}{R} \sum_{p=1}^R \widetilde{\boldsymbol{\epsilon}}(p) \right]_f \right) }_2^2  &\leq& \frac{1}{K} \sum_{t \in \mathcal{T}_0}\left(\frac{1}{R} \sum_{p=1}^R \tilde \epsilon_{ft}(p) \right)^2  = \frac{1}{R} \frac{T_0}{K} \frac{1}{T_0} \sum_{t \in \mathcal{T}_0} \left(\frac{1}{\sqrt{R}} \sum_{p=1}^R \tilde \epsilon_{ft}(p) \right)^2 \\
&& \leq \frac{1}{R} \bar \gamma \bar{\mathbf{a}} \frac{T_0}{K} \frac{1}{T_0} \sum_{t \in \mathcal{T}_0} \left( \frac{ \frac{1}{\sqrt{R}} \sum_{p=1}^R \tilde \epsilon_{ft}(p)}{\sqrt{ var \left( \frac{1}{\sqrt{R}} \sum_{p=1}^R \tilde \epsilon_{ft}(p) \right)}} \right)^2.
\end{eqnarray}

Let $z_t = \left( \frac{ \frac{1}{\sqrt{R}} \sum_{p=1}^R \tilde \epsilon_{ft}(p)}{\sqrt{ var \left( \frac{1}{\sqrt{R}} \sum_{p=1}^R \tilde \epsilon_{ft}(p) \right)}} \right)^2$. By construction $\mathbb{E}[z_t] = 1$. Since  $\left(var \left( \frac{1}{\sqrt{R}} \sum_{p=1}^R \tilde \epsilon_{ft}(p) \right) \right)^2  \geq \underline{\gamma}^2 \left( \frac{1}{R} \sum_{p=1}^R \norm{\mathbf{a}_f (p)}_2^2 \right)^2$, and  $\mathbb{E} \left[ \left( \frac{1}{\sqrt{R}} \sum_{p=1}^R \tilde \epsilon_{ft}(p) \right)^4 \right] \leq \mbox{max} \{ \bar \gamma^2 , \bar \xi \} \left( \frac{1}{R} \sum_{p=1}^R \norm{\mathbf{a}_f (p)}_2^2 \right)^2 $, then $var(z_t)$ is uniformly bounded.\footnote{$\bar \xi$ is defined in the proof of Proposition \ref{Convergence_SC} as an upper bound for the fourth moment of $\epsilon_{it}$.} Therefore, $\frac{1}{T_0} \sum_{t \in \mathcal{T}_0} z_t \buildrel p \over \rightarrow 1$. Since $\frac{T_0}{K} \rightarrow \frac{1}{1-c}$, and    $R \rightarrow \infty$,  it follows that  $\frac{1}{\sqrt{K}}\mathbf{Q} \boldsymbol{\epsilon}(1) = o_p(1)$. Combining the results above, we have that  $ \frac{1}{K} \mathbf{Y}(1)' \mathbf{Q} \mathbf{Y}(1) = \boldsymbol{\mu}(1) \left[ \frac{1}{K}  \boldsymbol{\Lambda}' \mathbf{Q} \boldsymbol{\Lambda} \right] \boldsymbol{\mu}(1)'+ o_p(1)$. From Assumption \ref{Assumption_lambda_ols}, we have 
\begin{eqnarray}
\left( \boldsymbol{\mu}(1) \left[ \frac{1}{K}  \boldsymbol{\Lambda}' \mathbf{Q} \boldsymbol{\Lambda} \right] \boldsymbol{\mu}(1)' \right)^{-1}=(\boldsymbol{\mu}(1)')^{-1}  \left[ \frac{1}{K}  \boldsymbol{\Lambda}' \mathbf{Q} \boldsymbol{\Lambda} \right]^{-1} (\boldsymbol{\mu}(1))^{-1} = O_p(1), 
\end{eqnarray}
which implies that  $(\frac{1}{K}\mathbf{Y}' \mathbf{Q} \mathbf{Y})^{-1}=O_p(1)$.

Consider now $  \mathbf{Y}(1)' \mathbf{Q}  \mathbf{E} \boldsymbol{\beta}^\ast_J$. From the definition of $\mathbf{Q}$, we have that 
\begin{eqnarray} \nonumber \label{eq_Me}
\mathbf{Q} \mathbf{E} \boldsymbol{\beta}^\ast_J(p) &=& \mathbf{Q}\sum_{p=1}^{R+1} \boldsymbol{\epsilon}(p) \boldsymbol{\beta}^\ast_J(p)  =\sum_{p=1}^{R+1} \mathbf{Q} \boldsymbol{\epsilon}(1)(\boldsymbol{\mu}(1)')^{-1}\boldsymbol{\mu}(p)'\boldsymbol{\beta}^\ast_J(p)  \\
&=&  \mathbf{Q} \boldsymbol{\epsilon}(1)(\boldsymbol{\mu}(1)')^{-1} \sum_{p=1}^{R+1}\boldsymbol{\mu}(p)'\boldsymbol{\beta}^\ast_J (p)
= \mathbf{Q} \boldsymbol{\epsilon}(1)(\boldsymbol{\mu}(1)')^{-1} \boldsymbol{\mu}_0,
\end{eqnarray}
which implies that $\frac{1}{\sqrt{K}}\mathbf{Q} \mathbf{E} \boldsymbol{\beta}^\ast_J =o_p(1)$. Since $\frac{1}{\sqrt{K}}\mathbf{Q}\mathbf{Y}(1) =O_p(1)$, it follows that $\frac{1}{K} \mathbf{Y}(1)' \mathbf{Q}  \mathbf{E} \boldsymbol{\beta}^\ast_J=o_p(1)$.

Consider now $  \mathbf{Y}(1)' \mathbf{Q}  \boldsymbol{\epsilon}_0 = \boldsymbol{\mu}(1) \boldsymbol{\lambda}'  \mathbf{Q}  \boldsymbol{\epsilon}_0+\boldsymbol{\epsilon}(1)'\mathbf{Q}  \boldsymbol{\epsilon}_0$. Following the same arguments as above $\norm{\frac{1}{\sqrt{K}}\mathbf{Q}  \boldsymbol{\epsilon}_0}_2^2 \leq \frac{1}{K} \sum_{t \in \mathcal{T}_0} \epsilon_{0t}^2 = O_p(1)$, which implies  $\frac{1}{K} \boldsymbol{\epsilon}(1)'\mathbf{Q}  \boldsymbol{\epsilon}_0 = o_p(1)$. From Assumption \ref{Assumption_lambda_ols}, we have that $\frac{1}{K}  \boldsymbol{\lambda}'  \mathbf{Q}  \boldsymbol{\epsilon}_0 = o_p(1)$, which implies $ \frac{1}{K} \mathbf{Y}(1)' \mathbf{Q}  \boldsymbol{\epsilon}_0 =o_p(1)$. 

Combining all these results into equation \ref{equation_ols_hat}, we have that $\widehat{\boldsymbol{\mu}}_{\mbox{\tiny OLS}}  \buildrel p \over \rightarrow \boldsymbol{\mu}_0$.
\end{proof}

\subsubsection{Proof of Proposition \ref{Prop_OLS_K_finite}}
\label{Proof_OLS_finite}

\begin{proposition_b}{\ref{Prop_OLS_K_finite}}

Suppose we observe $(y_{0t},...,y_{Jt})$ for periods $t \in \{ -T_0+1,...,-1,0,1,...,T_1 \}$, where $J$ is a function of $T_0$. Potential outcomes are defined in equation (\ref{model}).  Let  $\widehat{\boldsymbol{\mu}}_{\mbox{\tiny OLS}}$  be defined as ${\mathbf{M}_J}' \widehat{{\mathbf{b}}}_{\mbox{\tiny OLS}}$, where $\widehat{{\mathbf{b}}}_{\mbox{\tiny OLS}}$ is defined in equation $(\ref{OLS_eq})$. Assume that $T_0 \geq J$, and that Assumptions \ref{Assumption_mu_ols} and \ref{Assumption_normal_ols} hold. Then,  when $T_0 \rightarrow \infty$, $\mathbb{E}[\widehat{\boldsymbol{\mu}}_{\mbox{\tiny OLS}}  - \boldsymbol{\mu}_0] \rightarrow 0$. 
\end{proposition_b}

\begin{proof}
Following the same steps as the proof of Proposition \ref{Prop_OLS_K_large}, we have that 
\begin{eqnarray} \label{equation_ols_hat} \nonumber
\widehat{\boldsymbol{\mu}}_{\mbox{\tiny OLS}} &=& \boldsymbol{\mu}_0 + \boldsymbol{\mu}(1)' \left(  \mathbf{Y}(1)' \mathbf{Q} \mathbf{Y}(1) \right)^{-1} \left(  \mathbf{Y}(1)' \mathbf{Q} \boldsymbol{\epsilon_0} \right) - \boldsymbol{\mu}(1)'\left(  \mathbf{Y}(1)' \mathbf{Q} \mathbf{Y}(1) \right)^{-1} \left(  \mathbf{Y}(1)' \mathbf{Q}  \mathbf{E} \boldsymbol{\beta}^\ast_J \right) ,
\end{eqnarray}
where $\mathbf{Q}$ is the $(T_0 \times T_0)$ residual-maker matrix for a regression on $\left\{\boldsymbol{\epsilon}(p) -  \boldsymbol{\epsilon}(1)(\boldsymbol{\mu}(1)')^{-1}\boldsymbol{\mu}(p)' \right\}_{p=2}^{R+1}$.

We want to show that $\mathbb{E} \left[\widehat{\boldsymbol{\mu}}_{\mbox{\tiny OLS}} - \boldsymbol{\mu}_0 \right] \rightarrow 0$. First, note that 
\begin{eqnarray} \nonumber
\mathbb{E} \left[ \boldsymbol{\mu}(1)' \left(  \mathbf{Y}(1)' \mathbf{Q} \mathbf{Y}(1) \right)^{-1}  \mathbf{Y}(1)' \mathbf{Q} \boldsymbol{\epsilon_0}  | \mathbf{Y} \right] &=&    \boldsymbol{\mu}(1)' \left( \mathbf{Y}(1)' \mathbf{Q} \mathbf{Y}(1) \right)^{-1}  \mathbf{Y}(1)' \mathbf{Q} \mathbb{E} \left[\boldsymbol{\epsilon_0}| \mathbf{Y} \right]  = 0.
\end{eqnarray}

Consider now the term $\boldsymbol{\mu}(1)'\left(  \mathbf{Y}(1)' \mathbf{Q} \mathbf{Y}(1) \right)^{-1} \left(  \mathbf{Y}(1)' \mathbf{Q}  \mathbf{E} \boldsymbol{\beta}^\ast_J \right)$. Note that $\mathbb{E}[\mathbf{E} | \mathbf{Y}] \neq 0$. Therefore, with finite $J$, the estimator is biased, which is consistent with the results from \cite{FP_SC}. We show that, in a setting in which $J \rightarrow \infty$, this bias goes to zero.  From equation \ref{eq_Me}, we have $
\mathbf{Q} \mathbf{E} \boldsymbol{\beta}^\ast_J = \mathbf{Q} \boldsymbol{\epsilon}(1)(\boldsymbol{\mu}(1)')^{-1} \boldsymbol{\mu}_0$.  Therefore, we need to consider the conditional expectation  $\mathbb{E}[\boldsymbol{\epsilon}(1) | \mathbf{Y}]$.

Given that $(\boldsymbol{\lambda}_t,\epsilon_{0t},\boldsymbol{\epsilon}_t)$ is iid multivariate normal (Assumption \ref{Assumption_normal_ols}), it follows that $\mathbb{E}[\boldsymbol{\epsilon}_t(1) | \mathbf{Y}] = \mathbb{E}[\boldsymbol{\epsilon}_t(1) | \mathbf{y}_t]$, where this conditional expectation is linear in $\mathbf{y}_t$. Using a change in variables, we can re-write this linear conditional expectation as $\mathbb{E}[\boldsymbol{\epsilon}_t(1) | \mathbf{y}_t] = \mathbf{B}^\ast \widetilde{\mathbf{y}}_t$, where $\widetilde{\mathbf{y}}_t = (\mathbf{y}_t(1)', ~ \mathbf{y}_t(2)' -  (\boldsymbol{\mu}(2)(\boldsymbol{\mu}(1))^{-1}\mathbf{y}_t(1))',\hdots,\mathbf{y}_t(R+1)' -   (\boldsymbol{\mu}(R+1)(\boldsymbol{\mu}(1))^{-1}\mathbf{y}_t(1))')'$, and $\mathbf{B}^\ast$ is an $F \times J$ matrix. The $j$-th row of matrix $\mathbf{B}^\ast$ is given by  $\mathbf{b}^\ast_j = \underset{\mathbf{b} \in \mathbb{R}^J}{\mbox{argmin}} \mathbb{E} \left[ \left( \epsilon_{jt} - \mathbf{b}' \widetilde{\mathbf{y}}_t \right)^2 \right] $. We show that the parameters associated with $\mathbf{y}_t(1)$, $\mathbf{b}_j^\ast(1)$, converge to zero when $J \rightarrow \infty$. For any $\widetilde{\mathbf{b}}_j$ such that  $\widetilde{\mathbf{b}}_j(1) \neq 0$, note that $\mathbb{E} \left[ \epsilon_{jt} - \widetilde{\mathbf{b}}_j' \widetilde{\mathbf{y}}_t \right]^2 \geq \widetilde{\mathbf{b}}_j(1)' \boldsymbol{\mu}(1) \mathbb{E}[\boldsymbol{\lambda}_t'\boldsymbol{\lambda}_t] \boldsymbol{\mu}(1)' \widetilde{\mathbf{b}}_j(1)>0$, since $\boldsymbol{\mu}(1) \mathbb{E}[\boldsymbol{\lambda}_t'\boldsymbol{\lambda}_t] \boldsymbol{\mu}(1)'$ is positive definite. Now note that 
\begin{eqnarray}
\left[ -\frac{1}{R} \boldsymbol{\mu}(1) (\boldsymbol{\mu}(p))^{-1} \right][\mathbf{y}_t(p) -  \boldsymbol{\mu}(p)(\boldsymbol{\mu}(1))^{-1}\mathbf{y}_t(1)] = \frac{1}{R} \left[\boldsymbol{\epsilon}_t(1) - \boldsymbol{\mu}(1) (\boldsymbol{\mu}(p))^{-1} \boldsymbol{\epsilon}_t(p) \right].
\end{eqnarray}

 Let  $\widetilde{\mathbf{a}}_j(p) $ be the $j$-th row of $-\boldsymbol{\mu}(1) (\boldsymbol{\mu}(p))^{-1} $. Then the $j$-th column of $\left[ -\frac{1}{R} \boldsymbol{\mu}(1) (\boldsymbol{\mu}(p))^{-1} \right][\mathbf{y}_t(p) -  \boldsymbol{\mu}(p)(\boldsymbol{\mu}(1))^{-1}\mathbf{y}_t(1)]$ is given by $\frac{1}{R} \left({\epsilon}_{jt} - \widetilde{\mathbf{a}}_j(p) \boldsymbol{\epsilon}_t(p) \right) $. Given Assumption \ref{Assumption_mu_ols}, $\norm{(\boldsymbol{\mu}(p))^{-1}}_2$ is uniformly bounded, for $p=1,...,R$. Therefore, $\norm{\widetilde{\mathbf{a}}_j(p)}_2^2$ is uniformly bounded, which implies that $var(\widetilde{\mathbf{a}}_j(p) \boldsymbol{\epsilon}_t(p))$ is uniformly bounded. Therefore, if we choose $\mathbf{b}$ with $-\frac{1}{R}\widetilde{\mathbf{a}}_j(p)$ in the $j$-entry of block $p$, and zero otherwise, we have  $\mathbb{E} \left[ \epsilon_{jt} - {\mathbf{b}}' \widetilde{\mathbf{y}}_t \right]^2 \rightarrow 0$ when $R \rightarrow \infty$. Therefore, it must be that $\mathbf{b}^\ast_j(1) \rightarrow 0$.

Back to equation \ref{eq_Me}, we have
\begin{eqnarray} 
\mathbf{Q} \mathbb{E}[\mathbf{E}  | \mathbf{Y}]\boldsymbol{\beta}^\ast_J &=& \mathbf{Q} \mathbb{E}[ \boldsymbol{\epsilon}(1)| \mathbf{Y}](\boldsymbol{\mu}(1)')^{-1} \boldsymbol{\mu}_0 = \mathbf{Q} \widetilde{\mathbf{Y}}(\mathbf{B}^\ast)'(\boldsymbol{\mu}(1)')^{-1}  \boldsymbol{\mu}_0 \\
& =&  \mathbf{Q} {\mathbf{Y}}(1) \mathbf{B}^\ast(1)(\boldsymbol{\mu}(1)')^{-1} \boldsymbol{\mu}_0 ,
\end{eqnarray}
where $\mathbf{B}^\ast(1)$ is the first $F$ columns of $\mathbf{B}^\ast$. The last equality follows from the definition of matrix $\mathbf{Q}$. Therefore,
\begin{eqnarray}
(\boldsymbol{\mu}(1)') \left(  \mathbf{Y}(1)' \mathbf{Q} \mathbf{Y}(1) \right)^{-1} \left(  \mathbf{Y}(1)' \mathbf{Q} \mathbb{E}[ \mathbf{E} | \mathbf{Y}] \boldsymbol{\beta}^\ast_J \right) = (\boldsymbol{\mu}(1)') \mathbf{B}^\ast(1)(\boldsymbol{\mu}(1)')^{-1} \boldsymbol{\mu}_0 \rightarrow 0.
\end{eqnarray}

Combining the results above, we have that $\mathbb{E}[\widehat{\boldsymbol{\mu}}_{\mbox{\tiny OLS}} - \boldsymbol{\mu}_0] \rightarrow 0$.
\end{proof}

\begin{rem}
\normalfont

In the proof of Proposition \ref{Prop_OLS_K_finite}  we have to consider the expectation of the term $\left(  \mathbf{Y}(1)' \mathbf{Q} \mathbf{Y}(1) \right)^{-1} \left(  \mathbf{Y}(1)' \mathbf{Q}  \mathbf{E} \boldsymbol{\beta}^\ast_J \right)$ conditional on $\widetilde{\mathbf{Y}}$, because this term involves a non-linear function of these variables.  We show that $\mathbb{E} \left[\left(  \mathbf{Y}(1)' \mathbf{Q} \mathbf{Y}(1) \right)^{-1} \left(  \mathbf{Y}(1)' \mathbf{Q}  \mathbf{E} \boldsymbol{\beta}^\ast_J \right) | \widetilde{\mathbf{Y}}    \right]$ equals a term that does not depend on $\widetilde{\mathbf{Y}} $, and converges to zero. Therefore, we also have that  the unconditional expectation $\mathbb{E} \left[\left(  \mathbf{Y}(1)' \mathbf{Q} \mathbf{Y}(1) \right)^{-1} \left(  \mathbf{Y}(1)' \mathbf{Q}  \mathbf{E} \boldsymbol{\beta}^\ast_J \right)     \right] \rightarrow 0$. If we consider a fixed sequence of $\boldsymbol{\Lambda}$, then this proof would not work, because $\mathbb{E} \left[\mathbf{E} | \widetilde{\mathbf{Y}} , \boldsymbol{\Lambda}  \right] = \mathbf{E}$. Therefore, we cannot guarantee that $\mathbb{E} \left[\left(  \mathbf{Y}(1)' \mathbf{Q} \mathbf{Y}(1) \right)^{-1} \left(  \mathbf{Y}(1)' \mathbf{Q}  \mathbf{E} \boldsymbol{\beta}^\ast_J \right)  | \boldsymbol{\Lambda}   \right] \rightarrow 0$. 
\end{rem}

\subsection{Other results}

\subsubsection{Conditions for Assumption \ref{Assumption_mu} }
\label{Appendix_mu}

Suppose the underlying distribution of $\boldsymbol{\mu}_i$  has finite support $\{\mathbf{m}_1,...,\mathbf{m}_{\bar q} \}$, with $Pr(\boldsymbol{\mu}_i = \mathbf{m}_q) = p_q > 0$ independent across $i$. Fix a $\boldsymbol{\mu}_0 = \mathbf{m}_q$. If we let $J_q$ be the number of observations with $\boldsymbol{\mu}_i = \mathbf{m}_q$, then by the Strong Law of Large Numbers, we have that $P \left( \frac{J_q}{J} \rightarrow p_q \right)=1$. Therefore, with probability one, there is a $\widetilde J \in \mathbb{N}$ such that   $J > \widetilde J$ implies $J_q/J > \frac{p_q}{2}$, which implies $J_q > cJ$ for a constant $c$. 
Now consider a $\mathbf{w}^\ast_J$ that assigns weights $ \frac{1}{\left\lfloor cJ \right \rfloor} $ for  $\left\lfloor cJ \right \rfloor$ control units with  $\boldsymbol{\mu}_0 = \mathbf{m}_q$. By construction ${\mathbf{M}_J}' \mathbf{w}^\ast_J = \mathbf{m}_q$. Moreover, we have that $\norm{\mathbf{w}^\ast_J}_2^2 = \frac{1}{\left\lfloor cJ \right \rfloor}$, which implies that $\norm{\mathbf{w}^\ast_J}_2^2 \rightarrow 0$.

\subsubsection{Conditions for Assumption \ref{Assumption_technical} }
\label{Appendix_technical}

We show a very simple example in which Assumption \ref{Assumption_technical} is satisfied. Suppose $\epsilon_{it}$ is independent across $i$ and $t$, and has uniformly bounded $k$-th moments across $i$ and $t$, for an even $k$. In this case, for any $\eta>0$,
\begin{eqnarray}
P \left(  \underset{1 \leq j \leq J}{\mbox{max}} \left\{ \left| \frac{1}{T_0} \sum_{t \in \mathcal{T}_0} \epsilon_{0t} \epsilon_{jt} \right| \right\}   > \eta \right) &\leq& \sum_{j=1}^J P \left( \left|  \frac{1}{T_0} \sum_{t \in \mathcal{T}_0} \epsilon_{0t} \epsilon_{jt} \right| > \eta \right) \\
&\leq&  \sum_{j=1}^J \frac{\mathbb{E}\left[ \left( \frac{1}{T_0}  \sum_{t \in \mathcal{T}_0} \epsilon_{0t} \epsilon_{jt}  \right)^k   \right]}{\eta^k} \\
&\leq & \frac{J}{T_0^k} \left(\sum_{h=1}^\frac{k}{2} C_h T_0^h \right), 
\end{eqnarray}
for constants $C_1,...,C_{{k}/{2}}$. Therefore, $\underset{1 \leq j \leq J}{\mbox{max}} \left\{ \left| \frac{1}{T_0} \sum_{t \in \mathcal{T}_0} \epsilon_{0t} \epsilon_{jt} \right| \right\}  =o_p(1)$ if $\frac{J}{T_0^{k/2}} \rightarrow 0$, which implies that this condition can be valid even when $J$ grows at a faster rate than $T_0$ if we assume enough uniformly bounded moments for the idiosyncratic shocks.  We can check the other conditions considered in Assumption \ref{Assumption_technical} following the same idea. Note that the term $\underset{1 \leq i, j \leq J, i \neq j}{\mbox{max}} \left\{ \left| \frac{1}{T_0} \sum_{t \in \mathcal{T}_0}  \epsilon_{it} \epsilon_{jt} \right| \right\}$ would be bounded by the sum of $J(J-1)/2$ terms, which implies that we would require a larger $k$ to guarantee this conditions.

\subsubsection{Alternative for Corollary \ref{corollary} }
\label{Appendix_corollary}

We consider a different set of assumptions in which we can derive $\hat  \alpha_{0t}^{\mbox{\tiny SC}}  \buildrel p \over \rightarrow \alpha_{0t} + \epsilon_{0t}$ when $T_0 \rightarrow \infty$, allowing time-dependency for the idiosyncratic shocks. We consider the following set of assumptions, which are similar to the ones considered by \cite{Chernozhukov} for their Lemma 2.

\begin{ass}{(Number of control units and pre-treatment periods)} \label{Assumption_appendix_T}
\normalfont
$\mbox{log}(J) = o(T_0^{\frac{\tau}{3\tau + 1}})$, where $\tau$ is a constant defined in Assumption \ref{Assumption_appendix_e}.

\end{ass}

\begin{ass}{(idiosyncratic shocks)} \label{Assumption_appendix_e}
\normalfont
(a) $\mathbb{E}[\epsilon_{it}]=0$  for all $i$ and $t$; (b) $\{ \epsilon_{it} \}_{t \in \mathcal{T}_0 \cup \mathcal{T}_1}$ are independent across $i$; (c)  $\{\epsilon_{0t},...,\epsilon_{Jt} \}_{t \in \mathcal{T}_0}$ is $\beta$-mixing, with coefficients satisfying $\beta(t) \leq D_1 \mbox{exp}(-D_2 t^\tau)$, where $D_1,D_2,\tau>0$ are constants; (d) $\epsilon_{it}$ have uniformly bounded fourth moments across $i$ and $t$, and $\frac{1}{T_0}\sum_{t \in \mathcal{T}_0} \mathbb{E}[\epsilon_{0t}^2] \rightarrow \sigma^2_0$; (e) $\exists \underline{\gamma}>0$ such that  $ \mathbb{E}[\epsilon_{it}^2]  \geq \underline{\gamma}$ across $i$ and $t$.
\end{ass}

\begin{ass}{(factor loadings)} \label{Assumption_appendix_mu}
\normalfont
(a) As $J \rightarrow \infty$, there is a sequence $\mathbf{w}^\ast_J \in \Delta^{J-1}$ such that $\norm{{\mathbf{M}_J}' \mathbf{w}^\ast_J - \boldsymbol{\mu}_0}_2 \rightarrow 0$, and $\norm{\mathbf{w}^\ast_J}_2 \rightarrow 0$, and (b) the sequence $\boldsymbol{\mu}_i$ is uniformly bounded. 

\end{ass}

\begin{ass}{(common factors)} \label{Assumption_appendix_lambda}
\normalfont
(a) $\frac{1}{T_0} \sum_{t \in \mathcal{T}_0} \boldsymbol{\lambda}_t ' \boldsymbol{\lambda}_t   \rightarrow \boldsymbol{\Omega}$ positive definite; (b) let $m = \lfloor [4D_2^{-1} \mbox{log}(JT_0)]^\frac{1}{\tau} \rfloor$ and $k = \lfloor T_0 / m \rfloor$, and define the sets $H_p = \{-p, - m -p, -2m-p, \hdots, -(k-1)m-p   \}$ for $p=1,...,m$.   We assume that there are positive constants $b_1$ and $b_2$ such that $\underset{T_0 \rightarrow \infty }{\mbox{liminf}} \left(  \underset{p=1,...,m}{\mbox{min}}\left\{ \frac{1}{k} \sum_{t \in H_p} \left| \lambda_{ft} \right|^2  \right\} \right) > b_1$ and  $\underset{T_0 \rightarrow \infty }{\mbox{limsup}} \left(  \underset{p=1,...,m}{\mbox{max}}\left\{ \frac{1}{k} \sum_{t \in H_p} \left| \lambda_{ft} \right|^3  \right\} \right) < b_2$.

\end{ass}

\begin{ass}{(other assumptions)} \label{Assumption_appendix_technical}
\normalfont
(a) $\exists$ $c>0$ such that  $\mbox{max}_{1 \leq j \leq J} \sum_{t \in \mathcal{T}_0} \left| \epsilon_{0t}^2 \epsilon_{jt}^2 \right| \leq c^2 T_0$ and $\mbox{max}_{1 \leq j \leq J} \sum_{t \in \mathcal{T}_0} \left| \lambda_{ft}^2 \epsilon_{jt}^2 \right| \leq c^2 T_0$ for all $f \in \{1,...,F \}$ with probability $1-o(1)$; (b) there is a sequence $l_J>0 $ such that $l_J [\mbox{log}(T_0 \lor J)]^{\frac{1+\tau}{2\tau}}T_0^{-1/2} \rightarrow 0$ that satisfies (i) for any $t \in \mathcal{T}_1$, $(\boldsymbol{\epsilon}_t' \delta)^2 \leq l_J \frac{1}{T_0} \sum_{q \in \mathcal{T}_0} (\boldsymbol{\epsilon}_q'  \delta)^2$ for all $\delta \in \Delta^{J-1}$ with probability $1-o(1)$, and (ii) for $\mathbf{w}^\ast_J$ defined in Assumption \ref{Assumption_appendix_mu}, $l_J^{1/2} \norm{{\mathbf{M}_J}' \mathbf{w}^\ast_J - \boldsymbol{\mu}_0}_2 \rightarrow 0$, and $l_J^{1/2}\norm{\mathbf{w}^\ast_J}_2 \rightarrow 0$.

\end{ass}

We first prove  the following lemma, which is based on Lemma 18 from \cite{Chernozhukov}. 

\begin{lem} \label{lemma}
\normalfont
Under Assumptions \ref{Assumption_appendix_T}, \ref{Assumption_appendix_e}, \ref{Assumption_appendix_lambda},  and \ref{Assumption_appendix_technical}(a), we have that $ \frac{1}{T_0} \norm{\sum_{t \in \mathcal{T}_0} \boldsymbol{\epsilon}_t \epsilon_{0t} }_\infty=o_p(1)$, $ \frac{1}{T_0} \norm{\sum_{t \in \mathcal{T}_0} \boldsymbol{\epsilon}_t \lambda_{ft} }_\infty=o_p(1)$, and  $ \frac{1}{T_0} \norm{\sum_{t \in \mathcal{T}_0} {\epsilon}_{0t} \lambda_{ft} }_\infty=o_p(1)$. Moreover, $l_J \frac{1}{T_0} \norm{\sum_{t \in \mathcal{T}_0} \boldsymbol{\epsilon}_t \epsilon_{0t} }_\infty=o_p(1)$, $l_J \frac{1}{T_0} \norm{\sum_{t \in \mathcal{T}_0} \boldsymbol{\epsilon}_t \lambda_{ft} }_\infty=o_p(1)$, and  $l_J \frac{1}{T_0} \norm{\sum_{t \in \mathcal{T}_0} {\epsilon}_{0t} \lambda_{ft} }_\infty=o_p(1)$, for $l_t$  defined in Assumption \ref{Assumption_appendix_technical}(b).
\end{lem}

 \begin{proof}
The result $l_J \frac{1}{T_0} \norm{\sum_{t \in \mathcal{T}_0} \boldsymbol{\epsilon}_t \epsilon_{0t} }_\infty=o_p(1)$ follows simply from checking that the assumptions for Lemma 18 from \cite{Chernozhukov} are valid given  our assumptions. The results $l_J \frac{1}{T_0} \norm{\sum_{t \in \mathcal{T}_0} \boldsymbol{\epsilon}_t \lambda_{ft} }_\infty=o_p(1)$, and  $l_J \frac{1}{T_0} \norm{\sum_{t \in \mathcal{T}_0} {\epsilon}_{0t} \lambda_{ft} }_\infty=o_p(1)$ follow from a minor adjustment on Lemma 18 from \cite{Chernozhukov} to allow for a fixed sequence of $\boldsymbol{\lambda}_t$. We re-write their proof for our setting, considering these adjustments to allow for a   fixed sequence of $\boldsymbol{\lambda}_t$.

Let  $m = \lfloor [4D_2^{-1} \mbox{log}(JT_0)]^\frac{1}{\tau} \rfloor$ and $k = \lfloor T_0 / m \rfloor$, and define the sets $H_p = \{-p, - m -p, -2m-p, \hdots, -(k-1)m-p   \}$ for $p=1,...,m$.  We assume for now that $ T_0 / m $ is an integer. 
From Berbee's coupling, there exist a sequence of random variables $\{ \tilde \epsilon_{it} \}_{t \in H_p}$ such that (1)  $\{ \tilde \epsilon_{it} \}_{t \in H_p}$ is independent across $t$, (2) $ \tilde \epsilon_{it}$ has the same distribution as $\epsilon_{it}$ for $t \in H_p$, and (3) $P \left( \cup_{t \in H_p} \{\tilde \epsilon_{it} \neq \epsilon_{it} \} \right) \leq k \beta(m)$. Now note that 
\begin{eqnarray}
\frac{1}{k} \sum_{t \in H_p} \mathbb{E} \left|  \lambda_{ft} \tilde \epsilon_{it}  \right|^2  = \frac{1}{k} \sum_{t \in H_p} \left|  \lambda_{ft} \right|^2 \mathbb{E} \left| \tilde \epsilon_{it}  \right|^2 \geq \underline \gamma \frac{1}{k} \sum_{t \in H_p} \left|  \lambda_{ft} \right|^2,
\end{eqnarray}
where the last inequality follows from Assumption \ref{Assumption_appendix_e}. From Assumption \ref{Assumption_appendix_lambda}, there is a $T^\ast_1$ such that, for $T_0>T^\ast_1$, $\frac{1}{k} \sum_{t \in H_p} \left|  \lambda_{ft} \right|^2 > b_1$, implying that $\frac{1}{k} \sum_{t \in H_p} \mathbb{E} \left|  \lambda_{ft} \tilde \epsilon_{it}  \right|^2 > \underline{\gamma} b_1$. Likewise, there is a  $T^\ast_2$ such that $\frac{1}{k} \sum_{t \in H_p} \mathbb{E} \left|  \lambda_{ft} \tilde \epsilon_{it}  \right|^3 < \bar \gamma b_2$ for $T_0 > T_2^\ast$. Therefore, for $T_0 > \mbox{max} \{T_1^\ast, T_2^\ast \}$, we have 
\begin{eqnarray} \label{Pena}
\frac{  \left( \sum_{t \in H_p} \mathbb{E} \left|  \lambda_{ft} \tilde \epsilon_{it}  \right|^2 \right)^\frac{1}{2} }{\left(  \sum_{t \in H_p} \mathbb{E} \left|  \lambda_{ft} \tilde \epsilon_{it}  \right|^3 \right)^\frac{1}{3}    }  =\frac{ k^\frac{1}{2} \left( \frac{1}{k} \sum_{t \in H_p} \mathbb{E} \left|  \lambda_{ft} \tilde \epsilon_{it}  \right|^2 \right)^\frac{1}{2} }{ k^\frac{1}{3}\left( \frac{1}{k} \sum_{t \in H_p} \mathbb{E} \left|  \lambda_{ft} \tilde \epsilon_{it}  \right|^3 \right)^\frac{1}{3} } > k^\frac{1}{6}C_0 , 
\end{eqnarray}
for a constant $C_0$. Now let $W_{it} = \lambda_{ft}  \epsilon_{it}$ and $\widetilde W_{it} = \lambda_{ft} \tilde \epsilon_{it}$. Since $\mathbb{E}[ \lambda_{ft} \tilde \epsilon_{it}]=0$, by Theorem 7.4 of \cite{delapena2004}, there is a constant $C_1$ such that, for any $0 \leq x \leq C_0 k^{\frac{1}{6}}$,
\begin{eqnarray}
P  \left( \left|  \frac{\sum_{t \in H_p} \widetilde W_{it}}{\sqrt{\sum_{t \in H_p} \widetilde W_{it}^2}} \right| > x \right) \leq C_1(1- \Phi(x)),
\end{eqnarray}
where $\Phi(.)$ is the cdf of a $N(0,1)$. Therefore, for any $0 \leq x \leq C_0 k^{\frac{1}{6}}$,
\begin{eqnarray}
P  \left( \left|  \frac{\sum_{t \in H_p}  W_{it}}{\sqrt{\sum_{t \in H_p}  W_{it}^2}} \right|  >x \right) &\leq&   P  \left( \left|  \frac{\sum_{t \in H_p} \widetilde W_{it}}{\sqrt{\sum_{t \in H_p} \widetilde W_{it}^2}} \right| >x \right) + P \left(\cup_{t \in H_p} \{\tilde \epsilon_{it} \neq \epsilon_{it} \} \right) \\
& \leq&  C_1(1- \Phi(x)) + k\beta(m).
\end{eqnarray}

Following exactly the same steps as the proof of Lemma 18 from \cite{Chernozhukov}, we have that, for any $0 \leq x \leq C_0 k^{1/6} \sqrt{m}$,
\begin{eqnarray} \label{bound_lemma}
P  \left( \underset{1 \leq i \leq J}{\mbox{max}} \left|  \frac{\sum_{t \in \mathcal{T}_0}  W_{it}}{\sqrt{\sum_{t \in\mathcal{T}_0}  W_{it}^2}} \right|  >x \right) &\leq& C_1 J m^{3/2} x^{-1} \mbox{exp} \left(  -\frac{x^2}{2m} \right) + D_1 J T_0 \mbox{exp} \left(  -D_2 m^\tau \right).
 \end{eqnarray}

Setting $x = 2 \sqrt{m \mbox{log}(Jm^{3/2})}$, and using that $\mbox{log}(J) = o(T_0^{\frac{\tau}{3\tau + 1}})$ (Assumption \ref{Assumption_appendix_T}), we have that, for large enough $T_0$, $x < C_0 k^{1/6} \sqrt{m}$. Moreover, we can show that the right-hand-side of equation \ref{bound_lemma} is $o(1)$. Therefore, for some constant $\kappa$,
\begin{eqnarray}
\underset{1 \leq i \leq J}{\mbox{max}} \left|  \sum_{t \in \mathcal{T}_0}  \lambda_{ft} \epsilon_{it} \right| < \kappa [\mbox{log}(T_0 \vee J)]^{(1+\tau)/(2\tau)}\underset{1 \leq i \leq J}{\mbox{max}} \sqrt{\sum_{t \in\mathcal{T}_0}   \lambda_{ft}^2 \epsilon_{it}^2}
\end{eqnarray}
with probability $1-o(1)$. Under Assumption \ref{Assumption_appendix_technical}(a), we have that, with probability $1-o(1)$,
\begin{eqnarray}
\underset{1 \leq i \leq J}{\mbox{max}} \frac{1}{T_0} \left|  \sum_{t \in \mathcal{T}_0}  \lambda_{ft} \epsilon_{it} \right| &<& \kappa [\mbox{log}(T_0 \vee J)]^{(1+\tau)/(2\tau)} T_{0}^{-1/2}  \underset{1 \leq i \leq J}{\mbox{max}} \sqrt{\frac{1}{T_0} \sum_{t \in\mathcal{T}_0}   \lambda_{ft}^2 \epsilon_{it}^2} \\
&<& \kappa [\mbox{log}(T_0 \vee J)]^{(1+\tau)/(2\tau)} T_{0}^{-1/2} c.
\end{eqnarray}

It follows that $ \underset{1 \leq i \leq J}{\mbox{max}} \frac{1}{T_0} \left|  \sum_{t \in \mathcal{T}_0}  \lambda_{ft} \epsilon_{it} \right|  \buildrel p \over \rightarrow 0$ and, for $l_J$ defined in Assumption \ref{Assumption_appendix_technical}(b), $l_J \underset{1 \leq i \leq J}{\mbox{max}} \frac{1}{T_0} \left|  \sum_{t \in \mathcal{T}_0}  \lambda_{ft} \epsilon_{it} \right|  \buildrel p \over \rightarrow 0$. 

The proof that $ \frac{1}{T_0} \norm{\sum_{t \in \mathcal{T}_0} \lambda_{ft} \epsilon_{0t} }_\infty=o_p(1)$ and $l_J \frac{1}{T_0} \norm{\sum_{t \in \mathcal{T}_0} \lambda_{ft} \epsilon_{0t} }_\infty=o_p(1)$ follows the same steps as above, by setting $J=1$ and considering only $i=0$. The proof that $ \frac{1}{T_0} \norm{\sum_{t \in \mathcal{T}_0} \boldsymbol{\epsilon}_t \epsilon_{0t} }_\infty=o_p(1)$ and $l_J \frac{1}{T_0} \norm{\sum_{t \in \mathcal{T}_0} \boldsymbol{\epsilon}_t \epsilon_{0t} }_\infty=o_p(1)$ follows from noting that, since $\epsilon_{0t}$ and $\epsilon_{it}$ are $\beta$-mixing and independent, $\epsilon_{0t}\epsilon_{it}$ is also $\beta$-mixing (Theorem 5.2 from \cite{bradley2005}). Moreover, from Assumption \ref{Assumption_appendix_e}, we have, similar to equation \ref{Pena},
\begin{eqnarray}
\frac{  \left( \sum_{t \in H_p} \mathbb{E} \left| \tilde  \epsilon_{0t} \tilde \epsilon_{it}  \right|^2 \right)^\frac{1}{2} }{\left(  \sum_{t \in H_p} \mathbb{E} \left| \tilde  \epsilon_{0t} \tilde \epsilon_{it}  \right|^3 \right)^\frac{1}{3}    }   > k^\frac{1}{6}C_0 ,  
\end{eqnarray}
for some constant $C_0$. This inequality is valid for all $T_0$. Then we just follow the same steps of the proof setting $ W_{it} =  \epsilon_{0t}  \epsilon_{it}$  and $\widetilde W_{it} = \tilde \epsilon_{0t} \tilde \epsilon_{it}$. 
\end{proof}

 Given Lemma \ref{lemma}, note that Assumptions \ref{Assumption_appendix_T}, \ref{Assumption_appendix_e},  \ref{Assumption_appendix_mu} \ref{Assumption_appendix_lambda},  and \ref{Assumption_appendix_technical}(a) imply Assumptions  \ref{Assumption_e}, \ref{Assumption_mu}, \ref{Assumption_lambda},   and \ref{Assumption_technical}(a). Therefore, the results (i) and (ii) from Proposition \ref{Convergence_SC} remain valid under the assumptions considered in this appendix. 

We now present the following conditions in which $\hat  \alpha_{0t}^{\mbox{\tiny SC}}  \buildrel p \over \rightarrow \alpha_{0t} + \epsilon_{0t}$.
 
\begin{cor}
\label{appendix_corollary}
Suppose Assumptions \ref{Assumption_appendix_T}, \ref{Assumption_appendix_e},  \ref{Assumption_appendix_mu}, \ref{Assumption_appendix_lambda},  and \ref{Assumption_appendix_technical} hold. Then, for any $t \in \mathcal{T}_1$, $\hat  \alpha_{0t}^{\mbox{\tiny SC}}  \buildrel p \over \rightarrow \alpha_{0t} + \epsilon_{0t}$ when $T_0 \rightarrow \infty$.
 
\end{cor}

\begin{proof}

The proof follows exactly the same steps as the proof of Proposition \ref{Convergence_SC} for results (i) and (ii).  Now note that $\hat  \alpha_{0t}^{\mbox{\tiny SC}} = \alpha_{0t} + \epsilon_{0t} + \boldsymbol{\lambda}_t (\boldsymbol{\mu}_0 - \widehat{\boldsymbol{\mu}}_{\mbox{\tiny SC}}) - \boldsymbol{\epsilon}_t ' \widehat{\mathbf{w}}_{\mbox{\tiny SC}}$. Since we are considering a fixed sequence of $\boldsymbol{\lambda}_t$, and $\widehat{\boldsymbol{\mu}}_{\mbox{\tiny SC}}  \buildrel p \over \rightarrow \boldsymbol{\mu}_0$, it follows that $\boldsymbol{\lambda}_t (\boldsymbol{\mu}_0 - \widehat{\boldsymbol{\mu}}_{\mbox{\tiny SC}})  \buildrel p \over \rightarrow 0$. It remains to show that $ \boldsymbol{\epsilon}_t ' \widehat{\mathbf{w}}_{\mbox{\tiny SC}}  \buildrel p \over \rightarrow 0$.

From the proof of result (iii) from Proposition \ref{Convergence_SC}, we have that 
\begin{eqnarray}
\norm{\mathbf{E}  \widehat{\mathbf{w}}_{\mbox{\tiny SC}}}_2^2 &\leq & \norm{\mathbf{E}   \widehat{\mathbf{w}}_{\mbox{\tiny SC}}}_2^2 + \norm{\boldsymbol{\Lambda}(\boldsymbol{\mu}_0 - \widehat{\boldsymbol{\mu}}_{\mbox{\tiny SC}})}_2^2  \leq 2 \left| \boldsymbol{\epsilon}_0' \boldsymbol{\Lambda} (\boldsymbol{\mu}_0-\widehat{\boldsymbol{\mu}}_{\mbox{\tiny SC}})  \right| + 2 \left| \boldsymbol{\epsilon}_0' \boldsymbol{\Lambda} (\boldsymbol{\mu}_0- \boldsymbol{\mu}_J^\ast)  \right| \\
&& +  2 \left| (\boldsymbol{\mu}_0-\widehat{\boldsymbol{\mu}}_{\mbox{\tiny SC}})' \boldsymbol{\Lambda}' \mathbf{E}  \widehat{\mathbf{w}}_{\mbox{\tiny SC}} \right| + 2 \left| (\boldsymbol{\mu}_0- \boldsymbol{\mu}_J^\ast)' \boldsymbol{\Lambda}' \mathbf{E}  \mathbf{w}_J^\ast    \right|  + 2 \left| \boldsymbol{\epsilon}_0' \mathbf{E}  \widehat{\mathbf{w}}_{\mbox{\tiny SC}} \right| + \\
&&+ 2 \left| \boldsymbol{\epsilon}_0' \mathbf{E}  \mathbf{w}_J^\ast \right| +  \norm{\mathbf{E}  \mathbf{w}_J^\ast}_2^2 +  \norm{\boldsymbol{\Lambda}(\boldsymbol{\mu}_0 -  \boldsymbol{\mu}_J^\ast)}_2^2.
\end{eqnarray}

Now differently from what we do in the proof of Proposition \ref{Convergence_SC}, we can show using the assumptions considered in the appendix and Lemma \ref{lemma} that all the terms in the right hand side of the equation above are $o_p(1)$ when multiplied by $l_J/T_0$, for $l_J$ defined in Assumption \ref{Assumption_appendix_technical}(b). Therefore, $\frac{l_J}{T_0} \norm{ \mathbf{E}  \widehat{\mathbf{w}}_{\mbox{\tiny SC}}}_2^2 = o_p(1)$. Now using Assumption \ref{Assumption_appendix_technical}(b), we have that, for $t \in \mathcal{T}_1$, with probability $1-o(1)$, $(\boldsymbol{\epsilon}_t ' \widehat{\mathbf{w}}_{\mbox{\tiny SC}})^2 \leq \frac{l_J}{T_0} \norm{ \mathbf{E}  \widehat{\mathbf{w}}_{\mbox{\tiny SC}}}_2^2 = o_p(1)$, which completes the proof. \end{proof}

\subsubsection{Demeaned SC estimator} \label{Appendix_demeaned}

Consider the demeaned SC estimator, where the weights are estimated by
\begin{eqnarray} 
\widehat{\mathbf{w}}_{\mbox{\tiny SC$'$}} =  \underset{{\mathbf{w}}  \in \Delta^{J-1}}{\mbox{argmin}} \left \{ \frac{1}{T_0} \sum_{t \in \mathcal{T}_0} \left( y_{0t} - { \mathbf{w}} ' {\mathbf{y}}_t  - (\bar{{y}}_0 - \bar{\mathbf{y}} ' \mathbf{w}) \right)^2    \right\} =  \underset{{\mathbf{w}}  \in \Delta^{J-1}}{\mbox{argmin}} \left \{ \frac{1}{T_0} \sum_{t \in \mathcal{T}_0} \left( y_{0t} - { \mathbf{w}} ' {\mathbf{y}}_t \right)^2  - \left(\bar{{y}}_0 - \bar{\mathbf{y}} ' \mathbf{w} \right)^2    \right\},
\end{eqnarray}
where $\bar{{y}}_0 $ is the pre-treatment average of $y_{0t}$, and $\bar{\mathbf{y}}$ is the pre-treatment average of $\mathbf{y}_t$. We show that the results from Proposition \ref{Convergence_SC} remain valid for the demeaned SC estimator with minor adjustments in the proof. 

We start adjusting the objective function in equation \ref{eq_H} to
\begin{eqnarray} \label{eq_H_prime}  
\ddot{\mathcal{H}}_{J} (\boldsymbol{\mu}) =  \underset{{\mathbf{w}}  \in \Delta^{J-1}: { ~ {\mathbf{M}_J}'{\mathbf{w}}  = \boldsymbol{\mu}}}{\mbox{min}} \left \{   \frac{1}{T_0} \sum_{t \in \mathcal{T}_0} ( \bar {\lambda}_t(\boldsymbol{\mu}) -   { \mathbf{w}} ' {\boldsymbol{\epsilon}}_t )  ^2  -  \left( \bar{\boldsymbol{\lambda}}(\boldsymbol{\mu}_0 - \boldsymbol{\mu}) +\bar{{\epsilon}}_0 -\bar{\boldsymbol{\epsilon}}' \mathbf{w}  \right)^2 \right \},
 \end{eqnarray}
where  $\bar{\boldsymbol{\lambda}}$, $\bar{{\epsilon}}_0$, and $\bar{\boldsymbol{\epsilon}}$ are pre-treatment averages of, respectively, $\boldsymbol{\lambda}_t$, $\epsilon_{0t}$, and $\boldsymbol{\epsilon}_t$.  We add to Assumption \ref{Assumption_lambda} that $\bar{\boldsymbol{\lambda}} \rightarrow \boldsymbol{\omega}_0$, and that $\boldsymbol{\Omega} - \boldsymbol{\omega}_0'\boldsymbol{\omega}_0$ is positive definite.\footnote{The matrix $\boldsymbol{\Omega} - \boldsymbol{\omega}_0'\boldsymbol{\omega}_0$ will not be positive definite if there is a time-invariant common factor. However, we can redefine $\boldsymbol{\lambda}_t$ so that it does not include time-invariant common factors. Since  time-invariant common factors are eliminated in the demeaning process, this will not lead to any problem in the analysis.  } We define $\widetilde{\ddot{\mathcal{H}}}_{T_0} (\boldsymbol{\mu})$ as we do in equation \ref{Definition_H_tilde}.  We also redefine the function $\sigma^2_{\bar \lambda}(\boldsymbol{\mu})$ to $\ddot \sigma^2_{\bar \lambda}(\boldsymbol{\mu}) = (\boldsymbol{\mu}_0 - \boldsymbol{\mu})' (\boldsymbol{\Omega} - \boldsymbol{\omega}_0'\boldsymbol{\omega}_0) (\boldsymbol{\mu}_0 - \boldsymbol{\mu}) + \sigma^2_{\epsilon}$, which is uniquely minimized at $\boldsymbol{\mu}_0$.

Following similar steps to the proof of Proposition \ref{Convergence_SC}, we can construct an upper bound to $\ddot{\mathcal{H}}_{J} ( \boldsymbol{\mu}_0)$,
\begin{eqnarray}
\widetilde{\ddot{\mathcal{H}}}_{T_0} ( \boldsymbol{\mu}_0) &\leq & \widetilde{\ddot{\mathcal{H}}}_{J}^{UB} (\boldsymbol{\mu}_0) \equiv \frac{1}{T_0} \sum_{t \in \mathcal{T}_0} \bar {\lambda}_t(\boldsymbol{\mu}^\ast_J)^2 + \frac{1}{T_0} \sum_{t \in \mathcal{T}_0} ({\boldsymbol{\epsilon}}_t'\mathbf{w}^\ast_J)^2 - 2\frac{1}{T_0} \sum_{t \in \mathcal{T}_0} \bar {\lambda}_t(\boldsymbol{\mu}^\ast_J) ({\boldsymbol{\epsilon}}_t'\mathbf{w}^\ast_J)  +  \\
 && -  \left( \bar{\lambda}(\boldsymbol{\mu}_0 -  \boldsymbol{\mu}^\ast_J) +\bar{{\epsilon}}_0 -\bar{\boldsymbol{\epsilon}}' \mathbf{w}^\ast_J  \right)^2+ K \norm{\boldsymbol{\mu}_0-\boldsymbol{\mu}_J^\ast}_2  \buildrel p \over \rightarrow \sigma^2_{\epsilon}.
  \end{eqnarray}

We can also define $ \widetilde{\ddot{\mathcal{H}}}_{J}^{LB} (\boldsymbol{\mu}_0)$ such that 
\begin{eqnarray} 
\widetilde{\ddot{\mathcal{H}}}_{T_0} (\boldsymbol{\mu}) &\geq & \widetilde{\ddot{\mathcal{H}}}_{J}^{LB} (\boldsymbol{\mu})=  \underset{{\mathbf{w}}  \in \mathcal{W}}{\mbox{min}} \left \{   \frac{1}{T_0} \sum_{t \in \mathcal{T}_0} ( \bar {\lambda}_t(\boldsymbol{\mu}) -   { \mathbf{w}} ' {\boldsymbol{\epsilon}}_t )  ^2  -  \left( \bar{\lambda}(\boldsymbol{\mu}_0 - \boldsymbol{\mu}) +\bar{{\epsilon}}_0 -\bar{\boldsymbol{\epsilon}}' \mathbf{w}  \right)^2 \right \},
 \end{eqnarray}
and show that $\widetilde{\ddot{\mathcal{H}}}_{J}^{LB} (\boldsymbol{\mu})$ converges uniformly in probability to $\ddot \sigma^2_{\bar \lambda}(\boldsymbol{\mu})$. We just have to show that the function $  \underset{{\mathbf{w}}  \in \mathcal{W}}{\mbox{min}} \left \{   \frac{1}{T_0} \sum_{t \in \mathcal{T}_0} ( \bar {\lambda}_t(\boldsymbol{\mu}) -   { \mathbf{w}} ' {\boldsymbol{\epsilon}}_t )  ^2  -  \left( \bar{\lambda}(\boldsymbol{\mu}_0 - \boldsymbol{\mu}) +\bar{{\epsilon}}_0 -\bar{\boldsymbol{\epsilon}}' \mathbf{w}  \right)^2 \right \}$ is Lipschitz as we do in equation \ref{eq_lipschitz}, and then use the same strategy as we do in equation \ref{Appendix_inequality} to show pointwise convergence. 

For the third result, we use that $\norm{\boldsymbol{\epsilon}_0  + \boldsymbol{\lambda}(\boldsymbol{\mu}_0 - \widehat{\boldsymbol{\mu}}_{\mbox{\tiny SC$'$}}) - \mathbf{E}  \widehat{\mathbf{w}}_{\mbox{\tiny SC$'$}} - (\bar{{\epsilon}}_0  + \bar{\boldsymbol{\lambda}}(\boldsymbol{\mu}_0 - \widehat{\boldsymbol{\mu}}_{\mbox{\tiny SC$'$}}) - \bar{\boldsymbol{\epsilon}}  \widehat{\mathbf{w}}_{\mbox{\tiny SC$'$}}) }_2^2 \leq \norm{\boldsymbol{\epsilon}_0  + \boldsymbol{\lambda} (\boldsymbol{\mu}_0 - \boldsymbol{\mu}^\ast_J) - \mathbf{E}  \mathbf{w}^\ast_J - (\bar{{\epsilon}}_0  + \bar{\boldsymbol{\lambda}}(\boldsymbol{\mu}_0 -  \boldsymbol{\mu}_J^\ast) - \bar{\boldsymbol{\epsilon}}  \mathbf{w}_J^\ast)}_2^2 $. Then we can follow the same steps as in the proof of part (iii) of Proposition \ref{Convergence_SC}.

\subsubsection{SC specifications with covariates}  \label{Appendix_covariates}

\textbf{Theoretical results}

 Consider now the case with covariates. The model for potential outcomes are now given by 
\begin{eqnarray} \label{model2}
\begin{cases} y_{it}^N =  \boldsymbol{\lambda}_t \boldsymbol{\mu}_i + \boldsymbol{\theta}_t \mathbf{z}_i + \epsilon_{it}  \\ 
y_{it}^I = \alpha_{it} + y_{it}^N, \end{cases}
\end{eqnarray}
where $\mathbf{z}_i$ is a $q \times 1$ vector of  observed time-invariant covariates, and $\boldsymbol{\theta}_t$ are unobserved time-varying effects.

Define $\boldsymbol{\rho}_i \equiv (\boldsymbol{\mu}_i',\mathbf{z}_i')'$ and $\boldsymbol{\gamma}_t \equiv (\boldsymbol{\lambda}_t, \boldsymbol{\theta}_t)$, and assume that Assumptions \ref{Assumption_mu}, \ref{Assumption_lambda}, and  \ref{Assumption_technical} are valid for $\boldsymbol{\rho}_i$ and $\boldsymbol{\gamma}_t$ instead of $\boldsymbol{\mu}_i$ and $\boldsymbol{\lambda}_t$.\footnote{\cite{FB} discuss cases in which the sequence $\boldsymbol{\gamma}_t$ might be multicolinear. In this case, Assumption \ref{Assumption_lambda} would not be valid for $\boldsymbol{\gamma}_t$, but the SC estimator would still control for the effects of these observed and unobserved covariates. While, for simplicity, we focus on the case in which Assumption \ref{Assumption_lambda} is valid for $\boldsymbol{\gamma}_t$, the same conclusions from \cite{FB} apply here.} In this case, if we also have Assumption \ref{Assumption_e}, then Proposition \ref{Convergence_SC} is valid for the model $y_{it}^N =  \boldsymbol{\gamma}_t \boldsymbol{\rho}_i + \epsilon_{it}$, where we treat the observed variables $\mathbf{z}_i$ as unobserved factor loadings, as  \cite{FB} do. Therefore, the SC weights using all pre-treatment outcomes as predictors will be such that  $\widehat{\boldsymbol{\mu}}_{\mbox{\tiny SC}}  \buildrel p \over \rightarrow \boldsymbol{\mu}_0$ and $\widehat{\mathbf{z}}_{\mbox{\tiny SC}}  \equiv \mathbf{Z}_J'  \widehat{\mathbf{w}}_{\mbox{\tiny SC}}  \buildrel p \over \rightarrow \mathbf{z}_0$, where $\mathbf{Z}_J$ is the $J \times q$ matrix with information on the covariates $\mathbf{z}_i$ for all the controls. 

Now consider an alternative SC specification, where the SC weights are estimated using both pre-treatment outcomes and the covariates as predictors. \cite{Abadie2003} suggest a nested minimization problem where in the first step we select an $(R \times 1)$ vector of predictors for the treated unit, $\mathbf{x}_0$, and the corresponding $(R \times J)$ matrix of predictors of the control units, $\mathbf{X}_1$.  The rows of these matrices may include functions of pre-treatment outcomes, and observed covariates. In the first step, they propose choosing weights that minimize the distance $\norm{\mathbf{x}_0-\mathbf{X}_1 \mathbf{w}}_\mathbf{V}$ for a given positive semi-definite matrix $\mathbf{V}$.  Then the matrix $\mathbf{V}$ is chosen to minimize mean squared errors of the difference between the pre-treatment outcomes and the weighted average of the control outcomes with weights  $\mathbf{w}(\mathbf{V}) \in \mbox{argmin}_{\mathbf{w} \in \Delta^{J-1}}  \norm{\mathbf{x}_0-\mathbf{X}_1 \mathbf{w}}_\mathbf{V}$. Therefore,   we can re-write this problem as $\widehat{\mathbf{w}}_{cov}= \underset{{\mathbf{w}}  \in \Theta_J }{\mbox{argmin}} \left \{ \frac{1}{T_0} \sum_{t \in \mathcal{T}_0} \left( y_{0t} - { \mathbf{w}} ' {\mathbf{y}}_t  \right)^2    \right\}$, where $\Theta_J = \{ \widetilde{\mathbf{w}} | \widetilde{\mathbf{w}} \in \mbox{argmin}_{\mathbf{w} \in \Delta^{J-1}}  \norm{\mathbf{x}_0-\mathbf{X}_1\mathbf{w}}_\mathbf{V} \mbox{ for some } \mathbf{V} \}$.
We can also consider a time frame for the minimization problem that defines $\widehat{\mathbf{w}}_{cov}$ different from the time frame for the minimization $\norm{\mathbf{x}_0-\mathbf{X}_1 \mathbf{w}}_\mathbf{V}$, as suggested by \cite{Abadie2015}. Our results still apply, as long as the number of periods in both minimization problems goes to infinity. As we do in Section \ref{Section_original_SC}, we define 
\begin{eqnarray} \label{H_cov}
\widetilde{\mathcal{H}}_J^{cov} (\boldsymbol{\rho})=  \underset{ \widetilde{\boldsymbol{\rho}} \in \Gamma_J }{\mbox{min}} \left\{  {\underset{{\mathbf{w}}  \in \Theta_J : { ~ {[\mathbf{M}_J}'   ~  \mathbf{Z}_J' ] {\mathbf{w}}  = \widetilde{\boldsymbol{\rho}}}}{\mbox{min}} \left \{   \frac{1}{T_0} \sum_{t \in \mathcal{T}_0} ( \bar \gamma_t(\widetilde{\boldsymbol{\rho}}) -   { \mathbf{w}} ' {\boldsymbol{\epsilon}}_t )  ^2 \right \}} + K \norm{\boldsymbol{\rho} - \widetilde{\boldsymbol{\rho}}}_2  \right\},
\end{eqnarray} 
where $\Gamma_J$ is the set of $\boldsymbol{\rho}$ that are attainable with the specification that includes covariates when there are $J$ control units, and $\bar \gamma_t(\boldsymbol{\rho}) = \boldsymbol{\gamma}_t (\boldsymbol{\rho}_0 - \boldsymbol{\rho}) + \epsilon_{0t}$.  As before, we have that $\widehat{\boldsymbol{\rho}}_{cov} = \mbox{argmin}_{\boldsymbol{\rho} \in \tilde \Gamma} \widetilde{\mathcal{H}}_J^{cov} (\boldsymbol{\rho})$, where $\tilde \Gamma = \mbox{cl} \left(\cup_{J \in \mathbb{N}} \Gamma_J \right)$.

Assume that the pre-treatment outcomes included as predictors when there are $T_0$ pre-treatment periods are such that Assumptions \ref{Assumption_lambda} and  \ref{Assumption_technical} are satisfied if we consider only the pre-treatment periods used as predictors. Therefore, it must be that the number of pre-treatment outcomes used as predictors goes to infinity when $T_0 \rightarrow \infty$. In this case, we show that the implied estimators for $\boldsymbol{\rho}_0 = (\boldsymbol{\mu}_0',\mathbf{z}_0')'$ are such that $\widehat{\boldsymbol{\mu}}_{cov}  \buildrel p \over \rightarrow \boldsymbol{\mu}_0$ and $ \widehat{\mathbf{z}}_{cov}  \buildrel p \over \rightarrow \mathbf{z}_0$. First, consider a sequence of diagonal matrices $\mathbf{V}_J$ where the diagonal elements are equal to one for the pre-treatment outcomes, and zero otherwise. In this case, the minimization problem $\norm{\mathbf{x}_0-\mathbf{X}_1 \mathbf{w}}_{\mathbf{V}_J}$ is equivalent to the problem analyzed in Section \ref{Section_original_SC}, and  satisfies the conditions from Proposition \ref{Convergence_SC}. Therefore,  $\widehat{\boldsymbol{\mu}}(\mathbf{V}_J) \equiv {\mathbf{M}_J}' \widehat{\mathbf{w}}(\mathbf{V}_J)  \buildrel p \over \rightarrow \boldsymbol{\mu}_0 $ and $\widehat{\mathbf{z}}(\mathbf{V}_J) \equiv {\mathbf{Z}_J}' \widehat{\mathbf{w}}(\mathbf{V}_J)  \buildrel p \over \rightarrow \mathbf{z}_0 $.  Since these are candidate solutions for the minimization problem \ref{H_cov}, it follows that  
 \begin{eqnarray} 
 \widetilde{\mathcal{H}}_J^{cov} (\boldsymbol{\rho}_0) \leq    \left \{   \frac{1}{T_0} \sum_{t \in \mathcal{T}_0} ( \bar \gamma_t(\widehat{\boldsymbol{\rho}}_J(\mathbf{V}_J)) -   \widehat{\mathbf{w}}(\mathbf{V}_J) ' {\boldsymbol{\epsilon}}_t )  ^2 \right \} + K\norm{\boldsymbol{\rho}_0 - \widehat{\boldsymbol{\rho}}(\mathbf{V}_J)}_2  \buildrel p \over \rightarrow \mbox{plim} \frac{1}{T_0} \sum_{t \in \mathcal{T}_0}  \bar \gamma_t^2(\boldsymbol{\rho}_0).  
\end{eqnarray}

 Also, similarly to the proof of Proposition \ref{Convergence_SC}, $ \widetilde{\mathcal{H}}_J^{cov} (\boldsymbol{\rho}) $ is bounded from below by a function that converges uniformly to $ \mbox{plim} \frac{1}{T_0} \sum_{t \in \mathcal{T}_0}  \bar \gamma_t^2(\boldsymbol{\rho})$, which is uniquely minimized at $\boldsymbol{\rho}_0$. Following the same steps of the proof of Proposition \ref{Convergence_SC}, it follows that $\widehat{\boldsymbol{\mu}}_{cov}  \buildrel p \over \rightarrow \boldsymbol{\mu}_0$ and $ \widehat{\mathbf{z}}_{cov}  \buildrel p \over \rightarrow \mathbf{z}_0$.

\textbf{Monte Carlo simulations}

We present an MC  exercise similar to the one presented in Section \ref{MC}, but with the inclusion of observed covariates. We continue to consider a setting with $\boldsymbol{\lambda}_t = [\lambda_{1t} ~ \lambda_{2t}]$, but now we also add two observed covariates $\mathbf{z}_i = [z_{1i} ~ z_{2i}]'$, with effects that vary with time, $\theta_{1t} \sim N(0,1)$ and $\theta_{2t} \sim N(0,1)$. The treated unit has $\boldsymbol{\mu}_0 = (1,0)$, and  $\mathbf{z}_0 = (1,0)$. The control units are divided in four groups of equal size, with a combination of $\boldsymbol{\mu}_i \in \{ (1,0), (0,1) \}$ and  $\mathbf{z}_i \in \{ (1,0), (0,1) \}$. Therefore, the goal of the SC estimator is to allocate positive weights only to the control units with $\boldsymbol{\mu}_i = (1,0)$, and  $\mathbf{z}_i = (1,0)$.

We consider three specifications for the SC estimator. The first one uses all pre-treatment periods as predictors, the second one includes the first half of the pre-treatment outcomes and the two covariates as predictors, and the third one includes the pre-treatment outcome average and the two covariates as predictors. We present in Appendix Table \ref{Table_MC_covariates} the implied estimators for $\mu_{10}$ and $z_{10}$ (given the adding-up constraint, $\hat \mu_{20} =  1-\hat\mu_{10}$ and $\hat z_{20} =  1-\hat z_{10}$). Overall, we see that the first two specifications perform very similarly. Moreover, when $J$ and $T_0$ gets large, the SC unit using either of these two specifications control well both for the unobserved factor loading $\boldsymbol{\mu}$ and the observed covariates $\mathbf{z}$.  The third specification, which does not satisfy the properties considered in the theory presented in this section, does a better job in matching the observed covariates, but does a poor job in matching the unobserved factor loadings. Even when $J,T_0 \rightarrow \infty$, the estimator for $\mu_{10}$ remains roughly constant around $0.67$, which suggests that the implied estimator for the factor loadings will not generally be consistent if we consider such specification to estimate the SC weights.   

\begin{center}

[Appendix Table \ref{Table_MC_covariates} here]

\end{center}

\subsubsection{Simple example with $F=1$} \label{Appendix_simple}

Consider a simple example in which $y_{it} = {\lambda}_t + \epsilon_{it}$ for all $i=0,...,J$. Assume ${\lambda}_t  \buildrel iid \over \sim N(0,\sigma^2_\lambda)$, $\epsilon_{jt}  \buildrel iid \over \sim N(0,\sigma^2)$, and that these variables are all independent from each other. Finally, assume that $J/T_0 \rightarrow  c \in [0,1)$. Let $\mathbf{y}_i$ and $\boldsymbol{\epsilon}_i$ be the $T_0 \times 1$ vectors with the outcomes and errors of unit $i$. Then $\mathbf{y}_0 = \sum_{i=1}^J \frac{1}{J} \mathbf{y}_i + \boldsymbol{\epsilon}_0 - \sum_{i=1}^J \frac{1}{J}\boldsymbol{\epsilon}_i$. Using Using Frisch-Waugh-Lovell theorem, we have that the OLS estimator associated to unit $i$, $b_i$, is given by 
\begin{eqnarray}
b_i &=& \left( \mathbf{y}_i' \mathbf{Q}_{(i)} \mathbf{y}_i \right)^{-1}  \left( \mathbf{y}_i' \mathbf{Q}_{(i)} \mathbf{y}_0 \right) = \frac{1}{J} + \left( \mathbf{y}_i' \mathbf{Q}_{(i)} \mathbf{y}_i \right)^{-1}  \left( \mathbf{y}_i' \mathbf{Q}_{(i)} \left( \boldsymbol{\epsilon}_0 - \sum_{i'=1}^J \frac{1}{J}\boldsymbol{\epsilon}_{i'}  \right) \right) \\
&=& \frac{1}{J} + \frac{\boldsymbol{\Lambda}' \mathbf{Q}_{(i)} \boldsymbol{\epsilon}_0 + \boldsymbol{\epsilon}_i' \mathbf{Q}_{(i)} \boldsymbol{\epsilon}_0 - \boldsymbol{\Lambda}' \mathbf{Q}_{(i)}  \sum_{i'=1}^J \frac{1}{J}\boldsymbol{\epsilon}_{i'}- \boldsymbol{\epsilon}_i' \mathbf{Q}_{(i)}  \sum_{i'=1}^J \frac{1}{J}\boldsymbol{\epsilon}_{i'} }{\boldsymbol{\Lambda}' \mathbf{Q}_{(i)} \boldsymbol{\Lambda} + 2\boldsymbol{\Lambda}' \mathbf{Q}_{(i)} \boldsymbol{\epsilon}_i + \boldsymbol{\epsilon}_i'\mathbf{Q}_{(i)} \boldsymbol{\epsilon}_i},
\end{eqnarray}
where $\mathbf{Q}_{(i)}$ is the residual-maker matrix of a regression on $\{ \mathbf{y}_1, \hdots,  \mathbf{y}_J \} \setminus \{ \mathbf{y}_i \} $. Let $A_i = \boldsymbol{\Lambda}' \mathbf{Q}_{(i)} \boldsymbol{\epsilon}_0 + \boldsymbol{\epsilon}_i' \mathbf{Q}_{(i)} \boldsymbol{\epsilon}_0 - \boldsymbol{\Lambda}' \mathbf{Q}_{(i)}  \sum_{i'=1}^J \frac{1}{J}\boldsymbol{\epsilon}_{i'}- \boldsymbol{\epsilon}_i' \mathbf{Q}_{(i)}  \sum_{i'=1}^J \frac{1}{J}\boldsymbol{\epsilon}_{i'} $ and $B_i = \boldsymbol{\Lambda}' \mathbf{Q}_{(i)}\boldsymbol{\Lambda} + 2\boldsymbol{\Lambda}' \mathbf{Q}_{(i)} \boldsymbol{\epsilon}_i + \boldsymbol{\epsilon}_i'\mathbf{Q}_{(i)} \boldsymbol{\epsilon}_i $. 

From Proposition \ref{Prop_OLS_K_large}, we know that $\sum_{i=1}^J b_i  \buildrel p \over \rightarrow 1$, which implies that $\sum_{i=1}^J \frac{A_i}{B_i}  \buildrel p \over \rightarrow 0$. We want to derive the probability limit of $\mathbf{b}'\mathbf{b} = \sum_{i=1}^J b_i^2$, which is given by 
\begin{eqnarray}
\sum_{i=1}^J b_i^2 = \frac{1}{J} + 2 \frac{1}{J} \sum_{i=1}^J \frac{A_i}{B_i} +\sum_{i=1}^J\left( \frac{A_i}{B_i} \right)^2 = \sum_{i=1}^J\left( \frac{A_i}{B_i} \right)^2 + o_p(1).
\end{eqnarray}

Now note that 
\begin{eqnarray} \label{bounds}
\frac{1}{\mbox{max}_{i=1,...,J} \left\{  \left(\frac{1}{K} B_i \right)^2 \right\}}  \sum_{i=1}^J\left(\frac{1}{K} {A_i} \right)^2 \leq \sum_{i=1}^J\left( \frac{A_i}{B_i} \right)^2 \leq \frac{1}{\mbox{min}_{i=1,...,J}\left\{  \left(\frac{1}{K} B_i \right)^2 \right\}} \sum_{i=1}^J\left( \frac{1}{K} {A_i} \right)^2, 
\end{eqnarray}
where $K = T_0 - J + 1$.

We first show that ${\mbox{min}_{i=1,...,J}\left\{  \left(\frac{1}{K} B_i \right)^2 \right\}}$ and ${\mbox{max}_{i=1,...,J}\left\{  \left(\frac{1}{K} B_i \right)^2 \right\}}$ converge in probability to $\sigma^4$. We start with the term $\frac{1}{K}\boldsymbol{\Lambda}' \mathbf{Q}_{(i)}  \boldsymbol{\Lambda}$.  We can write  $\boldsymbol{\Lambda} = \mathbf{Y}_{(i)} \boldsymbol{\phi} + \mathbf{u}_{(i)}$, where $\mathbf{Y}_{(i)}$ is a matrix with information on  $\{ \mathbf{y}_1, \hdots,  \mathbf{y}_J \} \setminus \{ \mathbf{y}_i \} $, and $\boldsymbol{\phi}$ is the population OLS parameters of a regression of $\lambda_t$ on  $\{ {y}_{1t}, \hdots,  {y}_{Jt} \} \setminus \{ {y}_{it} \} $. Given that the data is iid normal, we have that  $\mathbf{u}_{(i)} |  \mathbf{Y}_{(i)} \sim N(0 , \sigma_u^2 \mathbb{I}_{T_0})$. Moreover, it is easy to show that $\sigma_u^2=o(1)$ and $J \sigma_u^2 = O(1)$. The intuition is that with the average of many observations $y_{it}$ across $i$ we become close to $\lambda_t$, so  the variance of the error in this population OLS regression goes to zero when $J \rightarrow \infty$. 

Therefore,  conditional on $\mathbf{Y}_{(i)}$, $ \frac{\boldsymbol{\Lambda}' \mathbf{Q}_{(i)}  \boldsymbol{\Lambda}}{\sigma_u^2} = \frac{\mathbf{u}_{(i)}'\mathbf{Q}_{(i)}  \mathbf{u}_{(i)}}{\sigma_u^2} \sim \chi^2_K$, which implies that $\mathbb{E} \left[ \frac{1}{K}\frac{\boldsymbol{\Lambda}' \mathbf{Q}_{(i)}  \boldsymbol{\Lambda}}{\sigma_u^2} \right]= 1$ and $\mathbb{E} \left[ \left( \frac{1}{K}\frac{\boldsymbol{\Lambda}' \mathbf{Q}_{(i)}  \boldsymbol{\Lambda}}{\sigma_u^2} \right)^2 \right]= O(1)$. Given that, for any $e>0$,
\begin{eqnarray}
P\left( \underset{{i=1,...,J}}{\mbox{max}}  \left| \frac{1}{K}{\boldsymbol{\Lambda}' \mathbf{Q}_{(i)}  \boldsymbol{\Lambda}} \right|  > e \right) \leq J P\left(\left|  \frac{1}{K}\frac{\boldsymbol{\Lambda}' \mathbf{Q}_{(i)}  \boldsymbol{\Lambda}}{\sigma_u^2}\right|  > \frac{e}{\sigma_u^2} \right) \leq  \sigma_u^2 (J \sigma_u^2)    \frac{\mathbb{E} \left[ \left( \frac{1}{K}\frac{\boldsymbol{\Lambda}' \mathbf{Q}_{(i)}  \boldsymbol{\Lambda}}{\sigma_u^2} \right)^2 \right]}{e^2},
\end{eqnarray}
where $\sigma_u^2 = o(1)$ and the other two terms are $O(1)$, which implies that $\underset{{i=1,...,J}}{\mbox{max}}  \left| \frac{1}{K}{\boldsymbol{\Lambda}' \mathbf{Q}_{(i)}  \boldsymbol{\Lambda}} \right|  \buildrel p \over \rightarrow  0$. Now since $\boldsymbol{\epsilon}_i$ is independent from  $\{ \mathbf{y}_1, \hdots,  \mathbf{y}_J \} \setminus \{ \mathbf{y}_i \} $, we have $\frac{\boldsymbol{\epsilon}_i'\mathbf{Q}_{(i)} \boldsymbol{\epsilon}_i}{\sigma^2} \sim \chi^2_K$, which implies that $\frac{1}{K^2}\mathbb{E} \left[ \left( {\boldsymbol{\epsilon}_i'\mathbf{Q}_{(i)} \boldsymbol{\epsilon}_i} - K \sigma^2 \right)^4 \right] = O(1)$. 
Therefore, for any $e>0$,
\begin{eqnarray}
P\left( \underset{{i=1,...,J}}{\mbox{max}}  \left| \frac{1}{K} {\boldsymbol{\epsilon}_i'\mathbf{Q}_{(i)} \boldsymbol{\epsilon}_i} - \sigma^2  \right|  > e \right) &\leq& J P\left(  \left| \frac{1}{K} {\boldsymbol{\epsilon}_i'\mathbf{Q}_{(i)} \boldsymbol{\epsilon}_i} - \sigma^2 \right|  > e \right) \\ \nonumber
& \leq & \frac{1}{K}\frac{J}{K} \frac{\frac{1}{K^2} \mathbb{E}\left[ \left|  {\boldsymbol{\epsilon}_i'\mathbf{Q}_{(i)} \boldsymbol{\epsilon}_i} - K \sigma^2  \right| ^4 \right]}{e^4}  = o(1)O(1)O(1),
\end{eqnarray}
which implies that $\underset{{i=1,...,J}}{\mbox{max}}  \left| \frac{1}{K} {\boldsymbol{\epsilon}_i'\mathbf{Q}_{(i)} \boldsymbol{\epsilon}_i}   \right| \buildrel p \over \rightarrow \sigma^2$. 

Finally, consider the term $\boldsymbol{\Lambda}' \mathbf{Q}_{(i)} \boldsymbol{\epsilon}_i = \sum_{1=1}^K \tilde u_{q(i)}\tilde \epsilon_{iq}$, where $\tilde u_{q(i)}  \buildrel iid \over \sim N(0, \sigma_u^2)$, and  $\tilde \epsilon_{iq}  \buildrel iid \over \sim N(0, \sigma^2)$. Moreover, $\tilde u_{q(i)} $ and $\tilde \epsilon_q$ are independent. Therefore, $\mathbb{E}[\tilde u_{q(i)}\tilde \epsilon_{iq}]=0$, and $var[\tilde u_{q(i)}\tilde \epsilon_{iq}]=\sigma_u^2 \sigma^2$. Therefore,
\begin{eqnarray}
P\left( \underset{{i=1,...,J}}{\mbox{max}}  \left| \frac{1}{K} \sum_{q=1}^{K} \tilde u_{q(i)}\tilde \epsilon_{iq} \right|  > e \right) &\leq& J P\left(  \left| \frac{1}{K} \sum_{q=1}^{K}\tilde u_{q(i)}\tilde \epsilon_{iq}\right|  > e \right) \leq  \frac{J}{K} \frac{ \sigma_u^2 \sigma^2}{e^2} = o(1),
\end{eqnarray}
because $J/K=O(1)$ and $\sigma_u^2 = o(1)$. 

Likewise, we can do the same calculations for ${\mbox{min}_{i=1,...,J}\left\{  \left(\frac{1}{K} B_i \right)^2 \right\}}$. Combining all these results, we have that ${\mbox{min}_{i=1,...,J}\left\{  \left(\frac{1}{K} B_i \right)^2 \right\}}$ and ${\mbox{max}_{i=1,...,J}\left\{  \left(\frac{1}{K} B_i \right)^2 \right\}}$ converge in probability to $\sigma^4$.

Now we consider $\sum_{i=1}^J \left( \frac{1}{K}A_i \right)^2$. Consider first $\sum_{i=1}^J \left(\frac{1}{K} \boldsymbol{\epsilon}_i ' \mathbf{Q}_{(i)} \boldsymbol{\epsilon}_0 \right)^2 = \sum_{i=1}^J \left(\frac{1}{K} \sum_{q=1}^K \tilde \epsilon_{iq} \tilde \epsilon_{0(i)q} \right)^2 = \frac{J}{K} \frac{1}{J} \sum_{i=1}^J \left(\frac{1}{\sqrt{K}} \sum_{q=1}^K \tilde \epsilon_{iq} \tilde \epsilon_{0(i)q} \right)^2$, where  $\tilde \epsilon_{iq}  \buildrel iid \over \sim N(0, \sigma^2)$, and $\tilde \epsilon_{0(i)q}  \buildrel iid \over \sim N(0, \sigma^2)$. Note that $\mathbb{E} \left[\left(\frac{1}{\sqrt{K}} \sum_{q=1}^K \tilde \epsilon_{iq} \tilde \epsilon_{0(i)q} \right)^2\right] = var \left[\left(\frac{1}{\sqrt{K}} \sum_{q=1}^K \tilde \epsilon_{iq} \tilde \epsilon_{0(i)q} \right) \right] = \sigma^4$. Moreover, we also have that $var \left( \frac{1}{J} \sum_{i=1}^J \left(\frac{1}{\sqrt{K}} \sum_{q=1}^K \tilde \epsilon_{iq} \tilde \epsilon_{0(i)q} \right)^2 \right) \rightarrow 0$, which implies that $\frac{J}{K} \frac{1}{J} \sum_{i=1}^J \left(\frac{1}{\sqrt{K}} \sum_{q=1}^K \tilde \epsilon_{iq} \tilde \epsilon_{0(i)q} \right)^2  \buildrel p \over \rightarrow \frac{c}{1-c} \sigma^4$.
Using similar calculations, all the other terms in the numerator converge in probability to zero.  

Therefore, both the upper and the lower bounds from equation \ref{bounds} converge in probability to $\frac{c}{1-c}$, which implies that  $\mathbf{b}'\mathbf{b}  \buildrel p \over \rightarrow \frac{c}{1-c}$. Now note that, for any $t \in \mathcal{T}_1$, $\hat \alpha_{0t} = \alpha_{0t} + \lambda_t(1 - \widehat{{\mu}}_{\mbox{\tiny OLS}}) + \epsilon_{0t} - \boldsymbol{\epsilon}_t ' \mathbf{b}$. From Proposition \ref{Prop_OLS_K_large}, $ \lambda_t(1 - \widehat{{\mu}}_{\mbox{\tiny OLS}})  \buildrel p \over \rightarrow 0$. Since data is iid normal across time, we have that $\epsilon_{0t} - \boldsymbol{\epsilon}_t ' \mathbf{b} | \mathbf{b} \sim N \left(0, \sigma^2 \left( 1 + \mathbf{b}'\mathbf{b} \right)  \right)$. Since $\mathbf{b}'\mathbf{b}  \buildrel p \over \rightarrow \frac{c}{1-c}$, it follows that $\epsilon_{0t} - \boldsymbol{\epsilon}_t ' \mathbf{b}  \buildrel d \over \rightarrow N \left(0, \frac{\sigma^2}{1-c}   \right)$, which implies that $\hat \alpha_{0t}  \buildrel d \over \rightarrow N \left(\alpha_{0t}, \frac{\sigma^2}{1-c}   \right)$.

\renewcommand{\refname}{References} 

\bibliographystyle{apalike}
\bibliography{bib/bib.bib}

\pagebreak

\begin{landscape}
\begin{table}[h!]
  \centering
\caption{{\bf Monte Carlo Simulations}} \label{Table_MC}
      \begin{lrbox}{\tablebox}
\begin{tabular}{ccccccccccccccc}
\hline
\hline

 & \multicolumn{4}{c}{SC }  & & \multicolumn{4}{c}{Unrestricted  }   & & \multicolumn{4}{c}{OLS with }  \\
&  \multicolumn{4}{c}{estimator} & & \multicolumn{4}{c}{OLS} & & \multicolumn{4}{c}{adding-up constraint}  \\  \cline{2-5} \cline{7-10} \cline{12-15}
 
$J$  &4 & 10 & 50 & 100 & &4 & 10 & 50 & 100  & &4 & 10 & 50 & 100  \\

     & (1) &  (2) & (3) & (4) & & (5) & (6) & (7) & (8)  & & (9) & (10) & (11) & (12) \\ 
     
     \hline
     
\multicolumn{15}{c}{Panel A: $T_0=J + 5$} \\
\\
$\mathbb{E}[\hat \mu_{01}]$ & 0.760 & 0.817 & 0.905 & 0.929 &  & 0.653 & 0.816 & 0.962 & 0.976 &  & 0.829 & 0.910 & 0.982 & 0.989 \\
$se[\hat \mu_{01}]$ & 0.206 & 0.156 & 0.076 & 0.054 &  & 0.489 & 0.516 & 0.501 & 0.506 &  & 0.319 & 0.324 & 0.320 & 0.325 \\
 \\
$\mathbb{E}[\hat \mu_{02}]$ & 0.240 & 0.183 & 0.095 & 0.071 &  & -0.002 & -0.002 & -0.005 & 0.010 &  & 0.171 & 0.090 & 0.018 & 0.011 \\
$se[\hat \mu_{02}]$ & 0.206 & 0.156 & 0.076 & 0.054 &  & 0.498 & 0.509 & 0.497 & 0.518 &  & 0.319 & 0.324 & 0.320 & 0.325 \\
 \\
$se(\hat \alpha)$ & 1.288 & 1.194 & 1.084 & 1.073 &  & 1.586 & 1.984 & 3.791 & 5.220 &  & 1.486 & 1.806 & 3.437 & 4.661 \\
 \\

\multicolumn{15}{c}{Panel B: $T_0=2 \times J$} \\
\\
$\mathbb{E}[\hat \mu_{01}]$ & 0.753 & 0.831 & 0.922 & 0.944 &  & 0.637 & 0.828 & 0.960 & 0.982 &  & 0.825 & 0.915 & 0.981 & 0.991 \\
$se[\hat \mu_{01}]$ & 0.217 & 0.136 & 0.057 & 0.040 &  & 0.569 & 0.343 & 0.143 & 0.103 &  & 0.354 & 0.231 & 0.100 & 0.072 \\
 \\
$\mathbb{E}[\hat \mu_{02}]$ & 0.247 & 0.169 & 0.078 & 0.056 &  & 0.001 & 0.003 & 0.000 & 0.001 &  & 0.175 & 0.085 & 0.019 & 0.009 \\
$se[\hat \mu_{02}]$ & 0.217 & 0.136 & 0.057 & 0.040 &  & 0.582 & 0.335 & 0.143 & 0.102 &  & 0.354 & 0.231 & 0.100 & 0.072 \\
 \\
$se(\hat \alpha)$ & 1.297 & 1.186 & 1.050 & 1.047 &  & 1.798 & 1.586 & 1.420 & 1.444 &  & 1.571 & 1.519 & 1.411 & 1.437 \\

\hline

\end{tabular}
   \end{lrbox}
\usebox{\tablebox}\\
\settowidth{\tableboxwidth}{\usebox{\tablebox}} \parbox{\tableboxwidth}{\footnotesize{Notes: this table presents the expected value and the standard error of the estimators for $\boldsymbol{\mu}_0 = (\mu_{01}, {\mu}_{02})$.  It also presents the standard error of $\hat \alpha$. Since $\mathbb{E}[{\lambda}_t]=0$, $\mathbb{E}[\hat \alpha_{01}] =0$, which is the true treatment effect. Panel A presents results with  $T_0=J + 5$, while Panel B presents results with  $T_0=2\times J$. Columns 1 to 4 present the results using the SC estimator to estimate the weights. Columns 5 to 8 present results using OLS estimator with no constraint. Columns 9 to 12 present results using OLS estimator with adding-up constraint. Results based on 5000 simulations.  The DGP is described in detail in Section \ref{MC}.

}
}
\end{table}
\end{landscape}

 \setcounter{table}{0}
\renewcommand\thetable{A.\arabic{table}}

\setcounter{figure}{0}
\renewcommand\thefigure{A.\arabic{figure}}

\pagebreak

\begin{landscape}
\begin{table}[h!]
  \centering
\caption{{\bf Monte Carlo Simulations with Covariates}} \label{Table_MC_covariates}
      \begin{lrbox}{\tablebox}
\begin{tabular}{ccccccccccccccc}
\hline
\hline

 & \multicolumn{4}{c}{All pre-treatment }  & & \multicolumn{4}{c}{Half of the pre-treatment }   & & \multicolumn{4}{c}{Average of pre-treatment }  \\
&  \multicolumn{4}{c}{outcome lags} & & \multicolumn{4}{c}{outcome lags + covariates} & & \multicolumn{4}{c}{outcomes + covariates}  \\  \cline{2-5} \cline{7-10} \cline{12-15}
 
$J$  &4 & 12 & 40 & 100 & &4 & 12 & 40 & 100  & &4 & 12 & 40 & 100  \\

     & (1) &  (2) & (3) & (4) & & (5) & (6) & (7) & (8)  & & (9) & (10) & (11) & (12) \\ 
     
     \hline
     
\multicolumn{15}{c}{Panel A: $T_0=J + 5$} \\
\\
$\mathbb{E}[\hat \mu_{01}]$ & 0.732 & 0.814 & 0.885 & 0.925 &  & 0.731 & 0.811 & 0.889 & 0.927 &  & 0.675 & 0.686 & 0.659 & 0.673 \\
$se[\hat \mu_{01}]$ & 0.222 & 0.151 & 0.089 & 0.058 &  & 0.241 & 0.164 & 0.094 & 0.063 &  & 0.340 & 0.262 & 0.228 & 0.197 \\
 \\
$\mathbb{E}[\hat z_{01}]$ & 0.733 & 0.820 & 0.880 & 0.921 &  & 0.770 & 0.840 & 0.890 & 0.925 &  & 0.858 & 0.956 & 0.989 & 0.992 \\
$se[\hat z_{01}]$ & 0.200 & 0.148 & 0.090 & 0.055 &  & 0.202 & 0.147 & 0.091 & 0.060 &  & 0.188 & 0.122 & 0.048 & 0.038 \\
 \\
$se(\hat \alpha)$ & 1.408 & 1.275 & 1.132 & 1.063 &  & 1.430 & 1.277 & 1.142 & 1.070 &  & 1.496 & 1.353 & 1.222 & 1.184 \\

\\

\multicolumn{15}{c}{Panel B: $T_0=2 \times J$} \\
\\
$\mathbb{E}[\hat \mu_{01}]$ & 0.728 & 0.832 & 0.902 & 0.938 &  & 0.726 & 0.836 & 0.905 & 0.942 &  & 0.688 & 0.692 & 0.674 & 0.666 \\
$se[\hat \mu_{01}]$ & 0.219 & 0.126 & 0.069 & 0.039 &  & 0.230 & 0.131 & 0.073 & 0.042 &  & 0.342 & 0.264 & 0.231 & 0.192 \\
 \\
$\mathbb{E}[\hat z_{01}]$ & 0.738 & 0.827 & 0.908 & 0.938 &  & 0.772 & 0.840 & 0.912 & 0.941 &  & 0.865 & 0.962 & 0.986 & 0.995 \\
$se[\hat z_{01}]$ & 0.230 & 0.128 & 0.066 & 0.042 &  & 0.229 & 0.129 & 0.067 & 0.043 &  & 0.204 & 0.099 & 0.055 & 0.024 \\
 \\
$se(\hat \alpha)$ & 1.406 & 1.186 & 1.098 & 1.058 &  & 1.407 & 1.203 & 1.104 & 1.069 &  & 1.566 & 1.294 & 1.225 & 1.223 \\

\hline

\end{tabular}
   \end{lrbox}
\usebox{\tablebox}\\
\settowidth{\tableboxwidth}{\usebox{\tablebox}} \parbox{\tableboxwidth}{\footnotesize{Notes: this table presents the expected value and the standard error of the estimators for $\mu_{01}$ and $z_{01}$ using the specification that includes all pre-treatment outcomes lags as predictors (columns 1 to 4), the first half of the pre-treatment outcome lags and the covariates as predictors (columns 5 to 8), and the average of the pre-treatment outcomes and the covariates as predictors (columns 9 to 12).  It also presents the standard error of $\hat \alpha$. Since $\mathbb{E}[{\lambda}_t]=0$, $\mathbb{E}[\hat \alpha_{01}] =0$, which is the true treatment effect. Panel A presents results with  $T_0=J + 5$, while Panel B presents results with  $T_0=2\times J$. Results based on 500 simulations. The DGP is described in detail in Appendix Section \ref{Appendix_covariates}.

}
}
\end{table}
\end{landscape}

\end{document}